\documentclass[12pt]{article}
\usepackage{bbm}
\usepackage{mathrsfs}
\usepackage{amsfonts}
\usepackage{graphicx}
\usepackage{amsmath}
\usepackage{amssymb}
\usepackage{dsfont}
\usepackage{float}
\usepackage{natbib}
\usepackage{subcaption} 
\usepackage{rotating,color}
\usepackage{xcolor}
\usepackage{booktabs}
\newcommand{\be}{\begin{equation}}
\newcommand{\ee}{\end{equation}}
\newcommand{\beaa}{\begin{eqnarray*}}
\newcommand{\eeaa}{\end{eqnarray*}}
\newcommand{\bea}{\begin{eqnarray}}
\newcommand{\eea}{\end{eqnarray}}

\newtheorem{theorem}{ \noindent T{\footnotesize HEOREM}}
\newtheorem{prop}{ \noindent P{\footnotesize ROPOSITION}}[section]
\newtheorem{assumption}{ \noindent A{\footnotesize SSUMPTION}}

\newtheorem{lemma}{ \noindent L{\footnotesize EMMA}}[section]

\newtheorem{remark}{ \noindent R{\footnotesize EMARK}}[section]

\def\var{\mathrm {var}}

\def\cov{\mathrm {cov}}

\def\log{\mathrm {log}}
\newcommand{\bm}{\boldsymbol}

\def\tr{\mathrm {tr}}
\def\min{\mathrm {min}}

\def\B{{\bf B}}
\def\z{{\bm z}}

\def\R{{\bf R}}

\def\bmu{{\bm \mu}}

\def\X{{\bm X}}

\def\w{{\bm w}}

\def\bU{\boldsymbol{U}}\def\bX{\boldsymbol{X}}
\def\bZ{\boldsymbol{Z}}

\newcommand{\bGam}{\boldsymbol{\Gamma}}
\newcommand{\bOme}{\boldsymbol{\Omega}}

\def\O{{\bf \Omega}}
\def\bms{{\bm\Sigma}}

\def\eeta{{\bm\eta}}

\def\bdelta{{\bm\delta}}
\def\bvarepsilon{{\bm\varepsilon}}
\def\bxi{{\bm\xi}}
\def\bzeta{{\bm\zeta}}

\def\tr{\mathrm {tr}}

\def\bepsilon{\bm \epsilon}

\def\T{{\!\top\!}}
\def\tr{\mathrm{tr}}

\def\sup{\mathrm{sup}}
\def\inf{\mathrm{inf}}

\newcommand{\ind}[1]{\mathbb{I}(#1)}
\def\E{\mathbb{E}}
\def\pr{\mathrm{P}}

\allowdisplaybreaks[4]

\title{\bf Adaptive Change Point Inference for High Dimensional Time Series with Temporal Dependence }
\author{Xiaoyi Wang$^1$, Jixuan Liu$^2$ and Long Feng$^2$\\
Bejing Normal University$^1$ and Nankai University$^2$}
\date{}

\begin{document}

\def\spacingset#1{\renewcommand{\baselinestretch}%
{#1}\small\normalsize} \spacingset{1}

\maketitle

\begin{abstract}
This paper investigates change point inference in high-dimensional time series. We begin by introducing a max-$L_2$-norm based test procedure, which demonstrates strong performance under dense alternatives. We then establish the asymptotic independence between our proposed statistic and the two max-$L_\infty$-based statistics introduced by \cite{wang2023JRSSB}. Building on this result, we develop an adaptive inference approach by applying the Cauchy combination method to integrate these tests. This combined procedure exhibits robust performance across varying levels of sparsity. Extensive simulation studies and real data analysis further confirm the superior effectiveness of our proposed methods in the high-dimensional setting.

\end{abstract}

\newpage
\spacingset{1.8} 
\section{Introduction}

Change point inference plays a critical role in the analysis of high-dimensional data, where the structure or behavior of a system may shift over time due to external shocks, policy changes, or underlying regime switches. In modern applications—such as genomics, finance, neuroimaging, and environmental monitoring—data are often collected across a large number of variables simultaneously, making traditional low-dimensional change point detection methods inadequate. High-dimensional change point inference enables the detection and localization of structural breaks in complex systems, allowing for timely responses, better understanding of dynamic processes, and more accurate predictive modeling \citep{10.1214/18-ejs1442}. Given the growing prevalence of high-dimensional data in scientific and industrial contexts, developing robust and adaptive inference procedures that can handle varying degrees of sparsity and dependence is both practically important and theoretically challenging.

Consider a sequence of $p$-dimensional random vectors of sample size $n$, $\{\bX_i=(X_{i1},\ldots,$\\$X_{ip})^\T\}_{i=1}^n$, from the mean-change model,
\begin{align*}
	\bX_i = \bmu_0 + \bdelta \ind{i>\tau} + \bepsilon_i,~~~~\ i=1,\ldots,n,
\end{align*}
where $\bmu_0\in\mathbb{R}^p$ is the baseline mean level, $\bdelta\in\mathbb{R}^p$ is the mean-change signal parameter,
$\tau\in\{1,\ldots,n\}$ is the possible changepoint, and $\{\bepsilon_i=(\epsilon_{i1},\ldots,\epsilon_{ip})^\T\}_{i=1}^n$ are random noises with zero mean.
Of interest is to test whether there exists a change point, that is,
\begin{align*}
\begin{gathered}
	H_0: \tau=n\ \text{and}\ \bdelta=0\ \text{versus}\\
	H_1:\ \text{there exists}\ \tau\in\{1,\ldots,n-1\}\ \text{and}\ \bdelta\neq 0,
\end{gathered}
\end{align*}
under the paradigm that both the sample size $n$ and dimension $p$ grow to infinity.

Many existing studies focus on constructing test statistics by taking the maximum or the sum over the sample size or across dimensions of the corresponding individual CUSUM-type statistics. Specially, the individual cumulative sum (CUSUM) statistics $\{C_{\gamma, j}(k),\ j=1,\ldots,p\}_{k=1}^{n-1}$, with $\gamma=0$ or $0.5$ frequently used, is defined as
\begin{align}\label{CUSUM1}
    C_{\gamma, j}(k) = \left\{\frac{k}{n}\left(1-\frac{k}{n}\right)\right\}^{-\gamma}\frac{1}{\sqrt{n}}\left(S_{kj}-\frac{k}{n}S_{nj}\right)/\widehat{\sigma}_{j},
\end{align}
$S_{kj}=\sum_{i=1}^{k}X_{ij}$ and $\widehat{\sigma}_{j}$'s are estimators of component-wise (long-run) variances. Among them, \cite{10.1016/j.jeconom.2009.10.020}, \cite{10.1111/j.1467-9892.2012.00796.x}, and \cite{10.1007/s11425-016-0058-5} proposed test statistics of the form $\max_{k=1,\ldots,n-1}\sum_{j=1}^p C_{0, j}^2(k)$, where the maximum is taken over time points and the summation is taken across dimensions. We refer to such statistics as max-$L_2$-type test statistics. \cite{10.1214/17-aos1610} and \cite{10.1142/s201032631950014x} introduced sum-$L_2$-type statistics of the form $\sum_{k=1}^{n-1}\sum_{j=1}^p C_{0.5, j}^2(k)$, where the summation is taken over both time points and dimensions. These statistics are generally powerful under dense alternatives, in which many components of the mean-change vector $\bm \delta$ are nonzero but of small magnitude. However, their power diminishes under sparse alternatives, where only a few entries of $\bm \delta$ are nonzero but have large magnitudes. To address this issue, several works have proposed statistics based on maxima over the dimension. For example, \cite{10.1214/15-aos1347} proposed a max-$L_\infty$-type test statistic $\max_{k=1,\ldots,n-1}\max_{j=1,\ldots,p}|C_{0, j}(k)|$ with $\gamma = 0$ in \eqref{CUSUM1}, and demonstrated that, under $H_0$ and suitable normalization, this statistic converges to a Gumbel-type extreme value distribution. Furthermore, \cite{10.1111/rssb.12406} considered a similar max-$L_\infty$-type statistic, defined as $\max_{\lambda \leq k \leq n - \lambda}\max_{j=1,\ldots,p}|C_{0.5, j}(k)|$, where $\lambda \in [1, n/2]$ is a user-specified trimming parameter and $\gamma = 0.5$ in \eqref{CUSUM1}. Recently, \cite{wang2023JRSSB} established the asymptotic null distributions of both max-$L_\infty$-type test statistics under weaker and more general conditions.

In practical scenarios, the sparsity structure of the alternative hypothesis---whether dense or sparse---is typically unknown. To address this challenge, a number of studies have proposed adaptive testing procedures that are capable of maintaining good power across varying sparsity levels. For example, \cite{10.1111/rssb.12375} introduced an adaptive strategy based on adjusted max-$L_q$ aggregation. This method computes the test statistics for a set of values $q \in \{1, 2, 3, 4, 5, \infty\}$ and then selects the minimum $p$-value among them. This approach leverages the fact that larger $q$ values tend to be more powerful for sparse alternatives, whereas smaller $q$ values are better suited for dense alternatives. Similarly, \cite{doi:10.1080/01621459.2021.1884562} proposed a family of self-normalized U-statistics corresponding to each $L_q$ norm, providing a unified framework for adaptivity. More recently, \cite{wang2023JRSSB} established the asymptotic independence between the max-$L_\infty$-type statistics proposed in \cite{10.1214/15-aos1347,10.1111/rssb.12406} and the sum-$L_2$-type statistic in \cite{10.1142/s201032631950014x}. Building on this result, they developed a computationally efficient Fisher-type combination test that effectively adapts to different sparsity regimes. This combined procedure demonstrates strong performance across a wide range of alternatives. For a comprehensive review of recent advances in high-dimensional change-point inference, including these adaptive strategies, see \cite{10.1016/j.jmva.2021.104833}, \cite{wang2023JRSSB}, \cite{meng2024covariate}.

Most of the aforementioned methods assume that the observed time series \(\bX_i\) are independent. However, in practice, high-dimensional time series often exhibit some weak temporal dependence. In such cases, many of these methods fail to work effectively and may suffer from large size distortions. To address this issue, some studies have focused on high-dimensional change point inference in the presence of temporal dependence. For example, \cite{li2019change} proposed a sum-\(L_2\)-type test procedure that aggregates a bias-corrected version of the classic CUSUM test statistic. Meanwhile, \cite{wang2022inference} developed a trimmed version of the self-normalized U-statistic, which excludes pairs of points that are close in time. Both of these methods are sum-type test statistics over dimensions and perform well for dense alternatives. However, for sparse alternatives, it is worth noting that the two max-\(L_\infty\)-type test statistics \citep{10.1214/15-aos1347,10.1111/rssb.12406} can also be applied to high-dimensional temporal dependence data, provided that the variance estimator is replaced by a long-run variance estimator. To the best of our knowledge, there currently does not exist an adaptive test procedure that can handle a wide range of sparsity levels for alternatives. Therefore, in this paper, we introduce an adaptive change point inference method for high-dimensional time series with temporal dependence.

First, we propose a max-$L_2$-type test statistic along with a bias correction procedure to mitigate the effects of temporal dependence. We establish the limiting null distribution of the proposed test statistic. Furthermore, we demonstrate that the max-$L_2$-type test statistic is asymptotically independent of two existing max-$L_\infty$-type test statistics \citep{10.1214/15-aos1347,10.1111/rssb.12406}. Based on this result, we develop two Cauchy combination test procedures that combine the max-$L_2$-type statistic with the max-$L_\infty$-type statistics. Unlike the asymptotic independence between the Gumbel and normal distributions shown in \cite{wang2023JRSSB}, our work establishes the asymptotic independence between a Gumbel distribution and the maximum of a Gaussian process—representing a non-trivial theoretical contribution. Finally, we propose an estimator for the change-point location based on the adaptive Cauchy combination $p$-values and show its consistency under mild regularity conditions. Simulation studies and real data applications demonstrate that the proposed methods perform well in high-dimensional time series settings with temporal dependence. The Cauchy combination tests are particularly effective across a wide range of sparsity levels in the alternatives. This paper makes three key contributions:
\begin{itemize}
\item[1.] Development of a max-$L_2$-type test statistic: We propose a novel max-$L_2$-type test statistic for change-point inference in high-dimensional time series data with temporal dependence. Compared to the traditional sum-$L_2$-type test statistic, our proposed method exhibits significantly improved power, particularly under dense alternative scenarios.

\item[2.] New theoretical framework for asymptotic independence: We introduce a new analytical approach to establish the asymptotic independence between a Gumbel distribution and the maximum of a Gaussian process. This result is distinct from existing techniques, which primarily focus on proving the asymptotic independence between normal and Gumbel distributions. Our contribution broadens the theoretical understanding of extreme value theory in dependent high-dimensional settings.

\item[3.] Robust Cauchy combination testing procedure: Based on the established asymptotic independence, we construct two Cauchy combination test procedures that integrate the max-$L_2$-type and max-$L_\infty$-type test statistics. These procedures demonstrate strong and robust performance across a wide range of sparsity levels in the alternatives, making them highly suitable for practical applications involving high-dimensional, temporally dependent data.
\end{itemize}

The remainder of the paper is organized as follows. Section~\ref{sec:2-L2} introduces the proposed max-$L_2$-type test procedure. Section~\ref{sec:3-adaptive} establishes the asymptotic independence between the max-$L_2$-type and max-$L_\infty$-type test statistics and proposes two Cauchy combination test procedures. Section~\ref{sec:4-location} focuses on the estimation of the change-point location. Section~\ref{sec:5-simulation} presents simulation studies, and Section~\ref{sec:6-realdata} provides a real data application. All technical proofs are given in the Appendix.

Throughout this paper, we use $\gtrsim, \lesssim, (\sim)$ to denote (two sided) inequalities involving a multiplicative constant. For $a\in\mathbb{R}$, $\lfloor a \rfloor$ and $\lceil a \rceil$ denote the lower and upper integer part of $a$. For $a, b\in\mathbb{R}$, we write $a\wedge b=\min(a,b)$. For a vector $\mathbf{a} = (a_1,\dots,a_p)^{\T}\in\mathbb{R}^p$, $\|\mathbf{a}\|^2 = \sum\nolimits_{j=1}^p a_j^2$. For a matrix $\mathbf{A}$, $\lambda_{\min}(\mathbf{A})$ and $\lambda_{\max}(\mathbf{A})$ denote the minimum and maximum eigenvalue of matrix $\mathbf{A}, $and $\tr(\mathbf{A})$ denotes the trace of $\mathbf{A}$. For a set $\mathcal{A}$, we denote $|\mathcal{A}|$ by its cardinality, and by $\mathcal{A}^c$ its complementary. $\ind{h=0}$ denotes the indicator function.

\section{max-$L_2$-type test}\label{sec:2-L2}
Before giving the max-$L_2$-type test statistic, for any $k\in\{1,\dots,n\}$, we propose the CUSUM statistic
\begin{align*}
W(k)=&\frac{1}{n\sqrt{p}}\left(\sum_{i=1}^k \bX_i-\frac{k}{n}\sum_{i=1}^n \bX_i\right)^\top\left(\sum_{i=1}^k \bX_i-\frac{k}{n}\sum_{i=1}^n \bX_i\right).
\end{align*}
It can be showed that the expectation of $W(k)$ under $H_0$ as follows,
\begin{align*}
\mu_{k}
=\frac{k^2(n-k)^2}{n^3\sqrt{p}}\sum\limits_{h=0}^{n-1}\sum\limits_{i=1}^{n-h}\{2-\ind{h=0}\}a_{i,k}a_{i+h,k}\tr\{\bGam(h)\},
\end{align*}
where
\begin{eqnarray*}
a_{i,k}=\left\{
\begin{array}{lr}
k^{-1}, &  i\leq k, \\
-(n-k)^{-1}, &   i> k.
\end{array}
\right.
\end{eqnarray*}
We use
\begin{align*}
\mu_{M,k}=\frac{k^2(n-k)^2}{n^3\sqrt{p}}\sum\limits_{h=0}^{M}\sum\limits_{i=1}^{n-h}{\{2-\ind{h=0}\}}a_{i,k}a_{i+h,k}\tr\{\bGam(h)\},
\end{align*}
replace $\mu_{k}$, where $M=\lceil (n\wedge p)^{1/8}\rceil$. We have proved that the difference between $\mu_{k}$ and $\mu_{M,k}$ is negligible under the following Assumptions~\ref{ass:C1}--\ref{ass:C3} in the Appendix.
In order to show the asymptotic distribution of $W(k)$, we impose the following assumptions.
\begin{assumption}\label{ass:C1}
Suppose the noises $\{\bepsilon_i\}_{i=1}^n$ follow the linear structure
\begin{eqnarray*}
\bepsilon_i=\bms^{1/2}\sum\limits_{\ell=0}^{\infty}b_{\ell}\bZ_{i-\ell}.
\end{eqnarray*}
 \begin{itemize}
\item [(i)] $\bZ_i=(Z_{i1},\ldots,Z_{ip})^\T$, and $Z_{ij}$ are independent and identically distributed random variables with $\E(Z_{ij})=0,  \E(Z_{ij}^2)=1, \E(Z_{ij}^4)<+\infty$.
\item [(ii)] $\sum\nolimits_{\ell=0}^{\infty}|b_{\ell}|<\infty$, $\sum\nolimits_{\ell=0}^{\infty}b_{\ell}=s\neq 0$, and $\lim\limits_{\ell\rightarrow\infty}\ell^5b_{\ell} = 0$.
\end{itemize}
\end{assumption}
\begin{assumption}\label{ass:C3}
    There exist constants $C_0$ and $C_1$ such that $0\leq C_0<\lambda_{\min}(\bms)\leq \lambda_{\max}(\bms) \leq C_1<\infty$.
\end{assumption}

 Assumption~\ref{ass:C1} assumes the linear process model of $\{\bX_i\}_{i=1}^n$, which is widely used in time series analysis \citep{zhang2018clt}. This linear process model cover the well known MA($q$) model and AR($1$) model. Let $\bGam(h)=\cov(\bX_i, \bX_{i+h})$ for $h=0,1,2,\ldots,n-1$, and $\bOme_n=n\cov(\bar{\bX}_n) = \bGam(0)+2\sum\nolimits_{h=1}^{n-1}(1-n^{-1}h)\bGam(h)$. With Assumption~\ref{ass:C1}, $\E(\bX_i)=\bmu_0 + \bdelta \ind{i>\tau}$ and $\bGam(h)=c_h\bms$ for $h=0,1,2,\dots$ and $\bGam(h)=\bGam(-h)^\top$ with $c_h=\sum\nolimits_{\ell=0}^{\infty}b_{\ell}b_{\ell+h}$, and the long-run covariance matrix $\bOme = \bGam(0)+2\sum\nolimits_{h=1}^{\infty}\bGam(h)$.
Assumption~\ref{ass:C1}-(ii) guarantees that the dependence of the current error on past shocks decays. It rules out processes with unit roots or long memory where correlations decay too slowly. Assumption~\ref{ass:C3} ensures that $\tr(\bms)=O(p)$ and $\tr(\bms^2)=O(p)$, which  primarily guarantees the asymptotic independence of the test statistics. A similar assumption also appears in \cite{wang2023JRSSB}.

 \begin{theorem}
 \label{Th1}
Under Assumptions~\ref{ass:C1}--\ref{ass:C3} and $H_0$, if $M=\lceil(n\wedge p)^{1/8}\rceil$, we have
$$
W(\lfloor n t\rfloor)-\mu_{M,\lfloor n t\rfloor} \stackrel{d}{\rightarrow} \omega V(t), \quad t \in[0,1],
$$
as $(n,p)\rightarrow\infty$, where $\omega = \lim_{p\rightarrow\infty}\sqrt{2\tr(\bOme^2)/p}$.
$V(t)$ is a continuous Gaussian process with $\E\{ V(t)\}=0$ and
$$
\E\{V(t) V(s)\}=(1-t)^{2} s^{2}, \quad 0 \leqslant s \leqslant t \leqslant 1.
$$
\end{theorem}

In practice, it is essential to estimate $\mu_{M,\lfloor n t\rfloor}$ and $\omega$ in Theorem 1. Define $\breve{\bX}_{f,h,g,k} = (\bX_{f}-\bX_{f+M+h+1})^{\T}(\bX_{g}-\bX_{g+M+k+1})$.
Inspired by the procedure of moving ranges of neighboring samples proposed in \cite{10.1142/s201032631950014x}, we propose a consistent estimator of $\mu_{M,\lfloor n t\rfloor}$ as follows,
\begin{eqnarray}
\label{mu-hat}
\hat{\mu}_{M,\lfloor n t\rfloor}=\frac{\lfloor n t\rfloor^2(n-\lfloor n t\rfloor)^2}{n^3\sqrt{p}}\sum\limits_{h=0}^{M}\sum\limits_{i=1}^{n-h}{\{2-\mathbb{I}(h=0)\}}a_{i,\lfloor n t\rfloor}a_{i+h,\lfloor n t\rfloor}\widehat{\tr\{\bGam(h)\}},
\end{eqnarray}
where
\begin{align*}
\widehat{\tr\{\bGam(h)\}}=\frac{1}{2n}\sum\limits_{t=1}^{n-M-2h-1}\breve{\bX}_{t+h,h,t,h}.
\end{align*}
Motivated by \cite{zhang2018clt}, we consider the following estimator for $\tr\{\bGam(h)\bGam(k)\}$,
\begin{align*}
\widehat{\tr\{\bGam(h)\bGam(k)\}}=
\frac{\sum\limits_{t=1}^{[n/2]-M-2k-1}\sum\limits_{s=t+[n/2]}^{n-M-2k-1}
\breve{\bX}_{t,h,s,k}\breve{\bX}_{t+h,h,s+k,k}}{4(n-k-3/2[n/2]-M/2)([n/2]-M-2k-1)}.
\end{align*}
Then, we propose the ratio-consistent estimator of $\omega$,
\begin{align}
\label{omega-hat}
\hat{\omega}=\bigg\{\frac{2}{p}\bigg(\widehat{\tr\{\bGam(0)\bGam(0)\}}+2\sum\limits_{h=1}^M \widehat{\tr\{\bGam(h)\bGam(0)\}}+2\sum\limits_{k=1}^M \widehat{\tr\{\bGam(0)\bGam(h)\}}
+4\sum\limits_{h,k=1}^M \widehat{\tr\{\bGam(h)\bGam(k)\}}\bigg)\bigg\}^{1/2}.
\end{align}
The following Theorem~\ref{Th2} shows the consistency of the estimators.

\begin{theorem} \label{Th2}
Under Assumptions~\ref{ass:C1}--\ref{ass:C3}, if $M=\lceil(n\wedge p)^{1/8}\rceil$, $p=o(n^{7/4})$ and $\|\bdelta\|^2=o(n p^{1/2}/M^2)$, we have
$\hat{\mu}_{M,\lfloor n t\rfloor}-\mu_{M,\lfloor n t\rfloor} = o_p(\omega)$,
for any $t \in[0,1]$, and
$\hat{\omega}/\omega\stackrel{p}{\rightarrow} 1$, as $(n,p)\rightarrow\infty$.
\end{theorem}

Based on Theorem \ref{Th1}--\ref{Th2}, we finally propose the max-$L_2$-type test statistic,
\begin{align*}
S_{n,p}=\max_{t\in [0,1]}\left\{W(\lfloor n t\rfloor)-\hat{\mu}_{M,\lfloor n t\rfloor}\right\},
\end{align*}
and the following Theorem \ref{Tms-1} shows that the max-$L_2$-type test statistic $S_{n,p}$ converges to a Gaussian process.
\begin{theorem}
\label{Tms-1}
Under Assumptions~\ref{ass:C1}--\ref{ass:C3} and $H_0$, if $M=\lceil(n\wedge p)^{1/8}\rceil$ and $p=o(n^{7/4})$, we have
  \begin{align*}
  S_{n,p}/\hat{\omega} \stackrel{d}{\rightarrow}  \max_{t\in [0,1]} V(t),
  \end{align*}
   as $(n,p)\rightarrow\infty$, where $V(t)$ is a continuous Gaussian process with $\E \{V(t)\}=0$ and
$$
\E\{V(t) V(s)\}=(1-t)^{2} s^{2}, \quad 0 \leqslant s \leqslant t \leqslant 1.
$$
\end{theorem}

The $p$-value of the test based on $S_{n,p}$ is given by
\begin{align}
    \notag {p}_{S_{n,p}}=1-F_{V}\left(S_{n,p}/\hat{\omega}\right),
\end{align}
where $F_{V}(\cdot)$ is the cumulative distribution function (CDF) of $\max_{0\leq t\leq 1} V(t)$. If the $p$-value is smaller than the pre-specified significance level $\alpha\in (0,1)$,  we reject the null hypothesis, that is, there are no changepoints in the data sequence.
\begin{remark}
\label{rem-1}
Following Equation (15) in \cite{konakov2020extremes}, there exist a constant $c$, such that
\begin{align*}
\pr\left(\max_{t\in [0,1]} V(t)\ge u\right)=\frac{c}{u}e^{-8u^2}\{1+o(1)\}.
\end{align*}
The constant $c$ could be estimated via Monte Carlo simulation. Specifically, let $T=\{t_i = i/T_d, i=1,\dots,T_d\}$ and the number of replications $B$. For $b=1,\dots,B$, let $v_{b} = \max_{t\in T} V_b(t)$, where $\left(V_b(1),\dots,V_b(T_d)\right)^{\T}$ is sample from the $T_d$-dimensional multivariate normal distribution with mean zero and covariance matrix with the $(k,\ell)$th element given by $(1-t_{k})^2t_{\ell}^2$ for $1\leq k \leq \ell \leq T_d$. Naturally, the estimated CDF of $\max_{0\leq t\leq 1} V(t)$ can be constructed with $\{v_b\}_{b=1}^B$, and the estimation $\hat{c}$ will be obtained with a given nominal significant level $\alpha$.
Thus, the $p$-value can be alternatively calculated as
\begin{align*}
p_{S_{n,p}}=\frac{\hat{\omega}\hat{c}}{S_{n,p}}e^{-8 S_{n,p}^2/\hat{\omega}^2}.
\end{align*}
\end{remark}

Similar to the approach of \cite{10.1007/s11425-016-0058-5}, the test based on $S_{n,p}$ can be expected that $\max$-$L_{2}$-type test performs well in detecting small but dense signals. The following proposition established the test consistency under $H_1$.

 \begin{prop}\label{prop:Tms-alter}
     Under the assumptions in Theorem~\ref{Th2} and $H_1$, if $\tau\sim n$ and $\|\bdelta\|^2\gtrsim p^{1/2}n^{-1}$, then we have, the test based on $S_{n,p}$ is consistent, as $(n,p)\rightarrow\infty$.
 \end{prop}

\section{Adaptive tests}\label{sec:3-adaptive}

We first restate two versions of $\max$-$L_\infty$-type statistics that were adopted in the literature, to wit,
\begin{align*}\label{M}
	M_{n,p}:=\max_{k=1,\ldots,n-1}\max_{j=1,\ldots,p}\vert C_{0, j}(k)\vert\ \text{and}\
	M^{\dagger}_{n,p}:=\max_{\lambda_n\leq k\leq n-\lambda_n}\max_{j=1,\ldots,p}\vert C_{0.5, j}(k)\vert,
\end{align*}
respectively, where we recall that the CUSUM statistics $C_{\gamma, j}(k)$'s with $\gamma=0$ or 0.5 are defined in \eqref{CUSUM1}, and $\lambda_n\in[1,n/2]$ is a pre-specified boundary removal parameter.
We also restate Theorem 1 in \cite{wang2023JRSSB} with some mild modifications.

\begin{theorem}\label{null:Max}
Suppose $H_0$ and Assumptions 1--2 in \cite{wang2023JRSSB} hold and the component-wise (long-run) estimators $\widehat{\sigma}_j$ in \eqref{CUSUM1} satisfy the uniform convergence $\max_{1\leq j\leq p}\vert\widehat{\sigma}_j-\sigma_j\vert=o(1)$. If $p\lesssim n^{\nu}$ for some $0<\nu<q/2-2$, we have
\begin{itemize}
\item[(i)] As $(n,p)\to\infty$,
\[
    \pr\left(M_{n,p}\leq u_p\{\exp(-x)\}\right) \to \exp\{-\exp(-x)\},
\]
where $u_p\{\exp(-x)\}=\sqrt{\{x+\log(2p)\}/2}$.

\item[(ii)] If $\lambda_n\sim n^{\lambda}$ for some $\lambda\in(0,1)$, then, as $(n,p)\to\infty$,
\[
    \pr\left(M^{\dagger}_{n,p}\leq \frac{x+D(p\log h_n)}{A(p\log h_n)}\right) \to \exp\{-\exp(-x)\},
\]
where $A(x)=\sqrt{2\log x}$, $D(x)=2\log x+2^{-1}\log\log x-2^{-1}\log\pi$ and $h_n=\left\{(\lambda_n/n)^{-1}-1\right\}^2$.
\end{itemize}
\end{theorem}

{For the variance estimator $\hat{\sigma}_j$, $j \in \{1, \ldots, p\}$, we follow the difference-based approach proposed by \cite{chan2022AoS}, which  may provides a more efficient variance estimation method under change-point scenarios\citep{chan2022AoS}. Specifically, we employ the third-order difference-based long-run variance estimator applied componentwise, with the bandwidth optimally selected to minimize the mean squared error; see $\hat{v}_{(3)}$ and Equation (6.3) in \cite{chan2022AoS}. If we further assume that $\E (Z_{ij}^{4+r})<\infty$ for some $r>0$, where $Z_{ij}$ is defined in Assumption~\ref{ass:C1} for $i\in\{1,\ldots,n\}$ and $j\in\{1,\ldots,p\}$, or alternatively impose Assumption~\ref{ass:C4} directly, under additional assumptions on the parameters and kernel function used in the statistic, the desired consistency hold. For details, see the Lemma~\ref{lemma:LRV} in Appendix.}

By Theorem \ref{null:Max}, the $p$-values associated with $M_{n,p}$ and $M^{\dagger}_{n,p}$ are
\begin{align*}
	p_{M_{n,p}} &= 1- G\big(2M_{n,p}^2-\log(2p)\big),\\
        p_{M^{\dagger}_{n,p}} &= 1- G\big(A(p \, \log \, \lambda^{\dagger}_n) M^{\dagger}_{n,p}-D(p \, \log \, \lambda^{\dagger}_n)\big),
\end{align*}
where $\lambda^{\dagger}_n = \{(\lambda_n/n)^{-1}-1\}^2$, $A(x) = \sqrt{2\, \log \, x}$, $D(x) = 2\, \log \, x + 2^{-1} \log\,\log \, x -2 ^{-1} \log\, \pi$ and $G(x) = \exp\{-\exp(-x)\}$.

In practice, the sparsity level of the potential signal is always unknown. To accommodate both sparse and dense alternatives, we adopt the $\max\text{-}L_\infty$-type statistic proposed by \citet{wang2023JRSSB} for detecting sparse signals, and the $\max\text{-}L_2$-type statistic for identifying dense signals. We then integrate the $\max\text{-}L_\infty$-type and $\max\text{-}L_2$-type methods to construct a more powerful testing procedure. Following \cite{wang2023JRSSB}, we introduce the following additional assumptions to highlight the key steps and streamline the proof.

\begin{assumption}\label{ass:C4}
    The components of $\bZ_i$ are independent sub-Gaussian variables, that is, there exists a constant $\kappa_0$, such that $\sup_{\ell\geq 1}\ell^{-1/2}\{\E (Z_{ij}^\ell)\}^{1/\ell}\leq \kappa_0$, $i=1,\ldots,n,j=1,\ldots,p$.
\end{assumption}
The following theorem establishes the asymptotic independence between the $\max\text{-}L_\infty$-type and $\max\text{-}L_2$-type statistics under the null hypothesis $H_0$.
\begin{theorem}\label{indnull}
  Suppose $H_0$ and Assumptions~\ref{ass:C1}--\ref{ass:C4} and Assumptions 1--2 in \cite{wang2023JRSSB} hold. If $\log n = o(p^{1/4})$, $p=o(n^{7/4})$, we have
  \begin{itemize}
  \item [(i)] As $(n,p)\rightarrow \infty$, $S_{n,p}$ is asymptotically independent of $M_{n,p}$, that is
  \begin{align*}
  \pr\left(S_{n,p}\le \hat{\omega} x, ~M_{n,p}\le u_p\{\exp(-y)\}\right)\rightarrow F_V(x)\cdot\exp{\{-\exp(-x)\}};
  \end{align*}
 \item [(ii)] If $\lambda_n\sim n^{\lambda}$ for some $\lambda\in(0,1)$, then, as $(n,p)\to\infty$, $S_{n,p}$ is asymptotically independent of $M^{\dagger}_{n,p}$, that is
 \begin{align*}
  \pr\left(S_{n,p}\le \hat{\omega} x, ~M^{\dagger}_{n,p}\le \nu_{p,n}(y)\right)\rightarrow F_V(x)\cdot\exp{\{-\exp(-x)\}},
  \end{align*}
where $\nu_{p,n}(y) = \{y+D(p\log h_n)\}/A(p\log h_n)$.
 \end{itemize}
\end{theorem}

Following the independence between $S_{n,p}$ and $M_{n,p}$ or $M^{\dagger}_{n,p}$ induced by Theorem \ref{indnull}, the Minimum combination, Fisher combination, and Cauchy combination can be used to combine the max-$L_2$-type test  and max-$L_\infty$-type test, see \cite{wishart1932methods, littell1971asymptotic, liu2020cauchy}. \cite{long2023cauchy} have showed that the power of the test based on Cauchy combination would be more powerful than the Minimum combination test. Attracted by the stability of the Cauchy combination method, in this paper, we combine $S_{n,p}$ with $M_{n,p}$ and $M^{\dagger}_{n,p}$ using the Cauchy combination test procedure in \cite{liu2020cauchy}, that is,
\begin{align}
T_{CC}=0.5\tan\{(0.5-p_{S_{n,p}})\pi\}+ 0.5\tan\{(0.5-p_{M_{n,p}})\pi\}, \label{cc-1}\\
T_{CC}^\dagger=0.5\tan\{(0.5-p_{S_{n,p}})\pi\}+ 0.5\tan\{(0.5-p_{M^\dagger_{n,p}})\pi\}. \label{cc-2}
\end{align}
Since $T_{CC}$ and $T_{CC}^\dagger$ converge to a standard Cauchy distribution under $H_0$,
\begin{align}
p_{CC}=1-F(T_{CC}), \quad p_{CC}^\dagger=1-F(T_{CC}^\dagger),
\end{align}
where $F(x)$ is the cumulative distribution function of standard Cauchy distribution $C(0,1)$. Both $p_{CC}$ and $p_{CC}^\dagger$ can be used as the final $p$-values for testing the null hypothesis $H_0$. If the combined $p$-value is smaller than the pre-specified significance level $\alpha\in (0,1)$,  we reject the null hypothesis.

We also give the asymptotic independence under the local alternative hypothesis, to wit,
\begin{align}
\label{H:indalter}
\begin{gathered}
	H_{1;n,p}:\ \ \vert\mathcal{A}\vert = o\{p/(\log n)^2\} \ \text{and}\ \|\bdelta\|^2 = o\{n^{-1}\tr(\O_n^2)\},
\end{gathered}
\end{align}
where $\mathcal{A} = \{1\leq j \leq p: \delta_j\neq 0\}$ is the support of $\bdelta$.
\begin{theorem}\label{indalter}
 Suppose Assumptions~\ref{ass:C1}--\ref{ass:C4}, Assumptions 1--2 in \cite{wang2023JRSSB} and $\max_{1\leq j\leq p}\vert\widehat{\sigma}_j-\sigma_j\vert=o(1)$ hold. Under $H_{1;n,p}$, if $p=o(n^{7/4})$, we have
\begin{itemize}
  \item [(i)] As $(n,p)\rightarrow \infty$, $S_{n,p}$ is asymptotically independent of $M_{n,p}$, that is
  \begin{align*}
    \pr\left(S_{n,p}\le \hat{\omega} x, ~M_{n,p}\le u_p\{\exp(-y)\}\right)\rightarrow \pr\left(S_{n,p}\le \hat{\omega} x\right)\pr\left(M_{n,p}\le u_p\{\exp(-y)\}\right);
  \end{align*}
 \item [(ii)] If $\lambda_n\sim n^{\lambda}$ for some $\lambda\in(0,1)$, then, as $(n,p)\to\infty$, $S_{n,p}$ is asymptotically independent of $M^{\dagger}_{n,p}$, that is
  \begin{align*}
  \pr\left(S_{n,p}\le \hat{\omega} x, ~M^{\dagger}_{n,p}\le \nu_{p,n}(y)\right)\rightarrow  \pr\left(S_{n,p}\le \hat{\omega} x\right) \pr\left( M^{\dagger}_{n,p}\le \nu_{p,n}(y)\right),
  \end{align*}
where $\nu_{p,n}(y) = \{y+D(p\log h_n)\}/A(p\log h_n)$.
 \end{itemize}
\end{theorem}

Theorem \ref{indalter} enables a finer analysis of the power profiles of the proposed adaptive tests, moving beyond the guarantee of consistency. A key question that arises is the power performance of our adaptive procedure relative to non-adaptive approaches that rely on a single test statistic (either the max-$L_2$-type test or max-$L_\infty$-type test).
We compare the power of these two adaptive tests to their non-adaptive counterparts.
For a given significant level $\alpha\in(0,1)$, let $\beta_{S,\alpha}$, $\beta_{M,\alpha}$ and $\beta_{CC,\alpha}$ denote the power functions of the test $S_{n,p}$, $M_{n,p}$ and the corresponding adaptive test. That is
\begin{align*}
\beta_{CC,\alpha} \geq \beta_{M\wedge S,\alpha} =& \pr(\min\{p_{S_{n,p}}, p_{M_{n,p}}\}\leq 1-\sqrt{1-\alpha})\\
\geq& \pr(\min\{p_{S_{n,p}}, p_{M_{n,p}}\}\leq \alpha/2)\\
=& \beta_{M,\alpha/2} + \beta_{S,\alpha/2} - \pr(p_{S_{n,p}}\leq \alpha/2, p_{M_{n,p}}\leq \alpha/2)\\
\geq&\max\{\beta_{M,\alpha/2}, \beta_{S,\alpha/2}\}.
\end{align*}
On the other hand, under $H_{1;n,p}$ in \eqref{H:indalter}, we have
\begin{align*}
\beta_{CC,\alpha} \geq \beta_{M,\alpha/2} + \beta_{S,\alpha/2} - \beta_{M,\alpha/2}\beta_{S,\alpha/2} + o_p(1).
\end{align*}
For a small $\alpha$, the difference between $\beta_{S,\alpha/2}$ and $\beta_{S,\alpha}$ should be small. We conclude that the power of the adaptive test would be no smaller than that of either the max-$L_2$-type test or max-$L_\infty$-type test.

\section{Location estimation}\label{sec:4-location}
If the null hypothesis is rejected, it is of interest to detect the locations of the changepoints. We propose the following adaptive estimation procedure:
 \begin{equation}
    \notag \hat{\tau} :=\begin{cases}\hat{\tau}_S:=\underset{1 \leq k \leq n-1}{\arg \max } \{W(k)-\hat{\mu}_{M,k}\}, & \text { if } p_{S_{n, p}}<p_{M_{n, p}}, \\ \hat{\tau}_M:=\underset{1 \leq k \leq n-1}{\arg \max } \max_{j=1,\ldots,p}\vert C_{0,j}(k)\vert^2, & \text { otherwise, }\end{cases}
\end{equation}
or
 \begin{equation}
    \notag \hat{\tau}^\dagger :=\begin{cases}\hat{\tau}_S:=\underset{1 \leq k \leq n-1}{\arg \max } \{W(k)-\hat{\mu}_{M,k}\}, & \text { if } p_{S_{n, p}}<p_{M_{n, p}^\dagger}, \\ \hat{\tau}_{M^\dagger}:=\underset{\lambda_n \leq k \leq n-\lambda_n}{\arg \max } \max_{j=1,\ldots,p}\vert C_{0.5,j}(k)\vert^2, & \text { otherwise.}\end{cases}
\end{equation}

Similar to the proof of Theorem 2.1 in \cite{10.1007/s11425-016-0058-5} and Theorem 6 in \cite{wang2023JRSSB}, we establish the following theorem on the consistency of changepoint estimation under dependence.

\begin{theorem}
\label{consistency}
Under the Assumptions~\ref{ass:C1}--\ref{ass:C4}, assumptions in Theorem~\ref{Th2}, Assumptions 1--2 in \cite{wang2023JRSSB} and $\max_{1\leq j\leq p}\vert\widehat{\sigma}_j-\sigma_j\vert=o(1)$, if $M=\lceil(n\wedge p)^{1/8}\rceil$, we have $(\hat{\tau}-\tau)/n=o_p(1)$ or $(\hat{\tau}^\dagger-\tau)/n=o_p(1)$, provided that either $\|\boldsymbol\delta\|_\infty/\sqrt{\log(np)/n}\rightarrow\infty$ or $n\|\bdelta\|^2\rightarrow\infty$, as $(n,p)\rightarrow\infty$.
\end{theorem}

\section{Simulation Results}\label{sec:5-simulation}
In this section, we conduct a comparative analysis of our proposed methods against the test procedures presented by \cite{li2019change} (hereinafter referred to as LXZL) and \cite{wang2022inference} (hereinafter referred to as WZVS). Specifically, for the WZVS test, we set the parameter \(\eta\) to 0.02 as per their method. In addition to the aforementioned comparisons, we further evaluated our methods against two other approaches that rely on the assumption of independence and identical distribution—specifically, the methods proposed by \cite{10.1007/s11425-016-0058-5} (hereinafter referred to as JPYZ) and \cite{10.1142/s201032631950014x} (hereinafter referred to as WZWY).

Additionally, to ensure consistency and comparability, we set the parameter \(M=\lceil n^{1/8}\wedge p^{1/8}\rceil\) in our method \(S_{n,p}\). Given the convergence rate characteristics of the max-$L_2$-type and max-$L_\infty$-type test procedures, we normalize our test statistics \(S_{n,p}\), \(M_{n,p}\), and \(M_{n,p}^\dagger\) by dividing them by \(1 + n^{-2/3}\log p\) in practical applications.

We generate samples from the model with
\begin{eqnarray*}
\bX_i = \bdelta \ind{i>\tau} +\sum\limits_{h=0}^{M_0} \mathbf{A}_h\boldsymbol\varepsilon_{i-h},\ i=1,\ldots,n,
\end{eqnarray*}
where $\mathbf{A}_0,\dots,\mathbf{A}_{M_0}$ are $p\times p$ matrices which determine the autocovariance structure, $\boldsymbol\varepsilon_{i}=\bms_{\bepsilon}^{1/2}\bm u_i$, $\bm u_i=(u_{i1},\cdots,u_{ip})^\top$ and $u_{ij}$ are assumed to be generated independently and identically from $N(0,1)$ and $t(4)$. This generating model implied $\cov(\bX_i,\bX_{i+h})=\bGam(h)$ for all $h\in\{-M_0,\dots,0,\dots,M_0\}$, where $\bGam(h)=\sum\nolimits_{k=0}^{M_0-h}\mathbf{A}_k\bms_{\bepsilon} \mathbf{A}_{k+h}^{\T}$. The matrices $\bms_{\bepsilon}=\mathbf{I}_p$ or $\bms_{\bepsilon}=(\rho^{i-j})_{p\times p}$, and $\mathbf{A}_h=(a_{h,ij})_{p\times p}$ is given by
\begin{equation*}
\mathbf{A}_h(i,j)=\left\{
\begin{array}{cl}
\dfrac{\phi}{h}, & \mbox{if}~ |i-j|=0,\\
\dfrac{\phi}{h|i-j|^2}, & \mbox{if}~1 \leq |i-j|\leq \lfloor pv \rfloor,\\
0,  & \mbox{if}~ |i-j|> \lfloor pv \rfloor,
\end{array}
\right.
\end{equation*}
for $h=1,\dots,M_0$. Here, $\phi$ control the dependence among variables, and $v$ stands for the sparsity level of matrices $\mathbf{A}_h$. Specially, we set $\mathbf{A}_0=\mathbf{I}_p$. The dependency level $M_0=0,2$.
We consider the following two scenarios for the covariance matrix:
\begin{itemize}
\item [(S1)] $\bms_{\bepsilon}=\mathbf{I}_p$, $\phi=0.5, v=0.5$;
\item [(S2)] $\bms_{\bepsilon}=(0.5^{|i-j|})_{1\le i,j\le p}$, $\phi=0.2,v=0.2$.
\end{itemize}
Each scenario’s empirical size and power are evaluated over 500 Monte Carlo replications, with a nominal significance level of $\alpha = 5\%$. For evaluating the performance of $S_{n,p}$ test, we taking $T_d = B = 10000$ in Remark 2.3 and get $\hat{c} = 0.9345$.

Tables \ref{tab1}--\ref{tab2} present the empirical sizes of each test under normal errors and \( t(4) \) errors, respectively. The results indicate that all the tests based on temporal dependence--$S_{n,p}, M_{n,p}, M_{n,p}^\dagger$, LXZL and WZVS generally exhibit excellent control over the empirical sizes across most scenarios. This finding underscores the robustness of our proposed methods, which are capable of maintaining accurate size control even in the presence of high-dimensional time series data with temporal dependence. This is particularly noteworthy given the complexities associated with high-dimensional settings and the potential challenges posed by temporal dependence, which can often lead to inflated Type I error rates in those testing procedures based on the independence assumption---JPYZ and WZWY. Our methods, therefore, demonstrate a high degree of reliability and applicability in such intricate and realistic data contexts.

\begin{table}[htbp]
	\centering
\renewcommand\arraystretch{0.65}
	\begin{tabular}{ccccccccccccc} \hline \hline
 $n$& $p$ &$M_0$ &$S_{n,p}$ &$M_{n,p}$ &$M_{n,p}^\dagger$ & $T_{CC}$&$T_{CC}^\dagger$ &LXZL&WZVS&JYPZ&WZWY\\ \hline
 \multicolumn{12}{c}{Scenario (S1)}\\ \hline
 400  &250  &0  &4.2  &1.4  &1.6  &4.7  &4.3  &5.6  &4.3 &5.4 &4.7\\
 400  &500  &0  &5.0  &1.6  &1.0  &5.2  &4.6  &4.3  &5.7 &5.7 &6.1\\
 800  &250  &0  &4.2  &2.5  &2.1  &5.6  &5.4  &4.2  &4.5 &4.6 &5.3\\
 800  &500  &0  &4.5  &1.3  &1.4  &4.6  &3.8  &4.4  &4.9 &4.9 &5.1\\ \hline
 400  &250  &2  &3.5  &2.1  &1.6  &4.3  &3.9  &4.5  &5.4 &100 &100\\
 400  &500  &2  &3.9  &2.3  &2.2  &5.2  &4.7  &4.6  &5.0 &100 &100 \\
 800  &250  &2  &5.6  &3.4  &3.4  &6.1  &6.5  &5.9  &5.6 &100 &100 \\
 800  &500  &2  &4.1  &3.8  &2.4  &5.7  &4.7  &5.4  &4.7 &100 &100 \\ \hline
  \multicolumn{12}{c}{Scenario (S2)}\\ \hline
 400  &250  &0  &5.2  &2.0  &2.1  &4.5  &4.6  &5.0  &5.3 &4.1  &6.3\\
 400  &500  &0  &4.8  &0.8  &1.1  &4.7  &3.9  &5.7  &5.1 &5.2  &5.8\\
 800  &250  &0  &5.3  &3.4  &2.5  &6.3  &5.8  &5.3  &5.2 &5.6  &4.3\\
 800  &500  &0  &5.6  &1.8  &1.5  &5.7  &5.5  &6.0  &4.8 &5.9  &6.3\\ \hline
 400  &250  &2  &5.4  &2.8  &2.0  &5.7  &5.6  &5.1  &4.1 &100 &100 \\
 400  &500  &2  &3.8  &1.6  &1.8  &4.7  &4.3  &4.7  &5.9 &100 &100 \\
 800  &250  &2  &4.6  &3.1  &3.1  &5.9  &5.8  &5.8  &6.0 &100 &100 \\
 800  &500  &2  &5.3  &3.2  &3.0  &6.5  &6.6  &4.1  &5.8 &100 &100 \\ \hline \hline
	\end{tabular}
	\caption{Empirical sizes of tests with normal errors.}
	\label{tab1}
\end{table}

\begin{table}[htbp]
	\centering
\renewcommand\arraystretch{0.65}
	\begin{tabular}{ccccccccccccc} \hline \hline
 $n$& $p$ &$M_0$ &$S_{n,p}$ &$M_{n,p}$ &$M_{n,p}^\dagger$ & $T_{CC}$&$T_{CC}^\dagger$&LXZL&WZVS&JYPZ&WZWY\\ \hline
 \multicolumn{12}{c}{Scenario (S1)}\\ \hline
 400  &250  &0  &2.6  &3.7  &2.8  &6.2  &6.2  &6.0  &5.4 &5.3  &5.3\\
 400  &500  &0  &5.1  &3.6  &3.1  &6.5  &6.8  &4.9  &5.7 &5.2  &4.9\\
 800  &250  &0  &4.9  &3.9  &3.2  &6.3  &7.2  &5.8  &4.8 &5.1  &5.7 \\
 800  &500  &0  &4.7  &3.8  &3.4  &6.3  &5.6  &5.3  &4.1 &5.8  &5.1\\ \hline
 400  &250  &2  &5.1  &3.3  &4.4  &6.5  &6.8  &4.3  &4.5 &100 &100 \\
 400  &500  &2  &5.3  &4.4  &4.4  &5.5  &5.4  &4.5  &4.0 &100 &100 \\
 800  &250  &2  &5.2  &3.4  &3.2  &6.7  &6.0  &5.7  &4.9 &100 &100 \\
 800  &500  &2  &4.1  &3.9  &4.0  &7.1  &6.8  &5.2  &5.2 &100 &100 \\ \hline
  \multicolumn{12}{c}{Scenario (S2)}\\ \hline
 400  &250  &0  &3.6  &3.8  &4.1  &5.5  &5.7  &5.6  &6.0 &4.7  &5.2\\
 400  &500  &0  &4.4  &4.7  &5.2  &6.7  &6.2  &4.7  &4.7 &4.1  &5.6\\
 800  &250  &0  &4.3  &3.8  &4.3  &5.5  &5.7  &4.1  &4.6 &5.8  &5.9\\
 800  &500  &0  &5.1  &4.4  &4.3  &6.3  &5.8  &4.0  &5.5 &5.7  &4.3\\ \hline
 400  &250  &2  &4.9  &3.9  &4.2  &6.1  &6.3  &4.4  &5.8 &100 &100 \\
 400  &500  &2  &4.7  &3.5  &4.5  &6.9  &7.3  &5.9  &5.9 &100 &100 \\
 800  &250  &2  &5.8  &4.5  &3.9  &7.1  &7.0  &5.4  &4.4 &100 &100 \\
 800  &500  &2  &5.2  &3.8  &3.2  &5.9  &6.1  &4.8  &4.3 &100 &100 \\ \hline \hline
	\end{tabular}
	\caption{Empirical sizes of tests with $t(4)$ errors.}
	\label{tab2}
\end{table}

For the power comparison, we specifically set the alternative parameters as follows:
\[
\delta_i = c_\tau \sqrt{\frac{\log p}{ns}}, \quad \text{for } i = 1, \cdots, s, \quad \text{and} \quad \delta_i = 0 \quad \text{for } i > s.
\]
Here, the parameter \( s \) denotes the sparsity level of the alternative hypotheses. We set $s=1, 2, 3, 5, 7, 9, 20, 30, 40, 50$. For $\tau=0.5n$, we set $c_\tau=15$. While for $\tau=0.3n$, we set $c_\tau=20$. In this analysis, we focus on the scenario with a sample size \( n = 400 \) and dimension \( p = 500 \). We note that the results for other combinations of sample sizes and dimensions exhibit similar trends.
Furthermore, we restrict our analysis to time series data with temporal dependence, specifically the case where \( M_0 = 2 \). The methods JYPZ and WZWY are excluded from the power comparison due to their significant size distortions, which compromise the validity of their power assessments.

\begin{figure}[htbp]
    \centering
    \includegraphics[width=0.9\linewidth]{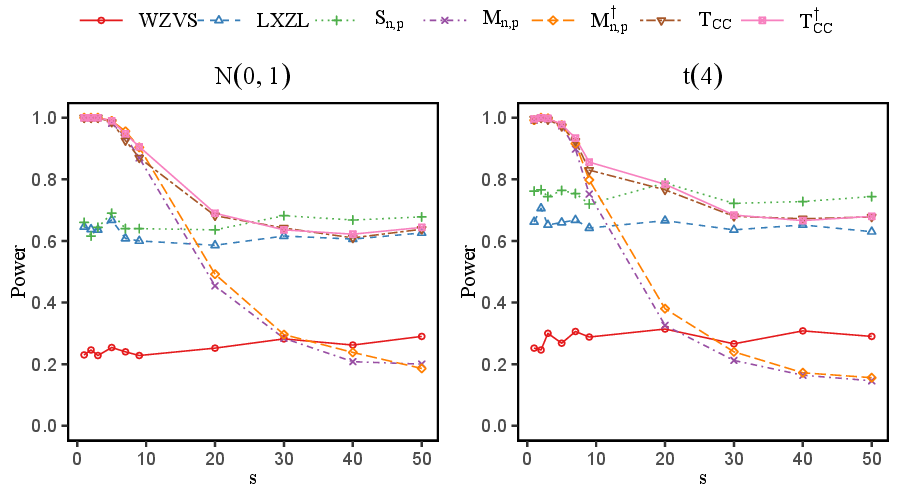}
    \caption{Power curves of each test under Scenario (S2) and $\tau=0.3n$.}
    \label{fig:power_s2_03}
\end{figure}
\begin{figure}[htbp]
    \centering
    \includegraphics[width=0.9\linewidth]{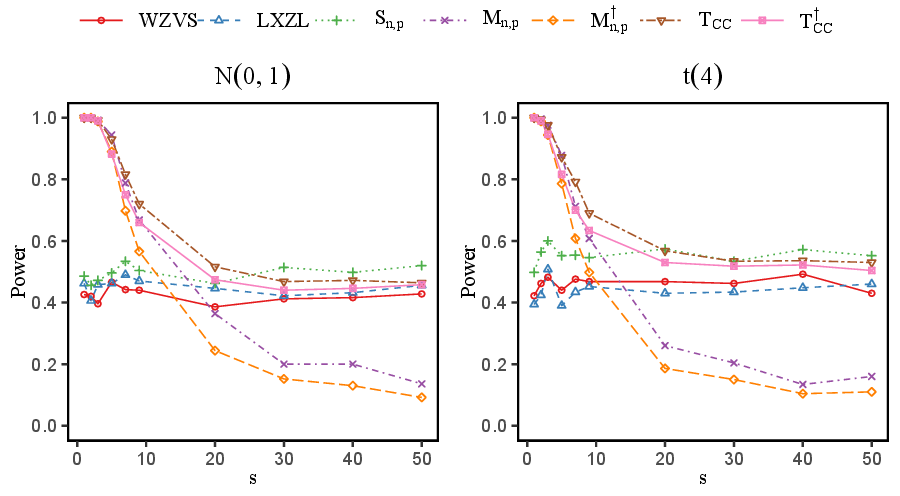}
    \caption{Power curves of each test under Scenario (S2) and $\tau=0.5n$.}
    \label{fig:power_s2_05}
\end{figure}
Figures \ref{fig:power_s2_03}--\ref{fig:power_s2_05} present the empirical power curves of each test under two scenarios for the change-point location: $\tau = 0.3n$ and $\tau = 0.5n$, respectively. Several important findings can be drawn from these results.

First, among the sum-type test procedures, our proposed test statistic $S_{n,p}$ consistently outperforms the WZVS and LXZL tests across all settings, aligning with prior results established under independent observations, see \cite{wang2023JRSSB}. This confirms the robustness and efficiency of $S_{n,p}$, even in the presence of temporal dependence. Moreover, we observe that the max-$L_2$ strategy underlying $S_{n,p}$ yields better power than the conventional sum-$L_2$-based strategies, demonstrating its advantage in aggregating signals across dimensions.

Second, in scenarios characterized by sparse alternatives—where only a small subset of components exhibits structural change—the max-type test statistics $M_{n,p}$ and $M_{n,p}^\dagger$ demonstrate significantly higher power compared to sum-type procedures. This is expected, as max-type tests are specifically tailored to detect large deviations in a few components. In contrast, under dense alternatives—where many components change simultaneously with relatively small magnitudes—the sum-type tests, particularly $S_{n,p}$, show superior power performance.

Third, the proposed Cauchy combination test procedures, which integrate both max- and sum-type information, achieve strong adaptivity across varying levels of sparsity. These combination tests consistently deliver near-optimal or optimal power in all settings, with particularly strong performance in intermediate sparsity regimes where neither extreme approach is optimal on its own. This highlights the practical advantage of the adaptive strategy.

Lastly, it is worth noting a subtle but important distinction between the two max-type statistics: $M_{n,p}$ tends to outperform $M_{n,p}^\dagger$ when the change-point is located near the center of the time series (e.g., $\tau = 0.5n$), while the reverse is true when the change-point is closer to the boundary (e.g., $\tau = 0.3n$). This behavior is consistent with findings in the change-point literature, reflecting the sensitivity of max-type statistics to boundary effects and their dependency on variance estimation near the endpoints.

\begin{figure}[htbp]
    \centering
    \includegraphics[width=0.9\linewidth]{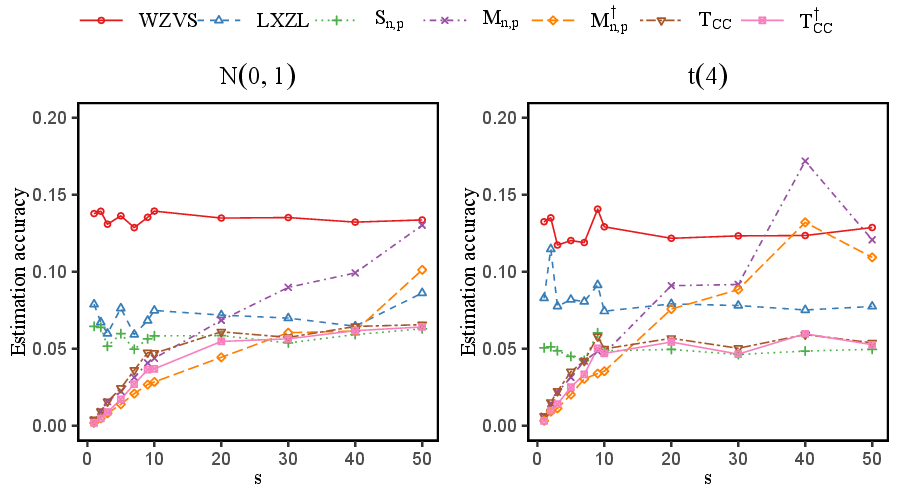}
    \caption{Estimation accuracy curves of each test under Scenario (S2) and $\tau=0.3n$.}
    \label{fig:acc_s2_03}
\end{figure}

\begin{figure}[htbp]
    \centering
    \includegraphics[width=0.9\linewidth]{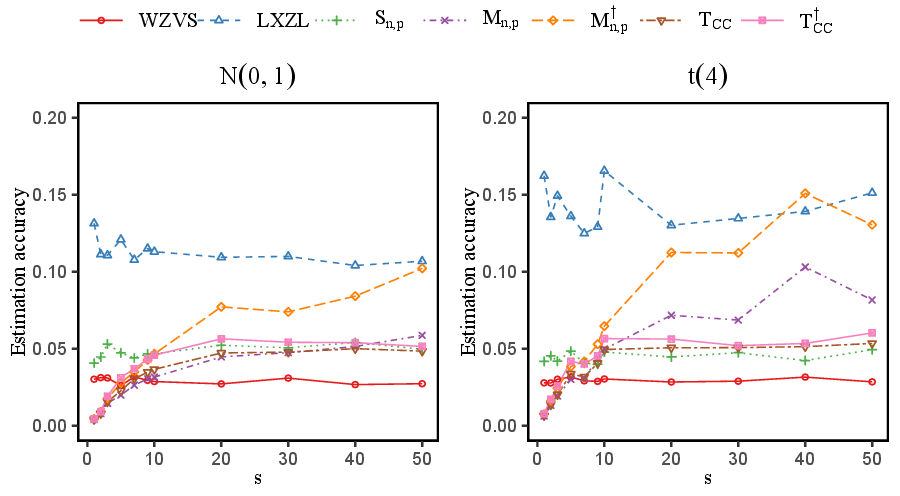}
    \caption{Estimation accuracy curves of each test under Scenario (S2) and $\tau=0.5n$.}
    \label{fig:acc_s2_05}
\end{figure}
We next turn to the single change point estimation problem.  Figures \ref{fig:acc_s2_03}--\ref{fig:acc_s2_05} present the estimation accuracy, measured by the absolute distance between the estimated and true changepoints, scaled by the sample size $n$. Each scenario’s accuracy are evaluated over 500 Monte Carlo replications.

The results show that our proposed estimator outperforms the LXZL method across all settings. The WZVS method exhibits a slight advantage, if the change-point is located near the center of the time series. However, when the change-point is near the boundaries, its estimation error increases significantly and exceeds that of our method. In contrast, our estimator demonstrates robust performance with respect to both the tail heaviness of the data distribution and the location of the change-point. Notably, the max-$L_2$ aggregation strategy adopted in the $S_{n,p}$ test statistic consistently achieves higher estimation accuracy than the traditional sum-$L_2$ strategy, highlighting its advantage in aggregating signals under high dimensionality.

Under sparse alternatives, the max-type statistics $M_{n,p}$ and $M_{n,p}^\dagger$ yield significantly better estimation accuracy than sum-based methods. This finding aligns with the testing results, as max-type procedures are designed to detect large deviations in a few components. In contrast, under dense alternatives, sum-based estimators, particularly $S_{n,p}$, demonstrate superior accuracy. Similarly, our proposed Cauchy combination estimators exhibit strong adaptability across varying levels of sparsity. This highlights the practical importance of adaptive strategies in change-point estimation.

 Overall, all the results confirm that our proposed methods—especially the Cauchy-type combination tests—offer a powerful and flexible framework for high-dimensional change-point inference under temporal dependence.
\section{Real Data Application}\label{sec:6-realdata}

In this section, we analyze financial data from the NASDAQ stock market, one of the most active equity markets globally. The NASDAQ, with its high concentration of technology and growth-oriented companies, provides a valuable platform for studying market dynamics. Its electronic trading structure and diverse listings offer a rich empirical basis for examining volatility, structural changes, and cross-sectional dependencies in modern financial markets.

In this paper, we consider the weekly closing prices of NASDAQ-listed constituent stocks over the period from January 2016 to January 2025. For each stock, we compute the corresponding weekly logarithmic return, defined as the difference between the natural logarithms of consecutive weekly closing prices, yielding 470 observations per stock. To ensure the comparability and completeness of the dataset, we retain only those firms that were continuously traded throughout the entire sample period, resulting in 1555 stocks. The log-return transformation is employed to stabilize the variance and capture proportional changes in prices, which are commonly used in empirical studies of financial returns.

We begin by examining whether serial dependence exists in the log-return series. Specifically, the Ljung–Box test \citep{ljung1978measure} for zero autocorrelation is applied to each stock series. At the $\alpha = 5\%$ significance level, a subset of the stocks exhibits statistically significant temporal dependence, while more rejections are observed when the level is relaxed to $\alpha = 5\%$. The histogram of the resulting p-values is displayed in Figure~\ref{fig:hist_of_pval}, suggesting that temporal dependence cannot be ignored. Furthermore, we employ the adaptive high-dimensional white noise test, as described in \cite{feng2022testing}, to verify if the residuals are white noise. The $p$-value of this test is 2e-26. This motivates us to employ subsequent inference procedures that explicitly accommodate such dependence structures.
\begin{figure}[htbp]
    \centering
\includegraphics[width=0.618\linewidth]{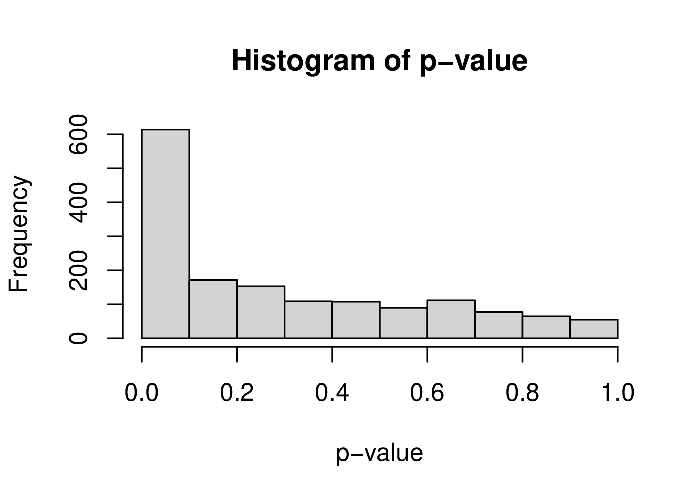}
    \caption{Histogram of $p$-value of NASDAQ constituent stocks.}
    \label{fig:hist_of_pval}
\end{figure}

\begin{table}[!htp]
\centering
\footnotesize
\begin{tabular}{ccccccc}
\hline\hline
$S_{n,p}$& $M_{n,p}$& $M_{n,p}^\dagger$& $T_{CC}$& $T_{CC}^\dagger$ &WZVS &LXZL \\
\hline
 0.0057&  0.0045& 0.0003 &0.0050   &0.0005  &0.2390 & 0.8389 \\
\hline\hline
\end{tabular}
\caption{The $p$-values for testing changepoints in weekly log-return rates.}\label{tab:NASDAQ}
\end{table}

Table~\ref{tab:NASDAQ} summarizes the $p$-values for testing changepoints in the weekly return rates. 
At the 5\% significance level, the WZVS and LXZL tests fail to reject the null hypothesis. In contrast, the sum-type $S_{n,p}$, the max-type $M_{n,p}$, $M_{n,p}^\dagger$, and their Cauchy-combination type tests yield substantially smaller $p$-values, leading to a rejection of the null and indicating the presence of at least one changepoint in the log-return rates. Notably, all the proposed test statistics reject the null hypothesis at the 5\% significance level, consistently indicating the presence of at least one changepoint in the NASDAQ log-return data. This collective evidence suggests that the market experienced structural shifts over the examined period. These results suggest that significant transitions occurred in market volatility or dependence structure, motivating further investigation into their timing and underlying economic drivers.


Building on this overall rejection, we proceed to single change-point detection. All methods  identified change point, occurring around June 2022, March 2023, February 2021, June 2022 and March 2023 for $M_{n,p},M_{n,p}^\dagger,S_{n,p},T_{CC}$ and $T_{CC}^\dagger$. The change-point in February 2021 corresponds to a dense structural transition, reflecting the post-pandemic period of abundant liquidity and the subsequent correction triggered by the sharp rise in U.S. Treasury yields, broadly affecting technology and growth stocks. In contrast, the change-points in June 2022 and March 2023 represent sparse shifts, with the former linked to the Federal Reserve’s accelerated rate hikes and quantitative tightening, and the latter associated with the collapse of Silicon Valley Bank, primarily impacting banking and technology-related sectors. These results align with economic events and demonstrate that our method can effectively detect both global and local structural changes in high-dimensional financial data.

\section{Conclusion}\label{sec:7-conclusion}

In this paper, we first propose a max-$L_2$-type test designed for detecting dense alternatives in high-dimensional change-point inference under temporal dependence. Building on this, we construct two Cauchy combination tests based on the established asymptotic independence between the proposed max-$L_2$-type test and existing max-$L_\infty$-type tests. Simulation studies and a real data application demonstrate the advantages of our adaptive procedures across a wide range of sparsity levels in the alternatives.

There remain several important directions for future research. First, our current framework assumes the observed time series follows a linear process. It would be valuable to extend the methodology to more general dependence structures, such as those satisfying $\alpha$-mixing or $\beta$-mixing conditions. Second, accurate estimation of the temporal dependence parameters---specifically the bandwidth $M$ and the long-run variance---remains a challenging task in change-point detection and warrants further investigation.

\section{Appendix}
\subsection{Notation}
For the convenience of proof, we reformulate
\begin{eqnarray*}
W(k)&=&\frac{k^2(n-k)^2}{n^3\sqrt{p}}\left(\frac{1}{k}\sum_{i=1}^k \bX_i-\frac{1}{n-k}\sum_{i=k+1}^n \bX_i\right)^\top\left(\frac{1}{k}\sum_{i=1}^k \bX_i-\frac{1}{n-k}\sum_{i=k+1}^n \bX_i\right).
\end{eqnarray*}
Let $M=\lceil(n\wedge p)^{1/8}\rceil$, we define a $M$-dependent approximation sequence for $\{\bepsilon_i\}_{i=1}^{n}$ as
\begin{eqnarray*}
\boldsymbol\varepsilon_i^{(M)}:=\E(\bepsilon_i \mid \bZ_{i-M},\dots,\bZ_i)=\bms^{1/2}\sum\limits_{\ell=0}^{M}b_{\ell}\bZ_{i-\ell}.
\end{eqnarray*}
In order to simplify notation, we omit the superscript in the absence of ambiguity. Furthermore, we denote the auto-covariance matrix of $\boldsymbol\varepsilon_i$ at lag $h$,
\begin{align*}
\bGam_M(h):=\E(\boldsymbol\varepsilon_i\boldsymbol\varepsilon_{i+h}^{\top})=c_{h,M}\bms,
\end{align*}
where $c_{h,M}=\sum\nolimits_{\ell=0}^{M-h}b_{\ell}b_{\ell+h}$. Let $\bOme_{n,M}:=\sum_{h\in \mathcal{M}}(1-\frac{|h|}{n})\bGam_M(h)$, where $\mathcal{M}=\{0,\pm1,\dots,\pm M\}$.
Let $\{\bxi_i\}_{i=1}^{n}$ be a Gaussian sequence which is independent of $\{\boldsymbol\varepsilon_i\}_{i=1}^{n}$ and preserves the auto-covariance structure, that is, $\E(\bxi_i\bxi_{i+h}^{\top})=\bGam_M(h)$. By replacing $\bX_i$ in $W(k)$ with $\boldsymbol\varepsilon_i$, we define
\begin{eqnarray*}
W(k)^{(NG)}&:=&\frac{k^2(n-k)^2}{n^3\sqrt{p}}\left(\frac{1}{k}\sum_{i=1}^k \bvarepsilon_i-\frac{1}{n-k}\sum_{i=k+1}^n \bvarepsilon_i\right)^\top\left(\frac{1}{k}\sum_{i=1}^k \bvarepsilon_i-\frac{1}{n-k}\sum_{i=k+1}^n \bvarepsilon_i\right).
\end{eqnarray*}
Similarly, we define
\begin{eqnarray*}
W(k)^{(G)}&:=&\frac{k^2(n-k)^2}{n^3\sqrt{p}}\left(\frac{1}{k}\sum_{i=1}^k \bxi_i-\frac{1}{n-k}\sum_{i=k+1}^n \bxi_i\right)^\top\left(\frac{1}{k}\sum_{i=1}^k \bxi_i-\frac{1}{n-k}\sum_{i=k+1}^n \bxi_i\right).
\end{eqnarray*}

We restate some of the notations and assumptions from \cite{wang2023JRSSB} that are referred to in this paper, which primarily concern the temporal and cross-sectional dependence structures of the noise terms $\epsilon_{ij}$. Suppose there exist measurable functions $g_j$'s such that $\epsilon_{ij} = g_j(e_{i},e_{i-1},\ldots)$, where $\{e_i\}_{i\in\mathbb{Z}}$ is a sequence of independent and identically distributed (i.i.d.) random variables. To measure temporal dependence, define for $q\geq 1$,
\begin{align*}
    a_i(q)=\limsup_{p\to\infty}\max_{j=1,\ldots,p}\|g_j(e_{i},e_{i-1},\ldots,e_{0},e_{-1},\ldots) - g_j(e_{i},e_{i-1},\ldots,e'_{0},e_{-1},\ldots)\|_{q},
\end{align*}
where $\{e'_{i}\}_{i\in\mathbb{Z}}$ is an independent copy of $\{e_i\}_{i\in\mathbb{Z}}$.
Let $\sigma_{jj'}=\lim_{n\to\infty}n^{-1}\{\E\left(\sum_{i=1}^n\sum_{i'=1}^n\epsilon_{ij}\epsilon_{i'j'}\right)\}$ be the long-run covariances, and denote $\sigma_{j}=\sigma_{jj}^{1/2}$. The componentwise correlations among $p$ variables can thus be defined as $\rho_{jj'}=\sigma_{jj'}/(\sigma_{j}\sigma_{j'})$. Denote $\R=(\rho_{jj'})_{p\times p}$.
Let for some sequences $\delta_p>0$ and $\kappa_p>0$, $B_{p,j}=\{1\leq j'\leq p: |\rho_{jj'}|\geq\delta_p\}$ and $C_p=\{1\leq j\leq p: |B_{p,j}|\geq p^{\kappa_p}\}$.

\begin{assumption}[Assumption 1 in \cite{wang2023JRSSB}]\label{asmp:corr_temporal}
There exist some constants $q>4$ and $\mathfrak{a}>5/2$ such that $a_{i}(q)\lesssim i^{-\mathfrak{a}}$. In addition, $\liminf_{p\to\infty}\min_{j=1,\ldots,p}\sigma_{j}\geq\underline{\sigma}$ for some constant $\underline{\sigma}>0$.
\end{assumption}

\begin{assumption}[Assumption 2 in \cite{wang2023JRSSB}]\label{asmp:corr_spatial}
(i) $|\rho_{jj'}|\leq\varrho$ for $1\leq j\neq j'\leq p$ and some constant $\varrho\in (0,1)$; (ii) $|C_p|/p\to 0$ for some $\delta_p=o\{(\log p)^{-1}\}$ and $\kappa_p\to 0$, as $p\to\infty$.
\end{assumption}

\subsection{Proof of Theorem \ref{Th1}}
Without loss of generality, we assume $\bmu_i:=\bmu_0 + \bdelta \ind{i>\tau} = \mathbf{0}$ under $H_0$ for $i=1,\dots,n$. Then, $\{\bX_i\}_{i=1}^n$ is a process with zero mean and auto-covariance structure given by $\bGam(h)$, that is, $\E(\bX_i\bX_{i+h})=\bGam(h)=c_h\bms$ for $h=0,1,2,\dots$ and $\bGam(h)=\bGam(-h)^\top$ with $c_h=\sum\nolimits_{\ell=0}^{\infty}b_{\ell}b_{\ell+h}$.

\begin{proof}
The skeleton of the proof can be divided into three sub-steps. Firstly, we prove the asymptotic distribution of $W(\lfloor nt\rfloor)^{(G)}$; secondly, we establish the asymptotic distribution of $W(\lfloor nt\rfloor)^{(NG)}$ with Gaussian approximation approach; finally, we use $W(\lfloor nt\rfloor)^{(NG)}$ to approximate $W(\lfloor n t\rfloor)$. Note that, the term $W(\lfloor n t\rfloor)-\mu_{M,\lfloor n t\rfloor}$ can be divided into the following four parts,
\begin{eqnarray*}
W(\lfloor n t\rfloor)-\mu_{M,\lfloor n t\rfloor}&=&\left\{W(\lfloor n t\rfloor)-W(\lfloor n t\rfloor)^{(NG)}\right\}+\left\{W(\lfloor n t\rfloor)^{(NG)}-W(\lfloor n t\rfloor)^{(G)}\right\}\\
&&+\left\{W(\lfloor n t\rfloor)^{(G)}-\mu_{t}^{(G)}\right\}+\left\{\mu_{t}^{(G)}-\mu_{M,\lfloor n t\rfloor}\right\},
\end{eqnarray*}
where $\mu_{\lfloor n t\rfloor}^{(G)}:=\E\{W(\lfloor n t\rfloor)^{(G)}\}$.
According to Lemmas \ref{Th1-L1}--\ref{Th1-L6} and Slutsky's Theorem, we will complete the proof of Theorem 1.
\end{proof}

\begin{lemma}
\label{Th1-L1}
Under Assumptions \ref{ass:C1}--\ref{ass:C3} and $H_0$, we have
$$
W(\lfloor n t\rfloor)^{(G)}-\mu_{\lfloor n t\rfloor}^{(G)} \stackrel{d}{\rightarrow} \omega V(t), \quad t \in[0,1],
$$
where $V(t)$ is a continuous Gaussian process with $\E\{V(t)\}=0$ and
$$
\E\{V(t) V(s)\}=(1-t)^{2} s^{2}, \quad 0 \leqslant s \leqslant t \leqslant 1.
$$
\end{lemma}

\begin{proof}
For $t\in[0,1]$ and $i=1,\ldots,n$, we define
\begin{eqnarray*}
a_{i,t}:=\left\{
\begin{array}{lr}
\dfrac{1}{\lfloor nt\rfloor}, &  i\leq \lfloor nt\rfloor,\\
\dfrac{-1}{n-\lfloor nt\rfloor}, &   i> \lfloor nt\rfloor.
\end{array}
\right.
\end{eqnarray*}
Define $g(k) := k^2(n-k)^2/(n^3\sqrt{p})$ and let $k = \lfloor n t\rfloor$. We reformulate $W(\lfloor nt\rfloor)^{(G)}$ as follows,
\begin{align*}
W(\lfloor nt\rfloor)^{(G)}&=g(\lfloor n t\rfloor)\sum\limits_{i,j=1}^n a_{i,t}a_{j,t}\bxi_{i}^\top \bxi_{j},
\end{align*}
and
\begin{eqnarray*}
\mu_{\lfloor n t\rfloor}^{(G)}= g(\lfloor n t\rfloor)\sum\limits_{i,j=1}^n a_{i,t}a_{j,t}\tr\{\bGam_{M}(|j-i|)\}.
\end{eqnarray*}
$F_{t}^{(G)}$ denotes an $n\times n$ matrix whose $(i,j)$-th element is defined as
\begin{eqnarray*}
F_{ij,t}^{(G)}&=&g(\lfloor n t\rfloor) a_{i,t}a_{j,t}  \left[\bxi_{i}^\top \bxi_{j}-\tr\{\bGam_{M}(|j-i|)\}\right],
\end{eqnarray*}
and hence
$$W(\lfloor nt\rfloor)^{(G)}-\mu_{\lfloor n t\rfloor}^{(G)}=\sum\limits_{i,j=1}^nF_{ij,t}^{(G)}.$$

\noindent{\textbf{Step 1.}} We first consider a single time point $t$.

For establishing the asymptotic normality of $W(\lfloor nt\rfloor)^{(G)}-\mu_{t}^{(G)}$, we use the two-dimensional triangular arrays to divide it into three parts. For any $n$, choose $\alpha_w\in(0,1)$ and $C>0$ and $w_n=Cn^{\alpha_w}>M$, then $n=w_nk_n+r_n$. Here, $w_n$ represents the size of each block, $k_n$ represents the number of blocks constructed direction, and $r_n$ stands for the points of the remainder terms which satisfy $0\leq r_n< w_n$. For any $i,j\in \{1,\dots,k_n\}$, define
\begin{align*}
B_{ij,t}^{(G)}&=\sum\limits_{\ell_1=(i-1)w_n+1}^{iw_n-M}\sum\limits_{\ell_2=(j-1)w_n+1}^{jw_n-M}F^{(G)}_{\ell_1\ell_2,t},\\
D_{ij,t}^{(G)}&=\sum\limits_{\ell_1=(i-1)w_n+1}^{iw_n}\sum\limits_{\ell_2=(j-1)w_n+1}^{jw_n}F^{(G)}_{\ell_1\ell_2,t}-B_{ij,t}^G,\\
F_{t}^{(G)} &= \sum\limits_{(i,j)\in\{1,\dots,n\}^2/}\sum\limits_{\{1,\dots,w_nk_n\}^2}F^{(G)}_{ij,t}.
\end{align*}
For non-Gaussian process $\{\boldsymbol\varepsilon_i\}_{i=1}^{n}$, we define $B_{ij,t}^{(NG)}, D_{ij,t}^{(NG)}, F_{t}^{(NG)}, F^{(NG)}_{ij,t}$ similarly. In this paper, we will use the unmarked symbol $F_{ij,t}$ to represent $F_{ij,t}^{(G)}$ or $F_{ij,t}^{(NG)}$ when we do not emphasize the difference between Gaussian and non-Gaussian situation.
Here, we reformulate $W(\lfloor nt\rfloor)^{(G)}-\mu_{t}^{(G)}:=\mathcal{H}_{1,t}+\mathcal{H}_{2,t}+\mathcal{H}_{3,t}+\mathcal{H}_{4,t}$, where
\begin{align*}
\mathcal{H}_{1,t}&=2\sum\limits_{1\leq i < j \leq k_n}B_{ij,t}, \quad
\mathcal{H}_{2,t}=\sum\limits_{i=1}^{k_n}B_{ii,t},\\
\mathcal{H}_{3,t}&=\sum\limits_{1\leq i, j \leq k_n}D_{ij,t},\quad
\mathcal{H}_{4,t}=F_t.
\end{align*}
We will show
\begin{align*}
\omega^{-1}\mathcal{H}_{1,t}\stackrel{d}{\rightarrow} V(t), \quad\quad
\omega^{-1}(\mathcal{H}_{2,t}+\mathcal{H}_{3,t}+\mathcal{H}_{4,t})\stackrel{p}{\rightarrow} 0.
\end{align*}

\noindent{\textbf{Step 1.1.}} We now focus on $\mathcal{H}_{1,t}$. Before deriving the asymptotic distribution of $\mathcal{H}_{1,t}$, we define
$\eeta_{i,t}:=a_{i,t}\bxi_{i}$, and
\begin{align*}
\tilde{\eeta}_{i,t}=\frac{1}{w_n-M}\sum\limits_{\ell=(i-1)w_n+1}^{iw_n-M}\eeta_{\ell,t},
\end{align*}
for $i=1,\dots, k_n$.
Then, for $i < j$,
\begin{align*}
B_{ij,t}=g(\lfloor n t\rfloor)(w_n-M)^2\tilde{\eeta}_{i,t}^\top\tilde{\eeta}_{j,t}.
\end{align*}
Furthermore, we define
\begin{align*}
V_{n j,t}&:=\sum\limits_{i=1}^{j-1}B_{i j,t}, ~\mbox{for}~ j=2,\dots, k_n\\
S_{n m,t}&:=\sum\limits_{j=2}^m V_{n j,t}, ~\mbox{for}~ m=2,\dots, k_n.
\end{align*}
For $m=2,\dots, k_n$, $\mathcal{F}_{n m}=\sigma(\tilde{\eeta}_{1,t},\dots,\tilde{\eeta}_{m,t})$ denotes the $\sigma$-algebra generated by $\tilde{\eeta}_{1,t},\dots,\tilde{\eeta}_{m,t}$. Then, we have
$$\mathcal{H}_{1,t}=2S_{nk_n,t}=2\sum\limits_{j=2}^{k_n}V_{nj,t}=2g(\lfloor n t\rfloor)(w_n-M)^2\sum\limits_{j=2}^{k_n}\sum\limits_{i=1}^{j-1}\tilde{\eeta}_{i,t}^\top\tilde{\eeta}_{j,t}.$$ Before deriving the asymptotic normality with Central limit theorems (CLT) of martingale differences, we need the following Propositions \ref{L6.1-P1}--\ref{L6.1-P3}.
\begin{prop}
\label{L6.1-P1}
For any $t\in [0,1]$, $\{V_{n m,t}, m=2,\dots,k_n\}$ is a martingale difference sequence with respect to the  $\sigma$-algebra $\{\mathcal{F}_{n m}, m=2,\dots,k_n\}$.
\end{prop}

\begin{proof}
We only need to prove $\{S_{n m,t}, m=2,\dots,k_n\}$ is a martingale sequence with respect to the  $\sigma$-algebra $\{\mathcal{F}_{nm}, m=2,\dots,k_n\}$ for any $t\in [0,1]$.

\begin{align*}
\E\{S_{n(m+1),t}|\mathcal{F}_{nm}\}&=g(\lfloor n t\rfloor)(w_n-M)^2\sum\limits_{j=2}^{m+1}\sum\limits_{i=1}^{j-1}\E\bigg(\tilde{\eeta}_{i,t}^\top\tilde{\eeta}_{j,t}|\mathcal{F}_{nm}\bigg)\\
&=g(\lfloor n t\rfloor)(w_n-M)^2\bigg\{\sum\limits_{j=2}^{m}\sum\limits_{i=1}^{j-1}\tilde{\eeta}_{i,t}^\top\tilde{\eeta}_{j,t}+\sum\limits_{i=1}^{m}\tilde{\eeta}_{i,t}^\top\E(\tilde{\eeta}_{m+1,t})\bigg\}\\
&=g(\lfloor n t\rfloor)(w_n-M)^2\bigg(\sum\limits_{j=2}^{m}\sum\limits_{i=1}^{j-1}\tilde{\eeta}_{i,t}^\top\tilde{\eeta}_{j,t}\bigg)=S_{n m,t}.
\end{align*}
\end{proof}

\begin{prop}
\label{L6.1-P2}
For any $t\in [0,1]$, we have
\begin{align*}
\sum\limits_{j=2}^{k_n}\E\bigg\{\frac{V_{nj,t}^2}{\omega^{2}}\bigg|\mathcal{F}_{n(j-1)}\bigg\} \stackrel{p}{\rightarrow}  \frac{t^2(1-t)^2}{4}.
\end{align*}

\end{prop}

\begin{proof}
We have
\begin{align*}
\E\left\{V_{nj,t}^2|\mathcal{F}_{n(j-1)}\right\}&=\{g(\lfloor n t\rfloor)\}^2(w_n-M)^4\E\bigg\{\bigg(\sum\limits_{i=1}^{j-1}\tilde{\eeta}_{i,t}^\top\tilde{\eeta}_{j,t}\bigg)^2\bigg|\mathcal{F}_{n(j-1)}\bigg\}\\
&=\{g(\lfloor n t\rfloor)\}^2(w_n-M)^4\sum\limits_{i_1=1}^{j-1}\sum\limits_{i_2=1}^{j-1}\tilde{\eeta}_{i_1,t}^\top\mathcal{V}_{j,t}\tilde{\eeta}_{i_2,t}.
\end{align*}
Let $\varsigma_{n,t}:=\omega^{-2}\sum\nolimits_{j=2}^{k_n}\E\{V_{nj,t}^2|\mathcal{F}_{n(j-1),t}\}$, then
\begin{align}
\label{P2-Eq1}
\E(\varsigma_{n,t})&=\{g(\lfloor n t\rfloor)\}^2(w_n-M)^4\omega^{-2}\sum\limits_{j=2}^{k_n}\sum\limits_{i=1}^{j-1}\tr(\mathcal{V}_{i,t}\mathcal{V}_{j,t}),
\end{align}
where $\mathcal{V}_{i,t}:=\var(\tilde{\eeta}_{i,t})$ for $i=1,\dots,k_n$.
Let $b_{i}^{l}:=(i-1)w_n+1$ and $b_{i}^{r}:=i w_n-M$ be the left and right points of the $i$th large block $[b_{i}^{l},b_{i}^{r}]$, define
\begin{align*}
\O_{w_n-M,M}&:=\sum_{h\in \mathcal{M}}\Big(1-\frac{|h|}{w_n-M}\Big)\bGam_M(h).
\end{align*}
When $\lfloor nt\rfloor<b_{i}^{l}$,
\begin{align*}
\mathcal{V}_{i,t}=\frac{1}{w_n-M}\cdot\frac{1}{(n-\lfloor nt\rfloor)^2}\O_{w_n-M,M}.
\end{align*}
When $\lfloor nt\rfloor\geq b_{i}^{r}$,
\begin{align*}
\mathcal{V}_{i,t}=\frac{1}{w_n-M}\cdot\frac{1}{\lfloor nt\rfloor^2}\O_{w_n-M,M}.
\end{align*}
When $b_{i}^{l}\leq \lfloor nt\rfloor< b_{i}^{r}$, we partition the $l$th large block $[b_{i}^{l},b_{i}^{r}]$ into $[b_{i}^{l},\lfloor nt\rfloor]$ and $[\lfloor nt\rfloor +1,b_{i}^{r}]$, and discuss$\mathcal{V}_{i,t}$ with the following three sub-cases (i)--(iii).
\begin{itemize}
\item [(i)] When both $\lfloor nt\rfloor-b_{i}^{l}+1> M$ and $b_{i}^{r}-\lfloor nt\rfloor > M$ hold, we have
\begin{align}
\label{V3}
\mathcal{V}_{i,t}&=\frac{1}{w_n-M}\bigg\{\frac{1}{\lfloor nt\rfloor^2}\cdot\frac{\lfloor nt\rfloor-b_{i}^{l}+1}{w_n-M}\sum_{h\in \mathcal{M}}\bigg(1-\frac{|h|}{\lfloor nt\rfloor-b_{i}^{l}+1}\bigg)\bGam_M(h)\notag\\
&+\frac{1}{(n-\lfloor nt\rfloor)^2}\cdot\frac{b_{i}^{r}-\lfloor nt\rfloor}{w_n-M}\sum_{h\in \mathcal{M}}\bigg(1-\frac{|h|}{b_{i}^{r}-\lfloor nt\rfloor}\bigg)\bGam_M(h)\bigg\}.
\end{align}
\item [(ii)] When $\lfloor nt\rfloor-b_{i}^{l}+1\leq M$ and $b_{i}^{r}-\lfloor nt\rfloor > M$, we have
\begin{small}
\begin{align}
\label{V4}
\mathcal{V}_{i,t}
&=\frac{1}{w_n-M}\bigg\{\frac{1}{\lfloor nt\rfloor^2}\cdot\frac{\lfloor nt\rfloor-b_{i}^{l}+1}{w_n-M}\sum_{h=0,\pm 1,\cdots, \pm (\lfloor nt\rfloor-b_{i}^{l})}\bigg(1-\frac{|h|}{\lfloor nt\rfloor-b_{i}^{l}+1}\bigg)\bGam_M(h)\notag\\
&+\frac{1}{(n-\lfloor nt\rfloor)^2}\cdot\frac{b_{i}^{r}-\lfloor nt\rfloor}{w_n-M}\sum_{h\in \mathcal{M}}\bigg(1-\frac{|h|}{b_{i}^{r}-\lfloor nt\rfloor}\bigg)\bGam_M(h)\notag\\
&-\frac{1}{\lfloor nt\rfloor(n-\lfloor nt\rfloor)}\bigg(\sum_{h=\pm 1,\dots, \pm(\lfloor nt\rfloor-b_i^l)}\frac{|h|}{w_n-M}\bGam_M(h)+\!\!\!\!\!\!\sum_{h=\pm (\lfloor nt\rfloor-b_i^l+1),\dots,\pm M}\!\!\!\!\!\!\frac{\lfloor nt\rfloor-b_i^l+1}{w_n-M}\bGam_M(h)\bigg)\bigg\}.
\end{align}
\end{small}
\item [(iii)] When $\lfloor nt\rfloor-b_{i}^{l}+1> M$ and $b_{i}^{r}-\lfloor nt\rfloor \leq M$, we have
\begin{small}
\begin{align}
\label{V5}
\mathcal{V}_{i,t}
&=\frac{1}{w_n-M}\bigg\{\frac{1}{\lfloor nt\rfloor^2}\cdot\frac{\lfloor nt\rfloor-b_{i}^{l}+1}{w_n-M}\sum_{h\in \mathcal{M}}\bigg(1-\frac{|h|}{\lfloor nt\rfloor-b_{i}^{l}+1}\bigg)\bGam_M(h)\notag\\
&+\frac{1}{(n-\lfloor nt\rfloor)^2}\cdot\frac{b_{i}^{r}-\lfloor nt\rfloor}{w_n-M}\sum_{h=0,\pm 1,\cdots, \pm (b_{i}^{r}-\lfloor nt\rfloor-1)}\bigg(\frac{b_{i}^{r}-\lfloor nt\rfloor}{w_n-M}-\frac{|h|}{w_n-M}\bigg)\bGam_M(h)\notag\\
&-\frac{1}{\lfloor nt\rfloor(n-\lfloor nt\rfloor)}\bigg(\sum_{h=\pm 1,\cdots,\pm (b_i^r-\lfloor nt\rfloor)}\frac{|h|}{w_n-M}\bGam_M(h)+\!\!\!\!\!\!\sum_{h=\pm (b_i^r-\lfloor nt\rfloor+1),\cdots,\pm M}\!\!\frac{b_i^r-\lfloor nt\rfloor}{w_n-M}\bGam_M(h)\bigg)\bigg\}.
\end{align}
\end{small}
\end{itemize}
Now, we back to Equation \eqref{P2-Eq1} by considering two cases:
\begin{itemize}
\item [(I)]
there exists $\ell_n$ such that $\ell_n$th large block fall in the left side of the point $\lfloor n t\rfloor$ and $(\ell_n+1)$th large block fall in the right side of the point $\lfloor n t\rfloor$, that is $b_{\ell_n}^{r}\leq \lfloor nt\rfloor< b_{\ell_n+1}^{l}$;
\item [(II)]
there exists $\ell_n$ such that the point $\lfloor n t\rfloor$ fall in the $\ell_n$th block, that is $b_{\ell_n}^{l}\leq \lfloor nt\rfloor< b_{\ell_n}^{r}$.
\end{itemize}
In case (I), $(\lfloor nt\rfloor-1-w_n)/w_n<\ell_n<(\lfloor nt\rfloor+M)/w_n$, it means that $\lim_{n\to \infty}\ell_nw_n/n= t$ and $\lim_{n\to \infty}(k_n-\ell_n)w_n/n=1-t$. We have
\begin{eqnarray*}
\sum\limits_{j=2}^{k_n}\E(V_{nj,t}^2)
&=&\sum\limits_{j=2}^{\ell_n}\E(V_{nj,t}^2)+\sum\limits_{j=\ell_n+1}^{k_n}\E(V_{nj,t}^2)\\
&=&\{g(\lfloor n t\rfloor)\}^2(w_n-M)^4\bigg(\sum\limits_{j=2}^{\ell_n}\sum\limits_{i=1}^{j-1}+\sum\limits_{j=\ell_n+1}^{k_n}\sum\limits_{i=1}^{\ell_n}+\sum\limits_{j=\ell_n+1}^{k_n}\sum\limits_{i=\ell_n+1}^{j-1}\bigg)\tr(\mathcal{V}_{i,t}\mathcal{V}_{j,t})\\
&=& \{g(\lfloor n t\rfloor)\}^2(w_n-M)^2\bigg\{\frac{1}{\lfloor n t\rfloor^4}\cdot\frac{\ell_n(\ell_n-1)}{2}+
\frac{1}{\lfloor n t\rfloor^2(n-\lfloor nt\rfloor)^2}\cdot(k_n-\ell_n)\ell_n\\
&&+\frac{1}{(n-\lfloor nt\rfloor)^4}\cdot\frac{(k_n-\ell_n)(k_n-\ell_n-1)}{2}\bigg\}\tr(\O_{w_n-M,M}^2)\\
&=& \bigg\{\frac{(n-\lfloor nt\rfloor)^4(w_n-M)^2}{n^6p}\cdot\frac{\ell_n(\ell_n-1)}{2}+
\frac{\lfloor n t\rfloor^2(n-\lfloor nt\rfloor)^2(w_n-M)^2}{n^6p}\cdot(k_n-\ell_n)\ell_n\\
&&+\frac{\lfloor n t\rfloor^4(w_n-M)^2}{n^6p}\cdot\frac{(k_n-\ell_n)(k_n-\ell_n-1)}{2}\bigg\}\tr(\O_{w_n-M,M}^2)\\
&\rightarrow& \bigg\{\frac{(1-t)^4t^2}{2}+\frac{2(1-t)^3t^3}{2}+\frac{(1-t)^2t^4}{2}\bigg\}\lim_{p\to \infty}\frac{\tr(\O_{w_n-M,M}^2)}{p}\\
&=& \frac{t^2(1-t)^2}{2}\cdot\lim_{p\to \infty}\frac{\tr(\O_{w_n-M,M}^2)}{p}.
\end{eqnarray*}
Define $d_{n,M} := \sum\nolimits_{h\in\mathcal{M}}(1-\frac{|h|}{n}) c_h$ and $d_{w_n-M,M} := \sum\nolimits_{h\in\mathcal{M}}(1-\frac{|h|}{w_n-M}) c_h$, by Assumption~\ref{ass:C1}-(ii) and $M=\lceil(n\wedge p)^{1/8}\rceil$, we have
\begin{align*}
d_{n,M} =  \sum\limits_{h\in\mathcal{M}}\sum\limits_{\ell=0}^{\infty}b_{\ell}b_{\ell+h} - \frac{1}{n}\sum\limits_{h\in\mathcal{M}}|h| c_h \rightarrow s^2, \quad \mbox{as} \quad n,p\rightarrow \infty
\end{align*}
with the fact that $\sum\nolimits_{h=-\infty}^{\infty}|h| |c_h| < \infty$. Similarly, $d_{w_n-M,M} \rightarrow s^2$ as $n,p\rightarrow \infty$. Thus, as $n,p\rightarrow \infty$,
\begin{align*}
\frac{|\tr(\O_{w_n-M,M}^2)-\tr(\O_{n,M}^2)|}{\tr(\O_{n,M}^2)}
&= \bigg|\frac{d_{w_n-M,M}^2 - d_{n,M}^2}{d_{n,M}^2}\bigg| = o(1).
\end{align*}
Furthermore, by Lemma \ref{Th1-L5}--\ref{Th1-L6}, we have
$$\E(\varsigma_{n,t}) = \omega^{-2}\sum\limits_{j=2}^{k_n}\E(V_{nj,t}^2)\rightarrow \frac{t^2(1-t)^2}{4}.$$
Furthermore,
\begin{eqnarray*}
\var(\varsigma_{n,t})&=&\E(\varsigma_{n,t}^2)-\{\E(\varsigma_{n,t})\}^2\\
&=&\{g(\lfloor n t\rfloor)\}^4(w_n-M)^8\omega^{-4}\bigg\{\sum\limits_{j_1,j_2=2}^{k_n}\sum\limits_{i_1,i_3=1}^{j_1-1}\sum\limits_{i_2,i_4=1}^{j_2-1}\E\bigg(\tilde{\eeta}_{i_1,t}^\top\mathcal{V}_{j_1,t}\tilde{\eeta}_{i_3,t}\tilde{\eeta}_{i_2,t}^\top\mathcal{V}_{j_2,t}\tilde{\eeta}_{i_4,t}\bigg)\\
&&-\sum\limits_{j_1,j_2=2}^{k_n}\sum\limits_{i_1=1}^{j_1-1}\sum\limits_{i_2=1}^{j_2-1}\tr\bigg(\mathcal{V}_{i_1,t}\mathcal{V}_{j_1,t}\bigg)\tr\bigg(\mathcal{V}_{i_2,t}\mathcal{V}_{j_2,t}\bigg)\bigg\}\\
&=&\{g(\lfloor n t\rfloor)\}^4(w_n-M)^8\omega^{-4}\bigg\{\sum\limits_{j_1,j_2=2}^{k_n}\sum\limits_{i_1\neq i_3=1}^{j_1-1}\sum\limits_{i_2\neq i_4=1}^{j_2-1}\E\bigg(\tilde{\eeta}_{i_1,t}^\top\mathcal{V}_{j_1,t}\tilde{\eeta}_{i_3,t}\tilde{\eeta}_{i_2,t}^\top\mathcal{V}_{j_2,t}\tilde{\eeta}_{i_4,t}\bigg)\bigg\}\\
&=&\{g(\lfloor n t\rfloor)\}^4(w_n-M)^8\omega^{-4}\bigg\{\sum\limits_{j_1,j_2=2}^{k_n}\sum\limits_{i_1\neq i_2=1}^{j_1\wedge j_2-1}\tr(\mathcal{V}_{i_1,t}\mathcal{V}_{j_1,t}\mathcal{V}_{i_2,t}\mathcal{V}_{j_2,t})\bigg\}\\
&=&O\bigg\{\frac{n^4(w_n-M)^8}{p^2}\cdot\frac{p^2}{\tr^2(\O_{n,M}^2)}\cdot\frac{k_n^4\tr(\O_{w_n-M,M}^4)}{n^8(w_n-M)^4}\cdot \frac{\tr^2(\O_{n,M}^2)}{\tr^2(\O_{n}^2)}\bigg\}\\
&=& O\bigg\{\frac{\tr(\O_{w_n-M,M}^4)}{\tr^2(\O_{n,M}^2)}\bigg\}=O(p^{-1})\\
&=&o(1),
\end{eqnarray*}
the last equation holds because both $\O_{w_n-M,M}$ and $\O_{n,M}$ are bounded in spectral norm.
Above all,
\begin{align*}
\sum\limits_{j=2}^{k_n}\E\bigg\{\frac{V_{nj,t}^2}{\omega^{2}}\bigg|\mathcal{F}_{n(j-1),t}\bigg\} \stackrel{p}{\rightarrow}  \frac{t^2(1-t)^2}{4}.
\end{align*}

\noindent In case (II), we follow the similar discussion of case (I). By Assumption~\ref{ass:C3}, considering Equations~\eqref{V3}--\eqref{V5}, for $i=1,\dots,k_n$, we have $\lambda_{\max}(\mathcal{V}_{i,t})=O\{n^{-2}(w_n-M)^{-1}\}$, and $p^{-1}\tr(\O_{w_n-M,M}^2)>d_{w_n-n,M}^2C_0^2$,
then
\begin{align*}
\bigg|\frac{\E(V_{n\ell_n,t}^2)}{p^{-1}\tr(\O_{w_n-M,M}^2)}\bigg|&\leq\{g(\lfloor n t\rfloor)\}^2(w_n-M)^4\sum\limits_{i=1}^{\ell_n-1}\bigg|\frac{\tr(\mathcal{V}_{i,t}\mathcal{V}_{\ell_n,t})}{d_{w_n-n,M}^2C_0^2}\bigg|\\
&=O\bigg\{\frac{n^2(w_n-M)^4(\ell_n-1)p}{pn^4(w_n-M)^2}\bigg\}=O\bigg(\frac{w_n}{n}\bigg)\\
&=o(1).
\end{align*}
Then
\begin{align*}
\E(V_{n\ell_n,t}^2)=o\bigg\{\frac{\tr(\O_{w_n-M,M}^2)}{p}\bigg\}.
\end{align*}
Thus,
\begin{eqnarray*}
\sum\limits_{j=2}^{k_n}\E(V_{nj,t}^2)
&=&\sum\limits_{j=2}^{\ell_n-1}\E(V_{nj,t}^2)+\E(V_{n\ell_n,t}^2)+\sum\limits_{j=\ell_n+1}^{k_n}\E(V_{nj,t}^2)\\
&=&\frac{\{g(\lfloor n t\rfloor)\}^2(w_n-M)^2}{\lfloor n t\rfloor^4}\sum\limits_{j=2}^{\ell_n-1}(j-1)\tr(\O_{w_n-M,M}^2)\\
&&+\frac{\{g(\lfloor n t\rfloor)\}^2(w_n-M)^2}{\lfloor n t\rfloor^2(n-\lfloor nt\rfloor)^2}(k_n-\ell_n)\ell_n\tr(\O_{w_n-M,M}^2)\\
&&+\frac{\{g(\lfloor n t\rfloor)\}^2(w_n-M)^2}{(n-\lfloor nt\rfloor)^4}\sum\limits_{j=\ell_n+1}^{k_n}(j-\ell_n-1)\tr(\O_{w_n-M,M}^2)
+o\bigg\{\frac{\tr(\O_{w_n-M,M}^2)}{p}\bigg\}\\
&= &\frac{(n-\lfloor nt\rfloor)^4(w_n-M)^2}{n^6p}\cdot\frac{(\ell_n-1)(\ell_n-2)}{2}\tr(\O_{w_n-M,M}^2)\\
&&+\frac{\lfloor n t\rfloor^2(n-\lfloor nt\rfloor)^2(w_n-M)^2}{n^6p}(k_n-\ell_n)\ell_n\tr(\O_{w_n-M,M}^2)\\
&&+\frac{\lfloor n t\rfloor^4(w_n-M)^2}{n^6p}\cdot\frac{(k_n-\ell_n)(k_n-\ell_n-1)}{2}\tr(\O_{w_n-M,M}^2)
+o\bigg\{\frac{\tr(\O_{w_n-M,M}^2)}{p}\bigg\}\\
&\rightarrow& \bigg\{\frac{(1-t)^4t^2}{2}+\frac{2(1-t)^3t^3}{2}+\frac{(1-t)^2t^4}{2}+o(1)\bigg\}\lim_{p\to \infty}\frac{\tr(\O_{w_n-M,M}^2)}{p}\\
&=& \frac{t^2(1-t)^2}{2}\cdot\lim_{p\to \infty}\frac{\tr(\O_{w_n-M,M}^2)}{p}\{1+o(1)\}.
\end{eqnarray*}
Then, $$\E(\varsigma_{n,t}) = \omega^{-2}\sum\limits_{j=2}^{k_n}\E(V_{nj,t}^2)\rightarrow \frac{t^2(1-t)^2}{4}.$$
By taking similar arguments, we have $\var(\varsigma_{n,t}) = o(1)$.
Now, Proposition \ref{L6.1-P2} is finished.
\end{proof}

\begin{prop}
\label{L6.1-P3}
For any $t\in [0,1]$ and $\epsilon>0$, we have
\begin{align*}
\frac{1}{\omega^2}\sum\limits_{j=2}^{k_n}\E\bigg\{V_{nj,t}^2\mathbb{I}(\vert V_{nj,t}\vert>\omega\epsilon)\bigg |\mathcal{F}_{n(j-1)}\bigg\} \stackrel{p}{\rightarrow}  0.
\end{align*}
\end{prop}

\begin{proof}
It suffices to show
\begin{align*}
\E\bigg[\sum\limits_{j=2}^{k_n}\E\left\{V_{nj,t}^4|\mathcal{F}_{n(j-1)}\right\}\bigg]=o(\omega^4).
\end{align*}

Due to the fact
\begin{align*}
\frac{1}{\omega^2}\sum\limits_{j=2}^{k_n}\E\{V_{nj,t}^2\mathbb{I}(|V_{nj,t}|>\omega\epsilon)|\mathcal{F}_{n(j-1)}\}\leq \frac{1}{\omega^4\epsilon^2}\sum\limits_{j=2}^{k_n}\E\{V_{nj,t}^4|\mathcal{F}_{n(j-1)}\},
\end{align*}
we have
\begin{eqnarray*}
&&\E\bigg[\sum\limits_{j=2}^{k_n}\E\{V_{nj,t}^4|\mathcal{F}_{n(j-1)}\}\bigg]\\
&=&\{g(\lfloor n t\rfloor)\}^4(w_n-M)^8\sum\limits_{j=2}^{k_n}\E\Big(\sum\limits_{i=1}^{j-1}\tilde{\eeta}_{i,t}^\top\tilde{\eeta}_{j,t}\Big)^4\\
&=&\{g(\lfloor n t\rfloor)\}^4(w_n-M)^8\sum\limits_{j=2}^{k_n}\E\Big(\sum\limits_{i_1,i_2,i_3,i_4=1}^{j-1}\tilde{\eeta}_{i_1,t}^\top\tilde{\eeta}_{j,t}\tilde{\eeta}_{i_2,t}^\top\tilde{\eeta}_{j,t}\tilde{\eeta}_{i_3,t}^\top\tilde{\eeta}_{j,t}\tilde{\eeta}_{i_4,t}^\top\tilde{\eeta}_{j,t}\Big)\\
&=&\{g(\lfloor n t\rfloor)\}^4(w_n-M)^8\sum\limits_{j=2}^{k_n}\E\Big\{3\sum\limits_{i_1\neq i_2}^{j-1}\tr(\mathcal{V}_{i_1,t}\mathcal{V}_{j,t})\tr(\mathcal{V}_{i_2,t}\mathcal{V}_{j,t})+\sum\limits_{i=1}^{j-1}\tr(\mathcal{V}_{i,t}\mathcal{V}_{j,t})^2\Big\}\\
&\leq& \frac{Cn^4(w_n-M)^8}{p^2}\cdot\frac{1}{n^8(w_n-M)^4}\bigg\{\frac{k_n(k_n+1)(2k_n+1)}{6}\tr^2(\O_{w_n-M,M}^2)\\
&&+\frac{k_n(k_n-1)}{2}\tr(\O_{w_n-M,M}^4)\bigg\}\\
&=&o(\omega^4),
\end{eqnarray*}
with $\tr(\O_{w_n-M,M}^4) = o\{\tr^2(\O_{w_n-M,M}^2)\}$.
Thus, the Lindberg condition of CLT is satisfied. Proposition \ref{L6.1-P3} has been proved.
\end{proof}

\noindent{\textbf{Step 1.2.}}
In this step, we will prove $\omega^{-1}(\mathcal{H}_{2,t}+\mathcal{H}_{3,t}+\mathcal{H}_{4,t})\stackrel{p}{\rightarrow} 0.$
For the sake of proof, we reformulate
\begin{align*}
&\mathcal{H}_{2,t}+\mathcal{H}_{3,t}+\mathcal{H}_{4,t}\\
=&\sum\limits_{i=1}^{k_n}B_{ii,t}+\sum\limits_{i=1}^{k_n}\sum\limits_{j=1}^{k_n}D_{ij,t}+F_t\\
=&\sum\limits_{i=1}^{k_n}\sum\limits_{\ell_1=(i-1)w_n+1}^{iw_n-M}\sum\limits_{\ell_2=(i-1)w_n+1}^{iw_n-M}F_{\ell_1\ell_2,t}+2\sum\limits_{i=1}^{k_n}\Big(\sum\limits_{\ell_1=1}^n\sum\limits_{\ell_2=iw_n-M+1}^{iw_n}F_{\ell_1\ell_2,t}\Big)\\
&+2\sum\limits_{i=1}^{k_n}\Big(\sum\limits_{\ell_1=k_nw_n+1}^n\sum\limits_{\ell_2=(i-1)w_n+1}^{iw_n-M}F_{\ell_1\ell_2,t}\Big)
-\sum\limits_{i=1}^{k_n}\sum\limits_{j=1}^{k_n}\sum\limits_{\ell_1=iw_n-M+1}^{iw_n}\sum\limits_{\ell_2=jw_n-M+1}^{jw_n}F_{\ell_1\ell_2,t}\\
&+\sum\limits_{\ell_1=k_nw_n+1}^n\sum\limits_{\ell_2=k_nw_n+1}^n F_{\ell_1\ell_2,t}.
\end{align*}
Thus, we have $\E(F_{ij,t})=0$, and
\begin{align*}
\cov(F_{i_1j_1,t}, F_{i_2j_2,t}) &= \{g(\lfloor n t\rfloor)\}^2 a_{i_1,t}a_{j_1,t}a_{i_2,t}a_{j_2,t}\big[\tr\{\bGam_{M}(|i_2-i_1|)\bGam_{M}(|j_2-j_1|)\}\\
&+\tr\{\bGam_{M}(|i_1-j_2|)\bGam_{M}(|i_2-j_1|)\big].
\end{align*}
Thus, $\E(B_{ii,t})=0$, $\cov(B_{ii,t},B_{jj,t})=0$ for any $i\neq j$, and
\begin{align*}
\var\Big(\sum\limits_{i=1}^{k_n}B_{ii,t}\Big)
&=\sum\limits_{i=1}^{k_n}\E\bigg(\sum\limits_{\ell_1=(i-1)w_n+1}^{iw_n-M}\sum\limits_{\ell_2=(i-1)w_n+1}^{iw_n-M}F_{\ell_1\ell_2,t}\bigg)^2\\
&=\sum\limits_{i=1}^{k_n}\sum\limits_{\ell_1,\ell_1^{'},\ell_2,\ell_2^{'}=(i-1)w_n+1}^{iw_n-M}\E\bigg(F_{\ell_1\ell_2,t}F_{\ell_1{'}\ell_2{'},t}\bigg)\\
&=O\bigg\{\frac{k_n(w_n-M)^2\tr(\O_{w_n-M, M}^2)}{n^2p}\bigg\}.
\end{align*}
Recall $\omega=\lim_{p\to \infty}\sqrt{2\tr(\O^2)/p}$ and $M=\lceil(n\wedge p)^{1/8}\rceil$, we have
\begin{align*}
\var\Big(\omega^{-1}\sum\limits_{i=1}^{k_n}B_{ii,t}\Big)
&= O\bigg\{\frac{(w_n-M)\tr(\O_{w_n-M, M}^2)}{n\tr(\O_{n}^2)}\bigg\}=o(1).
\end{align*}
Then, we have $\omega^{-1}\sum\nolimits_{i=1}^{k_n}B_{ii,t}=o_p(1)$.
Let $\mathcal{R}_{i, t}:=\sum\limits_{\ell_1=1}^n\sum\limits_{\ell_2=iw_n-M+1}^{iw_n}F_{\ell_1\ell_2,t}$, for sufficiently large $n$, we have $\mathcal{R}_{i, t}$ and $\mathcal{R}_{j, t}$ are independent because $w_n$ is greater than $M$.
Then,
\begin{align*}
\var\Big(\sum\limits_{i=1}^{k_n}\mathcal{R}_{i, t}\Big)
&=\sum\limits_{i=1}^{k_n}\E\bigg(\sum\limits_{\ell_1=1}^n\sum\limits_{\ell_2=iw_n-M+1}^{iw_n}F_{\ell_1\ell_2,t}\bigg)^2\\
&=\sum\limits_{i=1}^{k_n}\sum\limits_{\ell_1=1}^n\sum\limits_{\ell_2=iw_n-M+1}^{iw_n}\sum\limits_{\ell_1^{'}=1}^n\sum\limits_{\ell_2^{'}=iw_n-M+1}^{iw_n}\E\bigg(F_{\ell_1\ell_2,t}F_{\ell_1{'}\ell_2{'},t}\bigg)\\
&= O\bigg\{\frac{k_nnM\tr(\O_{n,M}\O_{M,M})}{n^2p}\bigg\}.
\end{align*}
Thus $\omega^{-1}\sum\nolimits_{i=1}^{k_n}\mathcal{R}_{i, t}=o_p(1)$.
Similarly, we have
\begin{align*}
\var\bigg\{\omega^{-1}\sum\limits_{i=1}^{k_n}\Big(\sum\limits_{\ell_1=k_nw_n+1}^n\sum\limits_{\ell_2=(i-1)w_n+1}^{iw_n-M}F_{\ell_1\ell_2,t}\Big)\bigg\}=o(1),
\end{align*}
and
\begin{align*}
\var\bigg(\omega^{-1}\sum\limits_{\ell_1=k_n w_n+1}^n\sum\limits_{\ell_2=k_nw_n+1}^n F_{\ell_1\ell_2,t}\bigg)=o(1).
\end{align*}
Thus, $\mathcal{H}_{2,t}+\mathcal{H}_{3,t}+\mathcal{H}_{4,t}\stackrel{p}{\rightarrow} 0.$ Following Propositions \ref{L6.1-P1}--\ref{L6.1-P3}, for fixed $t\in[0,1]$, we have
$$
W(\lfloor n t\rfloor)^{(G)}-\mu_{\lfloor n t\rfloor}^{(G)}\stackrel{d}{\rightarrow} \omega N\left(0, t^2(1-t)^{2}\right).
$$

\noindent{\textbf{Step 2.}}  We consider two time points $t$ and $s$ with $s<t$.
In this step, we consider the convergence in distribution of
\begin{eqnarray*}
&&\omega^{-1}\left[a\{W(\lfloor n t\rfloor)^{(G)}-\mu_{t}^{(G)}\}+b\{W(\lfloor n s\rfloor)^{(G)}-\mu_{s}^{(G)}\}\right]\\
&=&\omega^{-1}(a\mathcal{H}_{1,t}+b\mathcal{H}_{1,s})+o_p(1)\\
&=&2\omega^{-1}\bigg(a\sum\limits_{j=2}^{k_n}V_{nj,t}+b\sum\limits_{j=2}^{k_n}V_{nj,s}\bigg)+o_p(1):=2\omega^{-1}\sum\limits_{j=2}^{k_n}\tilde{V}_{nj,ts}+o_p(1).
\end{eqnarray*}
It is obvious that $\{\tilde{V}_{nm,ts}, m=2,\dots,k_n\}$ is a martingale difference sequence with respect to the  $\sigma$-algebra $\{\mathcal{F}_{nm}, m=2,\dots,k_n\}$. Now, we need to consider
\begin{align*}
\sum\limits_{j=2}^{k_n}\E(\tilde{V}_{nj,ts}^2)&=a^2\sum\limits_{j=2}^{k_n}\E(V_{nj,t}^2)+b^2\sum\limits_{j=2}^{k_n}\E(V_{nj,s}^2)+2ab\sum\limits_{j=2}^{k_n}\E(V_{nj,t}V_{nj,s}),
\end{align*}
and
\begin{align*}
\E\{V_{nj,t}V_{nj,s}|\mathcal{F}_{n(j-1)}\}&=g(\lfloor n t\rfloor)g(\lfloor n s\rfloor)(w_n-M)^4\E\bigg\{\bigg(\sum\limits_{i=1}^{j-1}\tilde{\eeta}_{i,t}^\top\tilde{\eeta}_{j,t}\bigg)\bigg(\sum\limits_{i=1}^{j-1}\tilde{\eeta}_{i,s}^\top\tilde{\eeta}_{j,s}\bigg)\bigg|\mathcal{F}_{n(j-1)}\bigg\}\\
&=g(\lfloor n t\rfloor)g(\lfloor n s\rfloor)(w_n-M)^4\sum\limits_{i_1=1}^{j-1}\sum\limits_{i_2=1}^{j-1}\tilde{\eeta}_{i_1,t}^\top\mathcal{V}_{j,ts}\tilde{\eeta}_{i_2,s},\\
\E(V_{nj,t}V_{nj,s})&=g(\lfloor n t\rfloor)g(\lfloor n s\rfloor)(w_n-M)^4\sum\limits_{i=1}^{j-1}\tr(\mathcal{V}_{i,ts}\mathcal{V}_{j,ts}),
\end{align*}
where $\mathcal{V}_{i,ts}:=\E(\tilde{\eeta}_{i,t}\tilde{\eeta}_{i,s}^\top)$, for $i=1,\dots,k_n$.
If $b_{i}^{r}\leq  \lfloor ns\rfloor$,
\begin{align*}
\mathcal{V}_{i,ts}=
\frac{1}{(w_n-M)}\cdot\frac{1}{\lfloor nt\rfloor\lfloor ns\rfloor}\O_{w_n-M,M}.
\end{align*}
If $b_{i}^{l}\leq \lfloor ns\rfloor< b_{i}^{r}\leq \lfloor nt\rfloor$,
\begin{eqnarray*}
\mathcal{V}_{i,ts}&=&
\frac{1}{(w_n-M)}\cdot\bigg\{\frac{1}{\lfloor nt\rfloor\lfloor ns\rfloor}\O_{w_n-M,M}-\frac{1}{\lfloor nt\rfloor(n-\lfloor ns\rfloor)}\O_{w_n-M,M}\bigg\}.
\end{eqnarray*}
If $\lfloor ns\rfloor<[b_{i}^{l}, b_{i}^{r}]<\lfloor nt\rfloor$,
\begin{align*}
\mathcal{V}_{i,ts}=
-\frac{1}{(w_n-M)}\cdot\frac{1}{\lfloor nt\rfloor(n-\lfloor ns\rfloor)}\O_{w_n-M,M}.
\end{align*}
If $b_{i}^{l}\leq \lfloor nt\rfloor< b_{i}^{r}$,
\begin{eqnarray*}
\mathcal{V}_{i,ts}&=&
\frac{1}{(w_n-M)}\cdot\bigg\{\frac{1}{(n-\lfloor nt\rfloor)(n-\lfloor ns\rfloor)}\O_{w_n-M,M}-\frac{1}{(n-\lfloor nt\rfloor)\lfloor ns\rfloor}\O_{w_n-M,M}\bigg\}.
\end{eqnarray*}
If $b_{i}^{l}>  \lfloor nt\rfloor$,
\begin{align*}
\mathcal{V}_{i,ts}=
\frac{1}{(w_n-M)}\cdot\frac{1}{(n-\lfloor nt\rfloor)(n-\lfloor ns\rfloor)}\O_{w_n-M,M}.
\end{align*}
We consider the case: there exists $\ell_n^{'} < \ell_n$, such that $b_{\ell_n^{'}}^{r}\leq \lfloor ns\rfloor< b_{\ell_n^{'}+1}^{l}$ and $b_{\ell_n}^{r}\leq \lfloor nt\rfloor< b_{\ell_n+1}^{l}$. It follows that, $\lim_{n\to \infty}\ell_n^{'}w_n/n= s$, $\lim_{n\to \infty}(k_n-\ell_n^{'})w_n/n=1-s$, $\lim_{n\to \infty}\ell_nw_n/n= t$,  and $\lim_{n\to \infty}(k_n-\ell_n)w_n/n=1-t$. In this case, we have
\begin{eqnarray*}
\sum\limits_{j=2}^{k_n}\E(V_{nj,t}V_{nj,s})
&=&\sum\limits_{j=2}^{\ell_n^{'}}\E(V_{nj,t}V_{nj,s})+\sum\limits_{j=\ell_n^{'}+1}^{\ell_n}\E(V_{nj,t}V_{nj,s})+\sum\limits_{j=\ell_n+1}^{k_n}\E(V_{nj,t}V_{nj,s})\\
&=&\sum\limits_{j=2}^{\ell_n^{'}}\sum\limits_{i=1}^{j-1}\frac{g(\lfloor n t\rfloor)g(\lfloor n s\rfloor)(w_n-M)^2}{\lfloor nt\rfloor^2\lfloor ns\rfloor^2}\tr(\O^2_{w_n-M,M})\\
&&-\sum\limits_{j=\ell_n^{'}+1}^{\ell_n}\ell_n^{'}\frac{g(\lfloor n t\rfloor)g(\lfloor n s\rfloor)(w_n-M)^2}{\lfloor nt\rfloor^2\lfloor ns\rfloor(n-\lfloor ns\rfloor)}\tr(\O^2_{w_n-M,M})\\
&&+\sum\limits_{j=\ell_n^{'}+1}^{\ell_n}(j-1-\ell_n^{'})\frac{g(\lfloor n t\rfloor)g(\lfloor n s\rfloor)(w_n-M)^2}{\lfloor nt\rfloor^2(n-\lfloor ns\rfloor)^2}\tr(\O^2_{w_n-M,M})\\
&&+\sum\limits_{j=\ell_n+1}^{k_n}\ell_n^{'}\frac{g(\lfloor n t\rfloor)g(\lfloor n s\rfloor)(w_n-M)^2}{\lfloor nt\rfloor\lfloor ns\rfloor(n-\lfloor nt\rfloor)(n-\lfloor ns\rfloor)}\tr(\O^2_{w_n-M,M})\\
&&-\sum\limits_{j=\ell_n+1}^{k_n}(\ell_n-\ell_n^{'})\frac{g(\lfloor n t\rfloor)g(\lfloor n s\rfloor)(w_n-M)^2}{\lfloor nt\rfloor(n-\lfloor nt\rfloor)(n-\lfloor ns\rfloor)^2}\tr(\O^2_{w_n-M,M})\\
&&+\sum\limits_{j=\ell_n+1}^{k_n}(j-1-\ell_n)\frac{g(\lfloor n t\rfloor)g(\lfloor n s\rfloor)(w_n-M)^2}{(n-\lfloor nt\rfloor)^2(n-\lfloor ns\rfloor)^2}\tr(\O^2_{w_n-M,M})\\
&=&\bigg\{\frac{\ell_n^{'}(\ell_n^{'}-1)}{2}\cdot\frac{(n-\lfloor nt\rfloor)^2(n-\lfloor ns\rfloor)^2(w_n-M)^2}{n^6}\\
&&-(\ell_n-\ell_n^{'})\ell_n^{'}\cdot\frac{(n-\lfloor nt\rfloor)^2\lfloor ns\rfloor(n-\lfloor ns\rfloor)(w_n-M)^2}{n^6}\\
&&+\frac{(\ell_n-\ell_n^{'})(\ell_n-\ell_n^{'}-1)}{2}\cdot\frac{(n-\lfloor nt\rfloor)^2\lfloor ns\rfloor^2(w_n-M)^2}{n^6}\\
&&+(k_n-\ell_n)\ell_n^{'}\cdot\frac{\lfloor nt\rfloor\lfloor ns\rfloor(n-\lfloor nt\rfloor)(n-\lfloor ns\rfloor)(w_n-M)^2}{n^6}\\
&&-(k_n-\ell_n)(\ell_n-\ell_n^{'})\cdot\frac{\lfloor nt\rfloor(n-\lfloor nt\rfloor)\lfloor ns\rfloor^2(w_n-M)^2}{n^6}\\
&&+\frac{(k_n-\ell_n)(k_n-\ell_n-1)}{2}\cdot\frac{\lfloor nt\rfloor^2\lfloor ns\rfloor^2(w_n-M)^2}{n^6}\bigg\}\frac{\tr(\O^2_{w_n-M,M})}{p}\\
&\rightarrow& \frac{(1-t)^2s^2}{2}\lim_{p\rightarrow \infty}\frac{\tr(\O^2_{w_n-M,M})}{p}.
\end{eqnarray*}

If at least one point of $\lfloor ns \rfloor$ and $\lfloor nt \rfloor$ falls within a large block, the discussion is similar to the above.
Summarizing the above results, for $s<t$, we have
\begin{align*}
\lim_{n\rightarrow \infty}\omega^{-2}\sum\limits_{j=2}^{k_n}\E(\tilde{V}_{nj,ts}^2)= a^2\frac{t^2(1-t)^2}{4}+b^2\frac{s^2(1-s)^2}{4}+2ab\frac{(1-t)^2s^2}{4}.
\end{align*}
Repeating the discussion in Proposition \ref{L6.1-P3}, the Lindeberg condition can be easily obtained. Now, we have
$$\left(W(\lfloor n t\rfloor)^{(G)}-\mu_{t}^{(G)},W(\lfloor n s\rfloor)^{(G)}-\mu_{s}^{(G)}\right)^\top\stackrel{d}{\rightarrow} N_2(\mathbf{0},\Xi),$$
where
\begin{align*}
\Xi= \omega^{2}
\begin{pmatrix}  (1-t)^2t^2& (1-t)^2s^2 \\(1-t)^2s^2 & (1-s)^2s^2 \end{pmatrix}.
\end{align*}
More than three points can be treated in the same way and therefore the finite-dimensional distributions of $W(\lfloor n t\rfloor)^{(G)}-\mu_{t}^{(G)}$ has been established.

\noindent{\textbf{Step 3.}}  Prove the tightness.
We also consider the following two cases, which has been discussed in Step 1.
\begin{itemize}
\item [(I)]
there exists $\ell_n$ such that $\ell_n$th block fall in the left side of the point $\lfloor n t\rfloor$ and $(\ell_n+1)$th block fall in the right side of the point $\lfloor n t\rfloor$, that is $b_{\ell_n}^{r}\leq \lfloor nt\rfloor< b_{\ell_n+1}^{l}$;
\item [(II)]
there exists $\ell_n$ such that the point $\lfloor n t\rfloor$ fall in the $\ell_n$th block, that is $b_{\ell_n}^{l}\leq \lfloor nt\rfloor< b_{\ell_n}^{r}$.
\end{itemize}
Here we discuss only case (I), the discussion of case (II) proceeds analogously.
In case (II), we define
\begin{align*}
\tilde{\bxi}_i=\frac{1}{w_n-M}\sum\limits_{\ell=(i-1)w_n+1}^{iw_n-M}\bxi_{\ell},
\end{align*}
and rewrite
\begin{eqnarray*}
\mathcal{H}_{1,t}&=&2g(\lfloor n t\rfloor)(w_n-M)^2\sum\limits_{j=2}^{k_n}\sum\limits_{i=1}^{j-1}\tilde{\eeta}_{i,t}^\top\tilde{\eeta}_{j,t},\\
&=&(1-t)^2A_{2,\ell_n}^{1,j-1}-t(1-t)A_{\ell_n+1,k_n}^{1,\ell_n}+t^2A_{\ell_n+1,k_n}^{\ell_n+1,j-1}\\
&=&(1-t)A_{2,\ell_n}^{1,j-1}-t(1-t)A_{2,k_n}^{1,j-1}+tA_{\ell_n+1,k_n}^{\ell_n+1,j-1},
\end{eqnarray*}
where
\begin{align*}
A_{2,\ell_n}^{1,j-1}:=&\frac{2(w_n-M)^2}{n\sqrt{p}}\sum\limits_{j=2}^{\ell_n}\sum\limits_{i=1}^{j-1}\tilde{\bxi}_{i}^\top\tilde{\bxi}_{j},\\
A_{\ell_n+1,k_n}^{1,\ell_n}:=&\frac{2(w_n-M)^2}{n\sqrt{p}}\sum\limits_{j=\ell_n+1}^{k_n}\sum\limits_{i=1}^{\ell_n}\tilde{\bxi}_{i}^\top\tilde{\bxi}_{j},\\
A_{\ell_n+1,k_n}^{\ell_n+1,j-1}:=&\frac{2(w_n-M)^2}{n\sqrt{p}}\sum\limits_{j=\ell_n+1}^{k_n}\sum\limits_{i=\ell_n+1}^{j-1}\tilde{\bxi}_{i}^\top\tilde{\bxi}_{j},
\end{align*}

We now prove the tightness of $(1-t)\omega^{-1}A_{2,\ell_n}^{1,j-1}, t(1-t)\omega^{-1}A_{2,k_n}^{1,j-1}, t\omega^{-1}A_{\ell_n+1,k_n}^{\ell_n+1,j-1}$.
In step 1, we have proved $\E\{(\omega^{-1}A_{2,k_n}^{1,j-1})^2\}\leq C$, where $C$ is a constant. Then, $t(1-t)A_{2,k_n}^{1,j-1}$ is tight.
Now, we resort to the tightness of $(1-t)A_{2,\ell_n}^{1,j-1}$. It suffices to prove that for any $\varsigma,\varphi>0$, there exists an integer $n_0\geq 1$ and $\iota\in(0,1)$, such that
\begin{align}
\label{Eq-Lem8.2-3}
\pr\Big(\mathop{\sup}\limits_{s\leq t\leq s+\iota}\Big|(1-t)A_{2,\ell_n}^{1,j-1}-(1-s)A_{2,\ell_n^{'}}^{1,j-1}\Big|\geq \omega\varsigma\Big)\leq \iota\varphi, \quad n\geq n_0,
\end{align}
for all $s\in[0,1]$, where $\ell_n^{'}$ satisfy $\ell_n^{'}w_n +1 \leq \lfloor ns\rfloor \leq \ell_n^{'}w_n +w_n - M$.
Since
\begin{eqnarray*}
&&\Big|(1-t)A_{2,\ell_n}^{1,j-1}-(1-s)A_{2,\ell_n^{'}}^{1,j-1}\Big|\\
&=&\Big|-(t-s)A_{2,\ell_n^{'}}^{1,j-1}+(1-t)(A_{2,\ell_n}^{1,j-1}-A_{2,\ell_n^{'}}^{1,j-1})\Big|\\
&\leq&(t-s)\Big|A_{2,\ell_n^{'}}^{1,j-1}\Big|+\Big|A_{2,\ell_n}^{1,j-1}-A_{2,\ell_n^{'}}^{1,j-1}\Big|,
\end{eqnarray*}
then inequality \eqref{Eq-Lem8.2-3} follows if the following inequalities hold for all $s \in [0,1]$:
\begin{align}
\pr\Big(\mathop{\sup}\limits_{s\leq t\leq s+\iota}(t-s)|A_{2,\ell_n^{'}}^{1,j-1}|\geq \omega\frac{\varsigma}{2}\Big)\leq \frac{\iota\varphi}{2}, \quad n\geq n_0,
\label{Eq6.2.1}
\end{align}
and
\begin{align}
\pr\Big(\mathop{\sup}\limits_{s\leq t\leq s+\iota}|A_{2,\ell_n}^{1,j-1}-A_{2,\ell_n^{'}}^{1,j-1}|\geq \omega\frac{\varsigma}{2}\Big)\leq \frac{\iota\varphi}{2}, \quad n\geq n_0.  \label{Eq6.2.2}
\end{align}
To be noticed that, inequality (\ref{Eq6.2.1}) can be obtained from (\ref{Eq6.2.2}), and by Theorem 8.4 in \cite{billingsley2013convergence}, inequality (\ref{Eq6.2.2}) reduces to the following statement: there exists a $\lambda_0>1$ and an integer $n_0$ such that for all $h_n$ satisfy
\begin{align*}
\pr\Big(\mathop{\max}\limits_{m_n\leq k_n}|A_{2,h_n+m_n}^{1,j-1}-A_{2,h_n}^{1,j-1}|\geq \omega\lambda_0\Big)\leq \frac{\varsigma}{\lambda_0^2},
\end{align*}
for all $k_n\geq n_0$.
Since $\{\tilde{\bxi}_i\}_{i=1}^{k_n}$ are i.i.d., the above inequality can reduced to
\begin{align*}
\pr\Big(\mathop{\max}\limits_{m_n\leq k_n-1}|A_{2,m_n}^{1,j-1}|\geq \omega\lambda_0\Big)\leq \frac{\varsigma}{\lambda_0^2}, \quad n\geq n_0.
\end{align*}
Since $\{A_{2,m_n}^{1,j-1}; 2 \leq m_n \leq k_n\}$ is a martingale with respect to $\sigma$-algebra $\{\mathcal{F}_{nm_n}, m=2,\dots,k_n\}$. According to the Doob's martingale inequality, we have
\begin{align*}
\pr\Big(\mathop{\max}\limits_{m_n\leq k_n}|A_{2,m_n}^{1,j-1}|\geq \omega\lambda_0\Big)\leq \frac{\E(|A_{1,m_n}^{1,j-1}|^4)}{\omega^4\lambda_0^4}\leq \frac{\varsigma}{\lambda_0^2}
\end{align*}
for sufficiently large $\lambda_0$ because $\omega^{-4}\E(|A_{2,m_n}^{1,j-1}|^4)\leq C$. Thus, $(1-t)\omega^{-1}A_{2,\ell_n}^{1,j-1}$ is tight. The prove technique of tightness of $t\omega^{-1}A_{\ell_n+1,k_n}^{\ell_n+1,j-1}$ is similar, the details are omitted here.
Therefore, the prove of Lemma \ref{Th1-L1} is complete by the above three main steps.
\end{proof}

\begin{lemma}
\label{Th1-L2}
Under Assumptions~\ref{ass:C1}--\ref{ass:C3} and $H_0$, for any $t \in[0,1]$, we have
$$
\omega^{-1}\{W(\lfloor n t\rfloor)^{(NG)}-W(\lfloor n t\rfloor)^{(G)}\}=o_p(1).
$$
\end{lemma}

\begin{proof}
For ease of presentation, we remark $B_{ij,t}, D_{ij,t}, F_{t}$ in the proof of Lemma \ref{Th1-L1} as $B^{(G)}_{ij,t}, D^{(G)}_{ij,t}, F_{t}^{(G)}$. For non-Gaussian sequence $\{\boldsymbol\varepsilon_i\}_{i=1}^{n}$, we can define $B^{(NG)}_{ij,t}, D^{(NG)}_{ij,t}, F_{t}^{(NG)}$, similarly.
Let $\Delta \mathcal{H}_{i,t}= \mathcal{H}_{i,t}^{(NG)}- \mathcal{H}_{i,t}^{(G)}$, then $W(\lfloor n t\rfloor)^{(NG)}-W(\lfloor n t\rfloor)^{(G)}= \sum\nolimits_{i=1}^4 \Delta \mathcal{H}_{i,t}$.

\noindent{\textbf{Step 1.}} In this step, we will show that $\omega^{-1}\Delta \mathcal{H}_{1,t}=o(1)$.
Define $\bzeta_{i,t}:=a_{i,\lfloor nt\rfloor}\boldsymbol\varepsilon_{i}$, and
\begin{align*}
\tilde{\bzeta}_{i,t}:=\frac{1}{w_n-M}\sum\limits_{\ell=(i-1)w_n+1}^{iw_n-M}\bzeta_{\ell,t},
\end{align*}
we have
\begin{align*}
B_{ij,t}^{(G)}&=g(\lfloor n t\rfloor)(w_n-M)^2\tilde{\eeta}_{i,t}^\top\tilde{\eeta}_{j,t},\\
B_{ij,t}^{(NG)}&=g(\lfloor n t\rfloor)(w_n-M)^2\tilde{\bzeta}_{i,t}^\top\tilde{\bzeta}_{j,t}.
\end{align*}
Both $\{\tilde{\eeta}_{i,t}\}_{i=1}^{k_n}$ and $\{\tilde{\bzeta}_{i,t}\}_{i=1}^{k_n}$ are independent sequence with zero mean and $\var(\tilde{\eeta}_{i,t})=\var(\tilde{\bzeta}_{i,t})=(w_n-M)^{-1}\O_{w_n-M,M}$, where $\O_{w_n-M,M}=\sum\limits_{h\in \mathcal{M}}\{1-(w_n-M)^{-1}|h|\}\Gamma_M(h)$.
Define
\begin{align*}
\mathcal{W}(\tilde{\eeta}_{1,t},\dots,\tilde{\eeta}_{k_n,t}):=\frac{\mathcal{H}_{1,t}^{(G)}}{\sqrt{2p^{-1}\tr(\O_{n}^2)}} = 2\sum\limits_{1\leq i < j \leq k_n}\frac{B_{ij,t}^{(G)}}{\sqrt{2p^{-1}\tr(\O_{n}^2)}}&=2\sum\limits_{1\leq i < j\leq k_n}u(\tilde{\eeta}_{i,t},\tilde{\eeta}_{j,t}),
\end{align*}
where $u(\tilde{\eeta}_{i,t},\tilde{\eeta}_{j,t}):=g(\lfloor n t\rfloor)(w_n-M)^2\tilde{\eeta}_{i,t}^\top\tilde{\eeta}_{j,t}/\sqrt{2p^{-1}\tr(\O_{n}^2)}$. Since $\tilde{\eeta}_{i,t}$ and $\tilde{\bzeta}_{i,t}$ are linear combinations of $\bZ_i=(Z_{i1},\dots,Z_{ip})^\top$ and $\E(Z_{ij}^4)<\infty$. We have the following properties:
\begin{itemize}
\item [(i)] $\forall \boldsymbol{a} \in \mathcal{R}^p$,  $\E\{u(\tilde{\eeta}_{i,t},\boldsymbol{a})\}=\E\{u(\boldsymbol{a},\tilde{\eeta}_{i,t})\}=0$, $\E\{u(\tilde{\bzeta}_{i,t},\boldsymbol{a})\}=\E\{u(\boldsymbol{a},\tilde{\bzeta}_{i,t})\}=0$.
\item [(ii)] $\forall \boldsymbol{a}, \boldsymbol{b} \in \mathcal{R}^p$,
\begin{align*}
&\E\{u(\boldsymbol{a},\tilde{\eeta}_{i,t})u(\boldsymbol{b},\tilde{\eeta}_{i,t})\}=\E\{u(\boldsymbol{a},\tilde{\bzeta}_{i,t})u(\boldsymbol{b},\tilde{\bzeta}_{i,t})\},\\
&\E\{u(\boldsymbol{a},\tilde{\eeta}_{i,t})u(\tilde{\eeta}_{i,t},\boldsymbol{b})\}=\E\{u(\boldsymbol{a},\tilde{\bzeta}_{i,t})u(\tilde{\bzeta}_{i,t},\boldsymbol{b})\},\\
&\E\{u(\tilde{\eeta}_{i,t},\boldsymbol{a})u(\boldsymbol{b},\tilde{\eeta}_{i,t})\}=\E\{u(\tilde{\bzeta}_{i,t},\boldsymbol{a})u(\boldsymbol{b},\tilde{\bzeta}_{i,t})\}.
\end{align*}
\item [(iii)] Let $g_{ij,t}^2:=\E\{u(\tilde{\bzeta}_{i,t},\tilde{\bzeta}_{j,t})^2\}$, then
\begin{align*}
\max\left[\E\{u(\tilde{\bzeta}_{i,t},\tilde{\bzeta}_{j,t})^4\}, \E\{u(\tilde{\bzeta}_{i,t},\tilde{\eeta}_{j,t})^4\},\E\{u(\tilde{\eeta}_{i,t},\tilde{\eeta}_{j,t})^4\}\right]\leq \rho_0 g_{ij,t}^4 < \infty,
\end{align*}
where $\rho_0=\max(\tau_1,3)$.
\end{itemize}
It is known that a sequence of random variables $\{Z_n; n=1,\dots.\infty\}$ converges weakly to a random variable $Z$ if and only if for every $f\in\mathcal{C}_b^3(\mathbb{R})$, $\E\{f(Z_n)\}\rightarrow \E\{f(Z)\}$, see, \cite{pollard1984convergence}. It reduces to prove
\begin{align*}
\left\vert\E\{ f(\mathcal{W}(\tilde{\eeta}_{1,t},\dots,\tilde{\eeta}_{k_n,t}))\}-\E \{f(\mathcal{W}(\tilde{\bzeta}_{1,t},\dots,\tilde{\bzeta}_{k_n,t}))\}\right\vert\rightarrow 0.
\end{align*}
for every $f\in\mathcal{C}_b^3(\mathbb{R})$ as $n,p \rightarrow \infty$. For $i=1,\dots, k_n$, we define
\begin{align*}
\mathcal{W}_i&:=\mathcal{W}(\tilde{\eeta}_{1,t},\dots,\tilde{\eeta}_{i-1,t},\tilde{\bzeta}_{i,t},\dots,\tilde{\bzeta}_{k_n,t}),\\
\mathcal{W}_{i,0}&:=\sum\limits_{1\leq \ell_1<\ell_2 \leq i-1}u(\tilde{\eeta}_{\ell_1,t},\tilde{\eeta}_{\ell_2,t})
+\sum\limits_{i+1\leq \ell_1<\ell_2 \leq k_n}u(\tilde{\bzeta}_{\ell_1,t},\tilde{\bzeta}_{\ell_2,t})
+\sum\limits_{1\leq \ell_1<\leq i-1,i+1\leq \ell_2 \leq k_n}u(\tilde{\eeta}_{\ell_1,t},\tilde{\eeta}_{\bzeta_2,t}).
\end{align*}
Thus,
\begin{align*}
\left|\E\{ f(\mathcal{W}(\tilde{\eeta}_{1,t},\dots,\tilde{\eeta}_{k_n,t}))\}-\E \{f(\mathcal{W}(\tilde{\bzeta}_{1,t},\dots,\tilde{\bzeta}_{k_n,t}))\}\right| \leq \sum\limits_{i=1}^{k_n}\left|\E \{f(\mathcal{W}_i)\}-\E \{f(\mathcal{W}_{i+1})\}\right|.
\end{align*}
According to Taylor's expansion, we have
\begin{align*}
f(\mathcal{W}_i)-f(\mathcal{W}_{i,0})&=\sum\limits_{m=1}^{2}\frac{1}{m!}f^{(m)}(\mathcal{W}_{i,0})(\mathcal{W}_{i}-\mathcal{W}_{i,0})^m+O(|\mathcal{W}_{i}-\mathcal{W}_{i,0}|^3),\\
f(\mathcal{W}_{i+1})-f(\mathcal{W}_{i,0})&=\sum\limits_{m=1}^{2}\frac{1}{m!}f^{(m)}(\mathcal{W}_{i,0})(\mathcal{W}_{i+1}-\mathcal{W}_{i,0})^m+O(|\mathcal{W}_{i+1}-\mathcal{W}_{i,0}|^3).
\end{align*}
Based on the properties of $u(\tilde{\bzeta}_{i,t},\tilde{\eeta}_{j,t})$, we have
\begin{align*}
&\E\left(\mathcal{W}_i-\mathcal{W}_{i,0}|\tilde{\eeta}_{1,t},\dots,\tilde{\eeta}_{i-1,t},\tilde{\bzeta}_{i+1,t},\dots,\tilde{\bzeta}_{k_n,t}\right)\\
=&\sum\limits_{m_1=1}^{i-1}\E\left\{u(\tilde{\eeta}_{m_1,t},\tilde{\bzeta}_{i,t})|\tilde{\eeta}_{m_1,t}\right\}+\sum\limits_{m_2=i+1}^{k_n}\E\left\{u(\tilde{\bzeta}_{i,t},\tilde{\bzeta}_{{m_2},t})|\tilde{\bzeta}_{m_2,t}\right\}=0,\\
&\E\left(\mathcal{W}_{i+1}-\mathcal{W}_{i,0}|\tilde{\eeta}_{1,t},\dots,\tilde{\eeta}_{i-1,t},\tilde{\bzeta}_{i+1,t},\dots,\tilde{\bzeta}_{k_n,t}\right)\\
=&\sum\limits_{m_1=1}^{i-1}\E\left\{u(\tilde{\eeta}_{m_1,t},\tilde{\eeta}_{i,t})|\tilde{\eeta}_{m_1,t}\right\}+\sum\limits_{m_2=i+1}^{k_n}\E\left\{u(\tilde{\eeta}_{i,t},\tilde{\bzeta}_{{m_2},t})|\tilde{\bzeta}_{m_2,t}\right\}=0,\\
&\E\left\{(\mathcal{W}_{i}-\mathcal{W}_{i,0})^2|\tilde{\eeta}_{1,t},\dots,\tilde{\eeta}_{i-1,t},\tilde{\bzeta}_{i+1,t},\dots,\tilde{\bzeta}_{k_n,t}\right\}\\
=&\E\left\{(\mathcal{W}_{i+1}-\mathcal{W}_{i,0})^2|\tilde{\eeta}_1,\dots,\tilde{\eeta}_{i-1,t},\tilde{\bzeta}_{i+1,t},\dots,\tilde{\bzeta}_{k_n,t}\right\}.
\end{align*}
Thus,
\begin{align*}
\E\bigg\{\sum\limits_{m=1}^{2}\frac{1}{m!}f^{(m)}(\mathcal{W}_{i,0})(\mathcal{W}_{i}-\mathcal{W}_{i,0})^m\bigg\}
=\E\bigg\{\sum\limits_{m=1}^{2}\frac{1}{m!}f^{(m)}(\mathcal{W}_{i,0})(\mathcal{W}_{i+1}-\mathcal{W}_{i,0})^m\bigg\}.
\end{align*}
It must exists a positive constant $C$ such that
\begin{align*}
&~~~~\left\vert\E\{ f(\mathcal{W}(\tilde{\eeta}_{1,t},\dots,\tilde{\eeta}_{k_n,t}))\}-\E \{f(\mathcal{W}(\tilde{\bzeta}_{1,t},\dots,\tilde{\bzeta}_{k_n,t}))\}\right\vert\\
&\leq C\sum\limits_{i=1}^{k_n}\left(\E|\mathcal{W}_{i}-\mathcal{W}_{i,0}|^3+\E|\mathcal{W}_{i+1}-\mathcal{W}_{i,0}|^3\right)\\
&\leq C\sum\limits_{i=1}^{k_n}\left\{(\E|\mathcal{W}_{i}-\mathcal{W}_{i,0}|^4)^{3/4}+(\E|\mathcal{W}_{i+1}-\mathcal{W}_{i,0}|^4)^{3/4}\right\}.
\end{align*}
Now, we consider $\E|\mathcal{W}_{i}-\mathcal{W}_{i,0}|^4$ and $\E|\mathcal{W}_{i+1}-\mathcal{W}_{i,0}|^4$. The discussions of these two terms are similar.
We have
\begin{eqnarray*}
&&\E|\mathcal{W}_{i}-\mathcal{W}_{i,0}|^4\\
&=& \sum\limits_{\ell=1}^{i-1}\E\left\{u(\tilde{\eeta}_{\ell,t},\tilde{\bzeta}_{i,t})^4\right\}+\sum\limits_{\ell=i+1}^{k_n}\E\left\{u(\tilde{\bzeta}_{i,t},\tilde{\bzeta}_{\ell,t})^4\right\}+
6\sum\limits_{\ell_1=1}^{i-1}\sum\limits_{\ell_1=2}^{i-1}\E\left\{u(\tilde{\eeta}_{\ell_1,t},\tilde{\bzeta}_{i,t})^2u(\tilde{\eeta}_{\ell_2,t},\tilde{\bzeta}_{i,t})^2\right\}\\
&&+6\sum\limits_{\ell_1=i+1}^{k_n}\sum\limits_{\ell_2=i+1}^{k_n}\E\left\{u(\tilde{\bzeta}_{i,t},\tilde{\bzeta}_{\ell_1,t})^2u(\tilde{\bzeta}_{i,t},\tilde{\bzeta}_{\ell_2,t})^2\right\}+
6\sum\limits_{\ell_1=1}^{i-1}\sum\limits_{\ell_2=i+1}^{k_n}\E\left\{u(\tilde{\eeta}_{\ell_1,t},\tilde{\bzeta}_{i,t})^2u(\tilde{\bzeta}_{i,t},\tilde{\bzeta}_{\ell_2,t})^2\right\}\\
&\leq& \rho_0\bigg(\sum\limits_{\ell=1}^{i-1}g_{\ell i,t}^4+\sum\limits_{\ell=i+1}^{k_n}g_{i\ell,t}^4+6\sum\limits_{\ell_1=1}^{i-1}\sum\limits_{\ell_1=2}^{i-1}g_{\ell_1 i,t}^2g_{\ell_2 i,t}^2+6\sum\limits_{\ell_1=i+1}^{k_n}\sum\limits_{\ell_2=i+1}^{k_n}g_{i\ell_1,t}^2g_{i\ell_2,t}^2\\
&&+6\sum\limits_{\ell_1=1}^{i-1}\sum\limits_{\ell_2=i+1}^{k_n}g_{\ell_1i,t}^2g_{i\ell_2,t}^2\bigg)\\
&\leq& 3\rho_0\bigg(\sum\limits_{\ell=1}^{i-1}g_{\ell i,t}^2+\sum\limits_{\ell=i+1}^{k_n}g_{i\ell,t}^2\bigg)^2,
\end{eqnarray*}
and similarly,
$$\E|\mathcal{W}_{i+1}-\mathcal{W}_{i,0}|^4\leq 3\rho_0\bigg(\sum\limits_{\ell=1}^{i-1}g_{\ell i,t}^2+\sum\limits_{\ell=i+1}^{k_n}g_{i\ell,t}^2\bigg)^2.$$
Thus,
\begin{align*}
\left\vert\E \left\{f(\mathcal{W}(\tilde{\eeta}_{1,t},\dots,\tilde{\eeta}_{k_n,t}))\right\}-\E \left\{f(\mathcal{W}(\tilde{\bzeta}_{1,t},\dots,\tilde{\bzeta}_{k_n,t}))\right\}\right\vert
&= O\bigg\{\sum\limits_{i=1}^{k_n}\bigg(\sum\limits_{\ell=1}^{i-1}g_{\ell i,t}^2+\sum\limits_{\ell=i+1}^{k_n}g_{i\ell,t}^2\bigg)^{3/2}\bigg\}.
\end{align*}
For any $i\neq j$, we have
\begin{align*}
g_{ij,t}^2&=\frac{\{g(\lfloor n t\rfloor)\}^2(w_n-M)^4\E\left\{(\tilde{\eeta}_{i,t}^\top\tilde{\eeta}_{j,t})^2\right\}}{2p^{-1}\tr(\O_{n}^2)}\\
&=\frac{\{g(\lfloor n t\rfloor)\}^2}{2p^{-1}\tr(\O_{n}^2)}\sum\limits_{\ell_1=(i-1)w_n+1}^{iw_n-M}\sum\limits_{\ell_2=(j-1)w_n+1}^{jw_n-M}\sum\limits_{\ell_3=(i-1)w_n+1}^{iw_n-M}\sum\limits_{\ell_4=(j-1)w_n+1}^{jw_n-M}\E(\eeta_{\ell_1,t}^\top\eeta_{\ell_2,t}\eeta_{\ell_3,t}^\top\eeta_{\ell_4,t})\\
&=\frac{\{g(\lfloor n t\rfloor)\}^2}{2p^{-1}\tr(\O_{n}^2)}\sum\limits_{\ell_1=(i-1)w_n+1}^{iw_n-M}\sum\limits_{\ell_2=(j-1)w_n+1}^{jw_n-M}\sum\limits_{\ell_3=(i-1)w_n+1}^{iw_n-M}\sum\limits_{\ell_4=(j-1)w_n+1}^{jw_n-M}a_{\ell_1,\lfloor nt \rfloor}a_{\ell_2,\lfloor nt \rfloor}a_{\ell_3,\lfloor nt \rfloor}a_{\ell_4,\lfloor nt \rfloor}\E(\bxi_{\ell_1}^\top\bxi_{\ell_2}\bxi_{\ell_3}^\top\bxi_{\ell_4})\\
&\leq \frac{Cn^2}{n^4\tr(\O_{n}^2)}\tr\bigg\{\bigg(\sum\limits_{\ell_1=(i-1)w_n+1}^{iw_n-M}\sum\limits_{\ell_2=(i-1)w_n+1}^{iw_n-M}\bGam_{M}(|\ell_1-\ell_2|)\bigg)^2\bigg\}\\
&= O\bigg\{\frac{(w_n-M)^2\tr(\O^2_{w_n-M,M})}{n^2\tr(\O_{n}^2)}\{1+o(1)\}\bigg\}.
\end{align*}
Thus,
\begin{align*}
&\left\vert\E \big\{f(\mathcal{W}(\tilde{\eeta}_{1,t},\dots,\tilde{\eeta}_{k_n,t}))\big\}-\E\big\{ f(\mathcal{W}(\tilde{\bzeta}_{1,t},\dots,\tilde{\bzeta}_{k_n,t}))\big\}\right\vert \\
=& O\bigg\{k_n\bigg(k_n\frac{(w_n-M)^2\tr(\O^2_{w_n-M,M})}{n^2\tr(\O_{n}^2)}\bigg)^{3/2}\bigg\}=o(1).
\end{align*}

\noindent{\textbf{Step 2.}} In this step, we will show that $\omega^{-1}\Delta \mathcal{H}_{2,t}=o_p(1)$, where
\begin{align*}
\Delta  \mathcal{H}_{2,t}=\sum\limits_{i=1}^{k_n}B^{(NG)}_{ii,t}-\sum\limits_{i=1}^{k_n}B^{(G)}_{ii,t}.
\end{align*}
In the step 1 of Lemma \ref{Th1-L1}, we have prove that $\omega^{-1}\sum\nolimits_{i=1}^{k_n}B_{ii,t}=o_p(1)$, thus, this statement is proved.

\noindent{\textbf{Step 3.}} In this step, we will show that $\omega^{-1} \Delta \mathcal{H}_{3,t}=o_p(1)$.
\begin{align*}
\Delta \mathcal{H}_{3,t}&=\sum\limits_{1\leq i, j \leq k_n}D^{(NG)}_{ij,t}-\sum\limits_{1\leq i, j \leq k_n}D^{(G)}_{ij,t}.
\end{align*}
Let
\begin{align*}
L_{ij,t}&:=\sum\limits_{\ell_1=iw_n-M+1}^{iw_n}\sum\limits_{\ell_2=(j-1)w_n+1}^{jw_n-M} F_{ij,t},\\
C_{ij,t}&:=\sum\limits_{\ell_1=iw_n-M+1}^{iw_n}\sum\limits_{\ell_2=jw_n-M+1}^{jw_n} F_{ij,t},
\end{align*}
Then, $D_{ij,t}=2L_{ij,t}+C_{ij,t}$. Furthermore, define
\begin{gather*}
\mathcal{L}_{1,t}:=\sum\limits_{i=1}^{k_n-2}\sum\limits_{j=i+2}^{k_n}L_{ij,t}, \quad
\mathcal{L}_{2,t}:=\sum\limits_{i=1}^{k_n}L_{ii,t}, \quad  \mathcal{L}_{3,t}=\sum\limits_{i=1}^{k_n}L_{i i+1,t}, \\
\mathcal{C}_{1,t}:=\sum\limits_{i=1}^{k_n-1}\sum\limits_{j=i+1}^{k_n}C_{ij,t}, \quad \mathcal{C}_{2,t}:=\sum\limits_{i=1}^{k_n}C_{ii,t}.
\end{gather*}
Then, $\mathcal{H}_{3,t}=\sum\nolimits_{1\leq i, j \leq k_n}D_{ij,t}=4\mathcal{L}_{1,t}+2\mathcal{L}_{2,t}+4\mathcal{L}_{3,t}+2\mathcal{C}_{1,t}+\mathcal{C}_{2,t}$. We will prove that $\var(\omega^{-1}\mathcal{L}_{i,t})=o(1)$ for $i=1,2,3$, and $\var(\omega^{-1}\mathcal{C}_{i,t})=o(1)$ for $i=1,2$. Recall
\begin{eqnarray*}
F_{ij,t}^{(NG)}:=g(\lfloor n t\rfloor) a_{i,\lfloor nt \rfloor}a_{j,\lfloor nt \rfloor}  \left[\boldsymbol\varepsilon_{i}^\top \boldsymbol\varepsilon_{j}-\tr\{\bGam_{M}(|j-i|)\}\right].
\end{eqnarray*}
We first consider $\var(\omega^{-1}\mathcal{L}^{(NG)}_{1,t})$.
Since $j_1\geq i_1+2$, which implies $\ell_2-\ell_1\geq M$,
$$\E\left\{F^{(NG)}_{\ell_1\ell_2,t}F^{(NG)}_{\ell_3\ell_4,t}\right\}=\{g(\lfloor n t\rfloor)\}^2 \prod\limits_{\nu=1}^{4}a_{\ell_\nu,\lfloor nt \rfloor}  \E(\boldsymbol\varepsilon_{\ell_1}^\top \boldsymbol\varepsilon_{\ell_2}\boldsymbol\varepsilon_{\ell_3}^\top \boldsymbol\varepsilon_{\ell_4})$$
is nonzero only if $i_2=i_1$ or $i_1=j_2$ or $i_1=j_2-1$.
If $i_2=i_1$, $\E(\boldsymbol\varepsilon_{\ell_1}^\top \boldsymbol\varepsilon_{\ell_2}\boldsymbol\varepsilon_{\ell_3}^\top \boldsymbol\varepsilon_{\ell_4})\neq 0$ only if $j_2=j_1$. If $i_1=j_2$ or $i_1=j_2-1$, then $i_2\leq j_2-2\leq i_1-1$, which implies $\ell_{(1)}-\ell_{(2)}=\min\{\ell_1-\ell_3,\ell_4-\ell_3\}>M$. Therefore, $\E(\boldsymbol\varepsilon_{\ell_1}^\top \boldsymbol\varepsilon_{\ell_2}\boldsymbol\varepsilon_{\ell_3}^\top \boldsymbol\varepsilon_{\ell_4})\neq 0$ only if $i_2=i_1$ and $j_2=j_1$. Then
\begin{align*}
\E\left\{\left(\mathcal{L}_{1,t}^{(NG)}\right)^2\right\}
&=\sum\limits_{i_1=1}^{k_n-2}\sum\limits_{j_1=i_1+2}^{k_n}\sum\limits_{i_2=1}^{k_n-2}\sum\limits_{j_2=i_2+2}^{k_n}\E\left(L^{NG}_{i_1j_1,t}L^{NG}_{i_2j_2,t}\right)\\
&=\{g(\lfloor n t\rfloor)\}^2 \sum\limits_{i=1}^{k_n-2}\sum\limits_{j=i+2}^{k_n}\E\bigg\{\bigg(\sum\limits_{\ell_1=iw_n-M+1}^{iw_n}\sum\limits_{\ell_2=(j-1)w_n+1}^{jw_n-M} a_{\ell_1,\lfloor nt \rfloor}a_{\ell_2,\lfloor nt \rfloor}\epsilon_{\ell_1}^\top \epsilon_{\ell_2}\bigg)^2\bigg\}\\
&=O\bigg\{\frac{(k_n-1)(k_n-2)}{n^2p}(w_n-M)M \tr(\O_{n,M}\O_{n,w_n-M})\bigg\}\\
&=O\bigg\{\frac{(k_n-1)(k_n-2)(w_n-M)M}{n^2p}\cdot\{\tr(\O^2_{n,M})\tr(\O^2_{n,w_n-M})\}^{1/2}\bigg\}\\
&=o\left\{p^{-1}\tr(\O^2_{n})\right\}.
\end{align*}
Then $\omega^{-1}\mathcal{L}_{1,t}^{(NG)}=o_p(1)$. Similarly, we have $\omega^{-1}\mathcal{C}_{1,1t}^{(NG)}=o_p(1)$.
Next, we consider
\begin{align*}
\E\{(\mathcal{L}_{2,t}^{(NG)})^2\}
&=\sum\limits_{i_1=1}^{k_n}\sum\limits_{i_2=1}^{k_n}\E\left(L^{(NG)}_{i_1i_1,t}L^{(NG)}_{i_2i_2,t}\right)\\
&=\sum\limits_{i_1=1}^{k_n}\sum\limits_{i_2=1}^{k_n}\E\bigg(\sum\limits_{\ell_1=i_1w_n-M+1}^{i_1w_n}\sum\limits_{\ell_2=(i_1-1)w_n+1}^{i_1w_n-M}\sum\limits_{\ell_3=i_2w_n-M+1}^{i_2w_n}\sum\limits_{\ell_4=(i_2-1)w_n+1}^{i_2w_n-M}F_{\ell_1\ell_2,t}F_{\ell_3\ell_4,t}\bigg)\\
&=\sum\limits_{i=1}^{k_n}\E\bigg\{\bigg(\sum\limits_{\ell_1=iw_n-M+1}^{iw_n}\sum\limits_{\ell_2=(i-1)w_n+1}^{iw_n-M} F_{\ell_1\ell_2,t}\bigg)^2\bigg\}\\
&=O\bigg\{\frac{k_nM^2(w_n-M)^2}{n^2p}\tau_2\tr(\O_{n,M}^2)\bigg\}=o\left\{p^{-1}\tr(\O^2_{n})\right\},
\end{align*}
where the forth equation holds based on Lemma \ref{Appendix-L2}. Then $\omega^{-1}\mathcal{L}_{2,t}^{(NG)}=o_p(1)$. Similarly, we have $\omega^{-1}\mathcal{L}_{3,t}^{(NG)}=o_p(1)$, $\omega^{-1}\mathcal{C}_{2,t}^{(NG)}=o_p(1)$.

\noindent{\textbf{Step 4.}} In this step, we will show that $\omega^{-1}\Delta \mathcal{H}_{4,t}=o_p(1)$.
Define
\begin{align*}
\mathcal{L}&:=\{(\ell_1,\ell_2,\ell_3,\ell_4)\in \{k_nw_n+1,\dots,n\}\times \{1,\dots,n\}\times \{k_nw_n+1,\dots,n\}\times \{1,\dots,n\}\},\\
\mathcal{L}_1&:=\{(\ell_1,\ell_2,\ell_3,\ell_4)\in\{1,\dots,n\}^4:|\ell_1-\ell_3|\geq 3M, \min(\ell_1,\ell_3)>k_nw_n\},\\
\mathcal{L}_2&:=\{(\ell_1,\ell_2,\ell_3,\ell_4)\in\{1,\dots,n\}^4:|\ell_1-\ell_3|\geq 3M, \min(\ell_1,\ell_3)>k_nw_n, \\
&\quad\quad\quad\quad\quad\quad\quad\quad\quad\quad\quad\quad\quad\quad \min(|\ell_1-\ell_2|,|\ell_1-\ell_4|,|\ell_3-\ell_2|,|\ell_3-\ell_4|)\leq M\}.
\end{align*}
For all $(\ell_1,\ell_2,\ell_3,\ell_4)\in \mathcal{L}_1\backslash \mathcal{L}_2$, we have $\E(F_{\ell_1\ell_2,t}F_{\ell_3\ell_4,t})=0$. Since $\mathcal{L}_2\subset \{k_nw_n+1,\dots, n\}^4$, we have $|\mathcal{L}_2|\leq (w_n+M)^4$ and $|\mathcal{L}\backslash\mathcal{L}_1|\leq 6Mnw_n^2$. Then
\begin{align*}
\E\left\{\left(F_t^{(NG)}\right)^2\right\}
&=\E\bigg\{\bigg(2\sum\limits_{\ell_1=1}^{k_nw_n}\sum\limits_{\ell_2=k_nw_n+1}^{n}F^{(NG)}_{\ell_1\ell_2,t}+\sum\limits_{\ell_1=k_nw_n+1}^{n}\sum\limits_{\ell_2=k_nw_n+1}^{n}F^{(NG)}_{\ell_1\ell_2,t}\bigg)^2\bigg\}\\
&\leq 4\sum\limits_{\ell_1=1}^{n}\sum\limits_{\ell_2=k_nw_n+1}^{n}\sum\limits_{\ell_3=1}^{n}\sum\limits_{\ell_4=k_nw_n+1}^{n}\left\vert\E\left(F^{(NG)}_{\ell_1\ell_2,t}F^{(NG)}_{\ell_3\ell_4,t}\right)\right\vert\\
&=4\bigg(\sum\limits_{(\ell_1,\ell_2,\ell_3,\ell_4)\in \mathcal{L}\backslash\mathcal{L}_1}+\sum\limits_{(\ell_1,\ell_2,\ell_3,\ell_4)\in \mathcal{L}_1\backslash\mathcal{L}_2}+\sum\limits_{(\ell_1,\ell_2,\ell_3,\ell_4)\in\mathcal{L}_2}\bigg)\left\vert\E\left(F^{(NG)}_{\ell_1\ell_2,t}F^{(NG)}_{\ell_3\ell_4,t}\right)\right\vert\\
&=O\bigg\{\frac{6M n w_n^2+(w_n+M)^4}{n^2p}\tau_2\tr(\O_{n,M}^2)\bigg\}=o\left\{p^{-1}\tr(\O^2_{n})\right\},
\end{align*}
with Lemma \ref{Appendix-L2}.
\end{proof}

\begin{lemma}
\label{Th1-L3}
Under Assumptions~\ref{ass:C1}--\ref{ass:C3} and $H_0$, for any $t \in[0,1]$, we have
$$
\omega^{-1}\{W(\lfloor n t\rfloor)-W(\lfloor n t\rfloor)^{(NG)}\}=o_p(1).
$$
\end{lemma}

\begin{proof}
According to the definitions, we have
\begin{align*}
&W(\lfloor n t\rfloor)-W(\lfloor n t\rfloor)^{(NG)}\\
=&g(\lfloor n t\rfloor)\sum\limits_{i=1}^n\sum\limits_{j=1}^n a_{i,\lfloor nt \rfloor}a_{j,\lfloor nt \rfloor}(\bX_i^{\top}\bX_j-\boldsymbol\varepsilon_i^{\top}\boldsymbol\varepsilon_j)\\
=&g(\lfloor n t\rfloor)\sum\limits_{i=1}^n\sum\limits_{j=1}^n a_{i,\lfloor nt \rfloor}a_{j,\lfloor nt \rfloor} (\bX_i-\boldsymbol\varepsilon_i)^{\top}(\bX_j-\boldsymbol\varepsilon_j)\\
&+ 2g(\lfloor n t\rfloor)\sum\limits_{i=1}^n\sum\limits_{j=1}^n a_{i,\lfloor nt \rfloor}a_{j,\lfloor nt \rfloor} \boldsymbol\varepsilon_i^{\top}(\bX_j-\boldsymbol\varepsilon_j).
\end{align*}
Define
\begin{gather*}
\mathcal{B}_{i_1j_1i_2j_2}:=\{(\ell_1,\ell_2,\ell_3,\ell_4); \,  i_1-M\leq \ell_1 \leq i_1, \ell_2\leq j_1-M-1, \\
 i_2-M\leq \ell_3 \leq i_2, \ell_4\leq j_2-M-1\},\\
\mathcal{B}_{1}:=\{(\ell_1,\ell_2,\ell_3,\ell_4); \, \ell_1 = \ell_2 \neq \ell_3 = \ell_4\}, \quad
\mathcal{B}_{2}:=\{(\ell_1,\ell_2,\ell_3,\ell_4); \, \ell_1 = \ell_3 \neq \ell_2 = \ell_4\},\\
\mathcal{B}_{3}:=\{(\ell_1,\ell_2,\ell_3,\ell_4); \, \ell_1 = \ell_4 \neq \ell_2 = \ell_3\}, \quad
\mathcal{B}_{4}:=\{(\ell_1,\ell_2,\ell_3,\ell_4); \, \ell_1 = \ell_4 = \ell_2 = \ell_3\}.
\end{gather*}
Considering
\begin{align*}
&\var\bigg\{\sum\limits_{i=1}^n\sum\limits_{j=1}^n \boldsymbol\varepsilon_i^{\top}(\bX_j-\boldsymbol\varepsilon_j)\bigg\}\\
=& \sum\limits_{\substack{1\leq i_1, j_1, i_2, j_2\leq n \\ (\ell_1,\ell_2,\ell_3,\ell_4)\in \mathcal{B}_{i_1j_1i_2j_2}}}
b_{i_1-\ell_1}b_{j_1-\ell_2}b_{i_2-\ell_3}b_{j_2-\ell_4}\E\big[\{\bZ_{\ell_1}^\top\bms \bZ_{\ell_2}-\mathbb{I}(\ell_1=\ell_2)\tr(\bms)\}\{\bZ_{\ell_3}^\top\bms \bZ_{\ell_4}-\mathbb{I}(\ell_3=\ell_4)\tr(\bms)\}\big]\\
=& \sum\limits_{\substack{1\leq i_1, j_1, i_2, j_2\leq n \\ (\ell_1,\ell_2,\ell_3,\ell_4)\in \mathcal{B}_{i_1j_1i_2j_2}\cap\mathcal{B}_1}}
b_{i_1-\ell_1}b_{j_1-\ell_1}b_{i_2-\ell_3}b_{j_2-\ell_3}\E\big[\{\bZ_{\ell_1}^\top\bms \bZ_{\ell_1}-\tr(\bms)\}\{\bZ_{\ell_3}^\top\bms \bZ_{\ell_3}-\tr(\bms)\}\big]\\
&+ \sum\limits_{\substack{1\leq i_1, j_1, i_2, j_2\leq n \\ (\ell_1,\ell_2,\ell_3,\ell_4)\in \mathcal{B}_{i_1j_1i_2j_2}\cap\mathcal{B}_2}}
b_{i_1-\ell_1}b_{j_1-\ell_2}b_{i_2-\ell_1}b_{j_2-\ell_2}\E\big\{(\bZ_{\ell_1}^\top\bms \bZ_{\ell_2})^2\}\\
&+ \sum\limits_{\substack{1\leq i_1, j_1, i_2, j_2\leq n \\ (\ell_1,\ell_2,\ell_3,\ell_4)\in \mathcal{B}_{i_1j_1i_2j_2}\cap\mathcal{B}_3}}
b_{i_1-\ell_1}b_{j_1-\ell_2}b_{i_2-\ell_2}b_{j_2-\ell_1}\E\big\{(\bZ_{\ell_1}^\top\bms\bZ_{\ell_2})(\bZ_{\ell_2}^\top\bms\bZ_{\ell_1})\big\}\\
&+ \sum\limits_{\substack{1\leq i_1, j_1, i_2, j_2\leq n \\ (\ell_1,\ell_2,\ell_3,\ell_4)\in \mathcal{B}_{i_1j_1i_2j_2}\cap\mathcal{B}_4}}
b_{i_1-\ell_1}b_{j_1-\ell_1}b_{i_2-\ell_1}b_{j_2-\ell_1}\E\big[\{\bZ_{\ell_1}^\top\bms\bZ_{\ell_1}-\tr(\bms)\}^2\big]\\
\leq& \sum\limits_{\substack{1\leq i_1, j_1, i_2, j_2\leq n \\ (\ell_1,\ell_2,\ell_3,\ell_4)\in \mathcal{B}_{i_1j_1i_2j_2}\cap\mathcal{B}_2}}
\lvert b_{i_1-\ell_1}b_{j_1-\ell_2}b_{i_2-\ell_1}b_{j_2-\ell_2}\rvert \tr(\bms^2)\\
&+ \sum\limits_{\substack{1\leq i_1, j_1, i_2, j_2\leq n \\ (\ell_1,\ell_2,\ell_3,\ell_4)\in \mathcal{B}_{i_1j_1i_2j_2}\cap\mathcal{B}_3}}
\lvert b_{i_1-\ell_1}b_{j_1-\ell_2}b_{i_2-\ell_2}b_{j_2-\ell_1}\rvert \tr(\bms^2)\\
&+ \sum\limits_{\substack{1\leq i_1, j_1, i_2, j_2\leq n \\ (\ell_1,\ell_2,\ell_3,\ell_4)\in \mathcal{B}_{i_1j_1i_2j_2}\cap\mathcal{B}_4}}
\lvert b_{i_1-\ell_1}b_{j_1-\ell_1}b_{i_2-\ell_1}b_{j_2-\ell_1}\rvert \tau_1\tr^2(\bms)\\
\leq& n^2\bigg\{\bigg(\sum\limits_{i_1,i_2\geq 0}|b_{i_1}b_{i_2}|\bigg)\bigg(\sum\limits_{j_1,j_2\geq M}|b_{j_1}b_{j_2}|\bigg)
+\bigg(\sum\limits_{i_1\geq 0, j_2\geq M}|b_{i_1}b_{j_2}|\bigg)\bigg(\sum\limits_{i_2\geq 0, j_1\geq M}|b_{i_2}b_{j_1}|\bigg)
\bigg\}\tr(\bms^2)\\
&+\tau_1np\sum\limits_{i_1,i_2\geq 0, j_1,j_2\geq M}\lvert b_{i_1}b_{i_2}b_{j_1}b_{j_2} \rvert\tr(\bms^2).
\end{align*}
By Lemma \ref{Appendix-L2} and Lemma \ref{Appendix-L3}, we have
$$g(\lfloor n t\rfloor)\sum\limits_{i=1}^n\sum\limits_{j=1}^n a_{i,\lfloor nt \rfloor}a_{j,\lfloor nt \rfloor} \boldsymbol\varepsilon_i^{\top}(\bX_j-\boldsymbol\varepsilon_j)=o\{p^{-1}\tr(\O^2_{n})\}.$$
Similarly, we can prove
$$g(\lfloor n t\rfloor)\sum\limits_{i=1}^n\sum\limits_{j=1}^n a_{i,\lfloor nt \rfloor}a_{j,\lfloor nt \rfloor} (\bX_i-\boldsymbol\varepsilon_i)^{\top}(\bX_j-\boldsymbol\varepsilon_j)=o\{p^{-1}\tr(\O^2_{n})\}.$$
This completes Lemma \ref{Th1-L3}.
\end{proof}

\begin{lemma}
\label{Th1-L4}
Under Assumptions~\ref{ass:C1}--\ref{ass:C3}, if $M=\lceil(n\wedge p)^{1/8}\rceil$, for any $t \in[0,1]$, we have
$$
\mu_{\lfloor n t\rfloor}^{(G)}-\mu_{M,\lfloor n t\rfloor}=o(\omega).
$$
\end{lemma}

\begin{proof}
Rewrite
$$
\mu_{\lfloor n t\rfloor}^{(G)} = g(\lfloor n t\rfloor)\sum\limits_{h=0}^{M}\sum\limits_{i=1}^{n-h}\{2-\ind{h=0}\}a_{i,\lfloor nt \rfloor}a_{i+h,\lfloor nt \rfloor}\tr\{\bGam_{M}(h)\}.
$$
We have
\begin{eqnarray*}
\left\vert\mu_{\lfloor n t\rfloor}^{(G)}-\mu_{M,\lfloor n t\rfloor}\right\vert&=&g(\lfloor n t\rfloor)\Bigg|\sum\limits_{h=0}^{M}\sum\limits_{i=1}^{n-h}\{2-\ind{h=0}\}a_{i,\lfloor nt \rfloor}a_{i+h,\lfloor nt \rfloor}\tr\{\bGam_{M}(h)\}-\tr\{\bGam(h)\}\Bigg|\\
&=&g(\lfloor n t\rfloor)\Bigg|\sum\limits_{h=0}^{M}\sum\limits_{i=1}^{n-h}\{2-\ind{h=0}\}a_{i,\lfloor nt \rfloor}a_{i+h,\lfloor nt \rfloor}\sum\limits_{\ell=M-h+1}^{\infty}b_{\ell}b_{\ell+h}\tr(\bms)\Bigg|\\
&\lesssim& \frac{n^2\tr(\bms)}{n^2\sqrt{p}}\bigg(\sum\limits_{\ell=0}^{\infty}|b_{\ell}|\bigg)\bigg(\sum\limits_{\ell=M+1}^{\infty}|b_{\ell}|\bigg)=o\left\{\sqrt{p^{-1}\tr(\O_n^2)}\right\}
\end{eqnarray*}
due to the fact $\sum\limits_{\ell=M}^{\infty}|b_{\ell}|=o(M^{-4})$ with Lemma \ref{Appendix-L3} by Assumption~\ref{ass:C1}--(ii).
\end{proof}

\begin{lemma}
\label{Th1-L5}
Under Assumptions~\ref{ass:C1}--\ref{ass:C3}, if $M=\lceil(n\wedge p)^{1/8}\rceil$, for any $t \in[0,1]$, we have
$$
p^{-1}\tr(\O_{n,M}^2)=p^{-1}\tr(\O_{n}^2)\{1+o(1)\}.
$$
\end{lemma}
\begin{proof}
We have
\begin{eqnarray*}
\O_{n}-\O_{n,M}=\bigg[c_0+2\sum\limits_{h=1}^{n}\bigg(1-\frac{h}{n}\bigg)c_h-\bigg\{c_{0,M}+2\sum\limits_{h=1}^{M}\bigg(1-\frac{h}{n}\bigg)c_{h,M}\bigg\}\bigg]\bms:=C_2\bms,
\end{eqnarray*}
and by Lemma~\ref{Appendix-L3}, we have
\begin{eqnarray*}
|C_2|
&\leq& 2\bigg|\sum\limits_{h=0}^{n}\bigg(1-\frac{h}{n}\bigg)c_h-\sum\limits_{h=0}^{M}\bigg(1-\frac{h}{n}\bigg)c_{h,M}\bigg|\\
&=& 2\bigg|\sum\limits_{h=M+1}^{n}\bigg(1-\frac{h}{n}\bigg)\sum\limits_{\ell=0}^{\infty}b_{\ell}b_{\ell+h}+\sum\limits_{h=0}^{M}\bigg(1-\frac{h}{n}\bigg)\sum\limits_{\ell=M-h+1}^{\infty}b_{\ell}b_{\ell+h}\bigg|\\
&\lesssim& 4\bigg(\sum\limits_{\ell=0}^{\infty}|b_{\ell}|\bigg)\bigg(\sum\limits_{\ell=M+1}^{\infty}|b_{\ell}|\bigg)=o(M^{-4}).
\end{eqnarray*}
By
\begin{eqnarray*}
p^{-1}\tr(\O_{n,M}^2)-p^{-1}\tr(\O_{n}^2)=2p^{-1}\tr\left\{\O_{n,M}(\O_{n,M}-\O_{n})\right\}-p^{-1}\tr\left\{(\O_{n,M}-\O_{n})^2\right\}.
\end{eqnarray*}
then,
$$
p^{-1}\tr(\O_{n,M}^2)=p^{-1}\tr(\O_{n}^2)\{1+o(1)\}.
$$
\end{proof}

\begin{lemma}
\label{Th1-L6}
Under Assumptions~\ref{ass:C1}--\ref{ass:C3}, if $M=\lceil(n\wedge p)^{1/8}\rceil$, for any $t \in[0,1]$, we have
$$
p^{-1}\tr(\O_n^2)=p^{-1}\tr(\O^2)\{1+o(1)\}.
$$
\end{lemma}
\begin{proof}
    We have
    \begin{equation*}
    \begin{aligned}
                \O-\O_n=&n^{-1}\sum_{h=1}^{n-1}h\bGam(h)+\sum_{h=n}^\infty \bGam(h)\\
                =&2\left(n^{-1}\sum_{h=1}^{n-1}\sum_{\ell=0}^\infty h b_{\ell}b_{\ell+h}+\sum_{h=n}^\infty \sum_{\ell=0}^\infty b_{\ell}b_{\ell+h}\right)\bms:=C_3\bms,
    \end{aligned}
    \end{equation*}
    and by Lemma~\ref{Appendix-L3}, we have
    \begin{equation*}
        \begin{aligned}
            \vert C_3\vert\leq 2\left\vert n^{-1}\sum_{h=1}^{n-1}h\sum_{\ell=h}^{\infty}b_{\ell}+\sum_{h=n}^\infty\sum_{\ell=h}^\infty b_{\ell}\right\vert=O(n^{-1})+o(n^{-3})=o(1).
        \end{aligned}
    \end{equation*}

   It implies that,
   \begin{equation*}
       \|\O_n-\O\|_F\leq o(1)\|\O\|_F,
   \end{equation*}
where $\|\cdot\|_F$ denotes the Frobenius norm for a matrix.

   Then, we have
\begin{equation*}
    \begin{aligned}
        \left\vert \|\O\|_F^2-\|\O_n\|_F^2\right\vert= \left\vert (\|\O\|_F-\|\O_n\|_F)(\|\O\|_F+\|\O_n\|_F)\right\vert\leq o(1)\|\O\|^2_F,
    \end{aligned}
\end{equation*}
   that is,
   $$
p^{-1}\tr(\O_n^2)=p^{-1}\tr(\O^2)\{1+o(1)\}.
$$
\end{proof}

\subsection{Proof of Theorem \ref{Th2}}
To prove Theorem~\ref{Th2}, it is equivalent to prove Lemmas~\ref{Th2-L1} and \ref{Th2-L2} below.

\begin{lemma}
\label{Th2-L1}
Under Assumptions~\ref{ass:C1}--\ref{ass:C3}, if $M=\lceil(n\wedge p)^{1/8}\rceil$, $p=o(n^{7/4})$ and $\|\bdelta\|^2 = o(n p^{1/2}/M^2)$, for any $t \in[0,1]$ , we have
$
\hat{\mu}_{M,\lfloor n t\rfloor}-\mu_{M,\lfloor n t\rfloor} = o_p(\omega)
$ as $(n,p)\rightarrow\infty$.
\end{lemma}
\begin{proof}
By definition, $\widehat{\tr\{\bGam(h)\}}$ can be decomposed as,
\begin{equation*}
    \begin{aligned}
        \widehat{\tr\{\bGam(h)\}}=&\frac{1}{2n}\sum_{t=1}^{n-M-2h-1}(\bepsilon_t-\bepsilon_{t+M+h+1})^\top (\bepsilon_{t+h}-\bepsilon_{t+M+2h+1})\\
        &-\frac{1}{2n}\sum_{t=\tau-2h-M-1}^{\tau-h}\bdelta^\top (\bepsilon_t-\bepsilon_{t+M+h+1})\\
        &-\frac{1}{2n}\sum_{t=\tau-M-1}^{\tau+h}\bdelta^\top (\bepsilon_t-\bepsilon_{t+M+h+1})\\
        &+\frac{M}{2n}\bdelta^\top\bdelta.
    \end{aligned}
\end{equation*}
Let $k = \lfloor nt \rfloor$, then
\begin{small}
    \begin{align}
        &\hat{\mu}_{M,k}\nonumber\\
        =&g(k)\sum\limits_{h=0}^{M}\sum\limits_{i=1}^{n-h}{\{2-\ind{h=0}\}}a_{i,k}a_{i+h,k}\left\{\frac{1}{2n}\!\!\sum_{t=1}^{n-M-2h-1}\!\!\!\!(\bepsilon_t-\bepsilon_{t+M+h+1})^\top (\bepsilon_{t+h}-\bepsilon_{t+M+2h+1})\right\}\label{eq:mu1}\\
        &-g(k)\sum\limits_{h=0}^{M}\sum\limits_{i=1}^{n-h}{\{2-\ind{h=0}\}}a_{i,k}a_{i+h,k}\left\{ \frac{1}{2n}\sum_{t=\tau-2h-M-1}^{\tau-h}\bdelta^\top (\bepsilon_t-\bepsilon_{t+M+h+1}) \right\}\label{eq:mu2}\\
        &-g(k)\sum\limits_{h=0}^{M}\sum\limits_{i=1}^{n-h}{\{2-\ind{h=0}\}}a_{i,k}a_{i+h,k}\left\{ \frac{1}{2n}\sum_{t=\tau-M-1}^{\tau+h}\bdelta^\top (\bepsilon_t-\bepsilon_{t+M+h+1}) \right\}\label{eq:mu3}\\
        &+g(k)\sum\limits_{h=0}^{M}\sum\limits_{i=1}^{n-h}{\{2-\ind{h=0}\}}a_{i,k}a_{i+h,k}\left( \frac{M}{2n}\bdelta^\top\bdelta \right).\label{eq:mu4}
    \end{align}
\end{small}

We first show the terms in Equations~\eqref{eq:mu2}--\eqref{eq:mu3} converges to zero with probability 1, and the Equation~\eqref{eq:mu4} converges to zero in Step 1. Then, we will show that the Equation~\eqref{eq:mu1} converges to $\mu_{M,k}$ with probability 1 in Step 2.

{\noindent\textbf{Step 1.}}
We first consider Equation~\eqref{eq:mu4}, by $\|\bdelta\|^2=o(n p^{1/2}/M^2)$, we have
\begin{equation*}
    \begin{aligned}
        &\left\vert g(k)\sum\limits_{h=0}^{M}\sum\limits_{i=1}^{n-h}{\{2-\ind{h=0}\}}a_{i,k}a_{i+h,k}\left( \frac{M}{2n}\bdelta^\top\bdelta \right)\right\vert\\
        =&\left\vert\frac{k^2(n-k)^2}{n^3\sqrt{p}}\left( \frac{M}{2n}\bdelta^\top\bdelta \right)\left\{\frac{n(2M+1)}{k(n-k)}-\frac{M(M+1)(n^2+k^2-n k)}{2k^2(n-k)^2}\right\}\right\vert\\
        \lesssim &\frac{M^2}{n p^{1/2}}\bdelta^\top\bdelta=o(1).
    \end{aligned}
\end{equation*}
We next consider Equation \eqref{eq:mu3}, and Equation \eqref{eq:mu2} can be handled in the same way. Notice that,
\begin{equation*}
    \begin{aligned}
        &\sum\limits_{h=0}^{M}\sum\limits_{i=1}^{n-h}{\{2-\ind{h=0}\}}a_{i,k}a_{i+h,k}\left\{ \frac{1}{2n}\sum_{t=\tau-M-1}^{\tau+h}\bdelta^\top (\bepsilon_t-\bepsilon_{t+M+h+1}) \right\}\\
        =&\frac{1}{2n}\sum_{h=0}^{M}\{2-\ind{h=0}\}\left\{  \frac{n}{k(n-k)}-h\frac{n^2+k^2-nk}{k^2(n-k)^2}\right\}\sum_{t=\tau-M-1}^{\tau+h}\bdelta^\top (\bepsilon_t-\bepsilon_{t+M+h+1}),\\
    \end{aligned}
\end{equation*}
and
\begin{equation*}
    \begin{aligned}
        &\E \left\{\bdelta^\top (\bepsilon_{t_1}-\bepsilon_{t_1+M+h_1+1})(\bepsilon_{t_2}-\bepsilon_{t_2+M+h_2+1})^\top\bdelta\right\}\\
        =&\bdelta^\top \left\{\bGam(\vert t_1-t_2\vert)+\bGam(\vert t_1-t_2+h_1-h_2\vert)-\bGam(\vert t_1-t_2-h_2\vert)-\bGam(\vert t_1-t_2+h_1\vert)\right\}^\top\bdelta\\
        \lesssim &\bdelta^\top \bms\bdelta
    \end{aligned}.
\end{equation*}
The expectation of Equation \eqref{eq:mu3} is 0, we then calculation the variance.
\begin{equation*}
    \begin{aligned}
&\var\left(g(k)\sum\limits_{h=0}^{M}\sum\limits_{i=1}^{n-h}{\{2-\ind{h=0}\}}a_{i,k}a_{i+h,k}\left\{ \frac{1}{2n}\sum_{t=\tau-M-1}^{\tau+h}\bdelta^\top (\bepsilon_t-\bepsilon_{t+M+h+1}) \right\}\right)\\
        =&\frac{k^4(n-k)^4}{4 n^8 p}\sum_{h_1=0}^M\sum_{h_2=0}^M\sum_{t_1=\tau-M-1}^{\tau+h_1}\sum_{t_2=\tau-M-1}^{\tau+h_2}\left\{\frac{1}{k}+\frac{1}{n-k}-h\frac{n^2+k^2-nk}{k^2(n-k)^2}\right\}\\
        &~~~~~~~~~~~~~~~~~~~~~~~~~~~~~~~~~~~~~\cdot\E \left\{\bdelta^\top (\bepsilon_{t_1}-\bepsilon_{t_1+M+h_1+1})(\bepsilon_{t_2}-\bepsilon_{t_2+M+h_2+1})^\top\bdelta\right\}\\
        \lesssim& \frac{k^4(n-k)^4}{4 n^8 p} M^2(2M)^2\left\{\frac{1}{k}+\frac{1}{n-k}+M\frac{n^2+k^2-nk}{k^2(n-k)^2}\right\}^2\bdelta^\top \bms\bdelta\\
        \lesssim &\frac{M^4}{n^2p}\bdelta^\top \bms\bdelta = o(1),
    \end{aligned}
\end{equation*}
that is, the Equation \eqref{eq:mu3} converges to 0 with probability 1.

{\noindent\textbf{Step 2.}}

Considering Equation~\eqref{eq:mu1}, we have
\begin{eqnarray*}
&&\E\left\{\frac{1}{2n}\sum\limits_{t=1}^{n-M-2h-1}(\bepsilon_{t+h}-\bepsilon_{t+M+2h+1})^{\T}(\bepsilon_t-\bepsilon_{t+M+h+1})\right\}\\
&=&\frac{1}{2n}\sum\limits_{t=1}^{n-M-2h-1}\E\bigg\{\sum\limits_{\ell_1,\ell_2=0}^{\infty}b_{\ell_1}b_{\ell_2}\bZ_{t+h_1-\ell_1}^{\T}\bms \bZ_{t-\ell_2} - \sum\limits_{\ell_1,\ell_2=0}^{\infty}b_{\ell_1}b_{\ell_2}\bZ_{t+h_1-\ell_1}^{\T}\bms \bZ_{t+M+h_1+1-\ell_2}\\
&& -\sum\limits_{\ell_1,\ell_2=0}^{\infty}b_{\ell_1}b_{\ell_2}\bZ_{t+M+2h_1+1-\ell_1}^{\T}\bms \bZ_{t-\ell_2}+\sum\limits_{\ell_1,\ell_2=0}^{\infty}b_{\ell_1}b_{\ell_2}\bZ_{t+M+2h_1+1-\ell_1}^{\T}\bms \bZ_{t+M+h_1+1-\ell_2}\bigg\}\\
&=&\frac{1}{2n}\sum\limits_{t=1}^{n-M-2h-1}\Big\{\Big(\sum\limits_{\ell_1=0}^{\infty}b_{\ell_1}b_{\ell_1+h}\Big)\tr(\bms)-\Big(\sum\limits_{\ell_1=0}^{\infty}b_{\ell_1}b_{\ell_1+M+2h+1}\Big)\tr(\bms)\\
&& -\Big(\sum\limits_{\ell_1=0}^{\infty}b_{\ell_1}b_{\ell_1+M+h+1}\Big)\tr(\bms) +
 \Big(\sum\limits_{\ell_1=0}^{\infty}b_{\ell_1}b_{\ell_1+h}\Big)\tr(\bms) \Big\}\\
&=&\frac{n-M-2h-1}{2n}\Big\{2\tr\{\bGam(h)\}-\Big(\sum\limits_{\ell_1=0}^{\infty}b_{\ell_1}b_{\ell_1+M+2h+1}\Big)\tr(\bms)
-\Big(\sum\limits_{\ell_1=0}^{\infty}b_{\ell_1}b_{\ell_1+M+h+1}\Big)\tr(\bms)\Big\}.
\end{eqnarray*}

By Assumption~\ref{ass:C1}-(ii), we have
\begin{eqnarray*}
&&\frac{g(k)}{\sqrt{2p^{-1}\tr(\O_n^2)}}\bigg|\sum\limits_{h=0}^{M}\sum\limits_{i=1}^{n-h}{\{2-\mathbb{I}(h=0)\}}a_{i,k}a_{i+h,k}\frac{n-M-2h-1}{2n}\Big(\sum\limits_{\ell_1=0}^{\infty}b_{\ell_1}b_{\ell_1+M+2h+1}\Big)\tr(\bms)\bigg|\\
&\lesssim&\frac{\tr(\bms)}{\sqrt{\tr(\O_n^2)}}\sum\limits_{h=0}^{M}\sum\limits_{\ell_1=0}^{\infty}|b_{\ell_1}||b_{\ell_1+M+2h+1}|
\lesssim \frac{\sqrt{p\tr(\bms^2)}}{\sqrt{\tr(\O_n^2)}}\sum\limits_{\ell_1=0}^{\infty}|b_{\ell_1}|\Big(\sum\limits_{h=\ell_1+M+1}^{\ell_1+3M+1}|b_{h}|\Big) \rightarrow 0,
\end{eqnarray*}
with $\sum\limits_{\ell=M}^{\infty}|b_{\ell}|=o(M^{-4})$ by Lemma \ref{Appendix-L3}, as $n,p \rightarrow \infty$.
Similarly, we have
\begin{eqnarray*}
&&\frac{g(k)}{\sqrt{2p^{-1}\tr(\O_n^2)}}\left|\sum\limits_{h=0}^{M}\sum\limits_{i=1}^{n-h}{\{2-\mathbb{I}(h=0)\}}a_{i,k}a_{i+h,k}\frac{n-M-2h-1}{2n}\Big(\sum\limits_{\ell_1=0}^{\infty}b_{\ell_1}b_{\ell_1+M+h+1}\Big)\tr(\bms)\right|  \rightarrow 0.
\end{eqnarray*}
Thus, as $n, p\rightarrow \infty$,
\begin{eqnarray*}
\E\left\{\frac{1}{2n}\sum_{t=1}^{n-M-2h-1}(\bepsilon_t-\bepsilon_{t+M+h+1})^\top (\bepsilon_{t+h}-\bepsilon_{t+M+2h+1})\right\} = \tr\{\bGam(h)\} + o(\omega).
\end{eqnarray*}

It remains to investigate
\begin{align}
\label{L8-6-0}
&\frac{1}{4n^2}\sum_{i=1}^{n-M-2h-1}\sum_{j=1}^{n-M-2h-1}\cov\bigg[(\bepsilon_{i+h_1}-\bepsilon_{i+M+2h_1+1})^{\T}(\bepsilon_i-\bepsilon_{i+M+h_1+1}), \nonumber\\
&(\bepsilon_{j+h_2}-\bepsilon_{j+M+2h_2+1})^{\T}(\bepsilon_j-\bepsilon_{j+M+h_2+1})\bigg],
\end{align}
for any $i, j, h_1, h_2$. Since
\begin{small}
\begin{align*}
&(\bepsilon_{i+h_1}-\bepsilon_{i+M+2h_1+1})^{\T}(\bepsilon_i-\bepsilon_{i+M+h_1+1})(\bepsilon_{j+h_2}-\bepsilon_{j+M+2h_2+1})^{\T}(\bepsilon_j-\bepsilon_{j+M+h_2+1})\nonumber\\
=&\Big(\bms^{1/2}\sum\limits_{\ell=0}^{\infty}b_{\ell}\bZ_{i+h_1-\ell} - \bms^{1/2}\sum\limits_{\ell=0}^{\infty}b_{\ell}\bZ_{i+M+2h_1+1-\ell}\Big)^{\T}\nonumber\Big(\bms^{1/2}\sum\limits_{\ell=0}^{\infty}b_{\ell}\bZ_{i-\ell} - \bms^{1/2}\sum\limits_{\ell=0}^{\infty}b_{\ell}\bZ_{i+M+h_1+1-\ell}\Big)\nonumber\\
&\times\Big(\bms^{1/2}\sum\limits_{\ell=0}^{\infty}b_{\ell}\bZ_{j+h_2-\ell} - \bms^{1/2}\sum\limits_{\ell=0}^{\infty}b_{\ell}\bZ_{j+M+2h_2+1-\ell}\Big)^{\T}\Big(\bms^{1/2}\sum\limits_{\ell=0}^{\infty}b_{\ell}\bZ_{j-\ell} - \bms^{1/2}\sum\limits_{\ell=0}^{\infty}b_{\ell}\bZ_{j+M+h_2+1-\ell}\Big),\nonumber
\end{align*}
\end{small}
we have
\begin{small}
\begin{eqnarray}
\label{L8-6-1}
&&\cov\left[(\bepsilon_{i+h_1}-\bepsilon_{i+M+2h_1+1})^{\T}(\bepsilon_i-\bepsilon_{i+M+h_1+1}), (\bepsilon_{j+h_2}-\bepsilon_{j+M+2h_2+1})^{\T}(\bepsilon_j-\bepsilon_{j+M+h_2+1})\right]\nonumber\\
&=&\Big(\sum\limits_{k_1\leq i+h_1,k_2\leq i,  k_3\leq j+h_2, k_4\leq j}b_{i+h_1-k_1}b_{i-k_2}b_{j+h_2-k_3}b_{j-k_4} \nonumber\\
&&-\sum\limits_{k_1\leq i+h_1,k_2\leq i, k_3\leq j+h_2, k_4\leq j+M+h_2+1}b_{i+h_1-k_1}b_{i-k_2}b_{j+h_2-k_3}b_{j+M+h_2+1-k_4} \nonumber\\
&&-\sum\limits_{k_1\leq i+h_1,k_2\leq i, k_3\leq j+M+2h_2+1, k_4\leq j}b_{i+h_1-k_1}b_{i-k_2}b_{j+M+2h_2+1-k_3}b_{j-k_4} \nonumber\\
&&+\sum\limits_{k_1\leq i+h_1,k_2\leq i, k_3\leq j+M+2h_2+1, k_4\leq j+M+h_2+1}b_{i+h_1-k_1}b_{i-k_2}b_{j+M+2h_2+1-k_3}b_{j+M+h_2+1-k_4} \nonumber\\
&&+\sum\limits_{k_1\leq i+h_1,k_2\leq i+M+h_1+1, k_3\leq j+h_2, k_4\leq j}b_{i+h_1-k_1}b_{i+M+h_1+1-k_2}b_{j+h_2-k_3}b_{j-k_4} \nonumber\\
&&-\sum\limits_{k_1\leq i+h_1,k_2\leq i+M+h_1+1, k_3\leq j+h_2, k_4\leq j+M+h_2+1}b_{i+h_1-k_1}b_{i+M+h_1+1-k_2}b_{j+h_2-k_3}b_{j+M+h_2+1-k_4} \nonumber\\
&&-\sum\limits_{k_1\leq i+h_1,k_2\leq i+M+h_1+1, k_3\leq j+M+2h_2+1, k_4\leq j}b_{i+h_1-k_1}b_{i+M+h_1+1-k_2}b_{j+M+2h_2+1-k_3}b_{j-k_4} \nonumber\\
&&+\sum\limits_{k_1\leq i+h_1,k_2\leq i+M+h_1+1, k_3\leq j+M+2h_2+1, k_4\leq j+M+h_2+1}b_{i+h_1-k_1}b_{i+M+h_1+1-k_2}b_{j+M+2h_2+1-k_3}b_{j+M+h_2+1-k_4} \nonumber\\
&&+\sum\limits_{k_1\leq i+M+2h_1+1,k_2\leq i, k_3\leq j+h_2, k_4\leq j}b_{i+M+2h_1+1-k_1}b_{i-k_2}b_{j+h_2-k_3}b_{j-k_4} \nonumber\\
&&-\sum\limits_{k_1\leq i+M+2h_1+1,k_2\leq i, k_3\leq j+h_2, k_4\leq j+M+h_2+1}b_{i+M+2h_1+1-k_1}b_{i-k_2}b_{j+h_2-k_3}b_{j+M+h_2+1-k_4} \nonumber\\
&&-\sum\limits_{k_1\leq i+M+2h_1+1,k_2\leq i, k_3\leq j+M+2h_2+1-k_3, k_4\leq j}b_{i+M+2h_1+1-k_1}b_{i-k_2}b_{j+M+2h_2+1-k_3}b_{j-k_4} \nonumber\\
&&+\sum\limits_{k_1\leq i+M+2h_1+1,k_2\leq i, k_3\leq j+M+2h_2+1, k_4\leq j+M+h_2+1}b_{i+M+2h_1+1-k_1}b_{i-k_2}b_{j+M+2h_2+1-k_3}b_{j+M+h_2+1-k_4} \nonumber\\
&&+\sum\limits_{k_1\leq i+M+2h_1+1,k_2\leq i+M+h_1+1, k_3\leq j+h_2, k_4\leq j}b_{i+M+2h_1+1-k_1}b_{i+M+h_1+1-k_2}b_{j+h_2-k_3}b_{j-k_4} \nonumber\\
&&-\sum\limits_{k_1\leq i+M+2h_1+1,k_2\leq i+M+h_1+1, k_3\leq j+h_2, k_4\leq j+M+h_2+1}b_{i+M+2h_1+1-k_1}b_{i+M+h_1+1-k_2}b_{j+h_2-k_3}b_{j+M+h_2+1-k_4} \nonumber\\
&&-\sum\limits_{k_1\leq i+M+2h_1+1,k_2\leq i+M+h_1+1, k_3\leq j+M+2h_2+1, k_4\leq j}b_{i+M+2h_1+1-k_1}b_{i+M+h_1+1-k_2}b_{j+M+2h_2+1-k_3}b_{j-k_4} \nonumber\\
&&+\sum\limits_{k_1\leq i+M+2h_1+1,k_2\leq i+M+h_1+1, k_3\leq j+M+2h_2+1, k_4\leq j+M+h_2+1}\!\!\!\!\!\!\!\!\!\!\!\!\!\!\!\!\!\!\!\!\!\!\!\!\!\!\!\!\!\!b_{i+M+2h_1+1-k_1}b_{i+M+h_1+1-k_2}b_{j+M+2h_2+1-k_3}b_{j+M+h_2+1-k_4}\Big) \nonumber\\
&&\times\E\left[\left\{\bZ_{k_1}^{\T}\bms \bZ_{k_2} - \mathbb{I}(k_1=k_2)\tr(\bms)\right\}\left\{\bZ_{k_3}^{\T}\bms \bZ_{k_4} - \mathbb{I}(k_3=k_4)\tr(\bms)\right\}\right].
\label{L8-6-2}
\end{eqnarray}
\end{small}

Denote $\rho_{k_1,k_2,k_3,k_4} := \E\left[\{\bZ_{k_1}^{\T}\bms \bZ_{k_2} - \mathbb{I}(k_1=k_2)\tr(\bms)\}\{\bZ_{k_3}^{\T}\bms \bZ_{k_4} - \mathbb{I}(k_3=k_4)\tr(\bms)\}\right]$, we have
\begin{eqnarray*}
\rho_{k_1,k_2,k_3,k_4}
\begin{cases}
= \tr(\bms^2), & k_1 = k_3 = r, \, k_2 = k_4 = s, \, r \neq s, \\
= \tr(\bms^2),, & k_1 = k_4 = r, \, k_2 = k_3 = s, \, r \neq s, \\
\leq \tau_1 \tr^2(\bms), & k_1 = k_2 = k_3 = k_4 = r, \\
= 0, & \text{otherwise}.
\end{cases}
\end{eqnarray*}
We only need to calculate the coefficients of $\rho_{k_1,k_2,k_3,k_4}$ in \eqref{L8-6-0}, denoted as $C_{k_1,k_2,k_3,k_4}$. We will consider the first term in \eqref{L8-6-2}, the analysis of the others are similar.
$$
\sum\limits_{i=1}^{n-M-2h_1-1}\sum\limits_{j=1}^{n-M-2h_2-1}|C_{r,s,r,s}| \leq \!\!\! \sum\limits_{i=1}^{n-M-2h_1-1}\sum\limits_{j=1}^{n-M-2h_2-1}\!\!\!\!\sum\limits_
{r\leq (i+h_1)\wedge (j+h_2), s \leq i\wedge  j}\!\!\!\!|b_{i+h_1-r}b_{i-s}b_{j+h_2-r}b_{j-s}|
$$
We rewrite $k_1 = i+h_1-r, k_2 = i-s, k_3 = j+h_2-r, k_4 = j-s$. Fix $(i, k_1, k_2, k_3, k_4)$, $(r, s, j)$ will be uniquely determined or not exist.
Thus,
$$
\sum\limits_{i=1}^{n-M-2h_1-1}\sum\limits_{j=1}^{n-M-2h_2-1}|C_{r,s,r,s}| \leq n\Big(\sum\limits_{\ell=0}^{\infty}|b_{\ell}|\Big)^4.
$$
Similarly,
$$
\sum\limits_{i=1}^{n-M-2h_1-1}\sum\limits_{j=1}^{n-M-2h_2-1}|C_{r,s,s,r}| \leq n\Big(\sum\limits_{\ell=0}^{\infty}|b_{\ell}|\Big)^4,
\quad
\sum\limits_{i=1}^{n-M-2h_1-1}\sum\limits_{j=1}^{n-M-2h_2-1}|C_{r,r,r,r}| \leq \Big(\sum\limits_{\ell=0}^{\infty}|b_{\ell}|\Big)^4.
$$
By $\tr^2(\bms) \lesssim p \tr(\bms^2)$ and $p=o(n^{7/4})$, we have
\begin{align*}
&\var\left(g(k)\sum\limits_{h=0}^{M}\sum\limits_{i=1}^{n-h}{\{2-\ind{h=0}\}}a_{i,k}a_{i+h,k}\left\{\frac{1}{2n}\!\!\sum_{t=1}^{n-M-2h-1}\!\!\!\!(\bepsilon_t-\bepsilon_{t+M+h+1})^\top (\bepsilon_{t+h}-\bepsilon_{t+M+2h+1})\right\}\right)\\
=& g^2(k)\sum\limits_{h_1=0}^{M}\sum\limits_{h_2=0}^{M}\sum\limits_{i_1=1}^{n-h_1}\sum\limits_{i_2=1}^{n-h_2}{\{2-\mathbb{I}(h_1=0)\}}{\{2-\mathbb{I}(h_2=0)\}}a_{i_1,k}a_{i_1+h_1,k}a_{i_2,k}a_{i_2+h_2,k}\times \eqref{L8-6-0}\\
=&O\bigg\{ \frac{M^2(n+p)}{n^{2}p}\tr(\bms^2)\bigg\} = o(\omega^2).
\end{align*}
\end{proof}

\begin{lemma}
\label{Th2-L2}
Under Assumptions~\ref{ass:C1}--\ref{ass:C3} and $H_0$, if $M=\lceil(n\wedge p)^{1/8}\rceil$ and $p=o(n^{7/4})$, for any $t \in[0,1]$, we have
$
\hat{\omega}/\omega\stackrel{p}{\rightarrow} 1
$ as $(n,p)\rightarrow\infty$.
\end{lemma}
\begin{proof}
Recall
\begin{align*}
\hat{\omega}=\bigg\{\frac{2}{p}\bigg(\widehat{\tr\{\bGam(0)\bGam(0)\}}+2\sum\limits_{h=1}^M \widehat{\tr\{\bGam(h)\bGam(0)\}}+2\sum\limits_{k=1}^M \widehat{\tr\{\bGam(0)\bGam(k)\}}
+4\sum\limits_{h,k=1}^M \widehat{\tr\{\bGam(h)\bGam(k)\}}\bigg)\bigg\}^{1/2},
\end{align*}
where
\begin{align*}
\widehat{\tr\{\bGam(h)\bGam(k)\}}=
\frac{\sum\limits_{t=1}^{[n/2]-M-2k-1}\sum\limits_{s=t+[n/2]}^{n-M-2k-1}
\breve{\bX}_{t,h,s,k}\breve{\bX}_{t+h,h,s+k,k}}{4(n-k-3/2[n/2]-M/2)([n/2]-M-2k-1)},
\end{align*}
with $\breve{\bX}_{f,h,g,k} = (\bX_{f}-\bX_{f+M+h+1})^{\T}(\bX_{g}-\bX_{g+M+k+1})$. Define $\breve{\bepsilon}_{f,h,g,k}=(\bepsilon_{f}-\bepsilon_{f+M+h+1})^{\T}(\bepsilon_{g}-\bepsilon_{g+M+k+1})$.
We decompose $\widehat{\tr\{\bGam(h_1)\bGam(h_2)\}} = \tilde{S}_{h_1,h_2} + R_{h_1,h_2}$,
where
\begin{align*}
&\tilde{S}_{h_1,h_2}=\sum\limits_{i_1=1}^{[n/2]-M-2h_2-1}\sum\limits_{i_2=h_1+[n/2]}^{n-M-2h_2-1}\frac{
\breve{\bepsilon}_{i_1,h_1,i_2,h_2}\breve{\bepsilon}_{i_1+h_1,h_1,i_2+h_2,h_2}}{4(n-h_2-3/2[n/2]-M/2)([n/2]-M-2h_2-1)},
\end{align*}
and $R_{h_1,h_2}:=\widehat{\tr\{\bGam(h_1)\bGam(h_2)\}}-\tilde{S}_{h_1,h_2}$. Following the similar analysis in Lemma~\ref{Th2-L1}, it is naturally to prove that $R_{h_1,h_2}=o_p(1)$, and the details are omitted. Now, we turn to investigate $\tilde{S}_{h_1,h_2}:=\tilde{S}_{h_1,h_2,H}+\tilde{S}_{h_1,h_2,L}$, where $\tilde{S}_{h_1,h_2,H}$ and $\tilde{S}_{h_1,h_2,L}$ represent the terms constructed by the high order and low order of $Z_{ij}$.
$\tilde{S}_{h_1,h_2} $ is the sum of 16 terms. Here, we only consider the first term here, the others can be analyzed similarly.
The first term in $\breve{\bepsilon}_{i_1,h_1,i_2,h_2}\breve{\bepsilon}_{i_1+h_1,h_1,i_2+h_2,h_2}$ is
\begin{align}
\label{L7-8-1}
\sum\limits_{k_1\leq i_1,k_2\leq i_2,  k_3\leq i_1+h_1, k_4\leq i_2+ h_2}b_{i_1-k_1}b_{i_2-k_2}b_{i_1+h_1-k_3}b_{i_2+h_2-k_4}\bZ_{k_1}^{\T}\bms \bZ_{k_2}\bZ_{k_3}^{\T}\bms \bZ_{k_4}.
\end{align}
The high order term of Equation~\eqref{L7-8-1} is
\begin{align*}
&\sum\limits_{j_1,j_2=1}^{p}\sigma_{j_1j_1}\sigma_{j_2j_2}\bigg\{
\bigg(
\sum\limits_{\substack{r\leq i_1\wedge i_2 \wedge (i_1+h_1) \\ k \leq i_2+h_2}}b_{i_1-r}b_{i_2-r}b_{i_1+h_1-r}b_{i_2+h_2-k} \\
&+\sum\limits_{\substack{r\leq i_1\wedge i_2 \wedge (i_2+h_2) \\ k \leq i_1+h_1}}b_{i_1-r}b_{i_2-r}b_{i_1+h_1-k}b_{i_2+h_2-r}\\
&+ \sum\limits_{\substack{r\leq i_1\wedge (i_1+h_1) \wedge (i_2 +h_2) \\ k \leq i_2}}b_{i_1-r}b_{i_2-k}b_{i_1+h_1-r}b_{i_2+h_2-r}
\bigg)Z_{rj_1}^3Z_{kj_2}\\
&+\sum\limits_{\substack{r\leq i_1\wedge i_2\wedge (i_1+h_1) \wedge (i_2 +h_2)}}b_{i_1-r}b_{i_2-r}b_{i_1+h_1-r}b_{i_2+h_2-r}Z_{rj_1}^4 \bigg\},
\end{align*}
where $\sigma_{j_1j_2}$ is $(j_1,j_2)$th components of $\bms$.
By Assumption~\ref{ass:C1}-(i), we have $\E|\tilde{S}_{h_1,h_2,H}| = O(p^2n^{-2})$.
Furthermore,
we have
\begin{eqnarray*}
\frac{\E\bigg|\tilde{S}_{0,0,H}+2\sum\limits_{h_1=1}^M \tilde{S}_{h_1,0,H}+2\sum\limits_{h_2=1}^M \tilde{S}_{0,h_2,H}
+4\sum\limits_{h_1=1}^M\sum\limits_{h_2=1}^M \tilde{S}_{h_1,h_2,H}\bigg|}{\bigg(c_0^2+4\sum\limits_{r=1}^Mc_0c_r+4\sum\limits_{r=1}^M\sum\limits_{s=1}^Mc_rc_s\bigg)\tr(\bms^2)}
=O(pn^{-2})=o(1).
\end{eqnarray*}
Besides that, the expectation of the low order term of Equation~\eqref{L7-8-1} is
\begin{eqnarray*}
&&\sum\limits_{j_1,j_2=1}^{p}\sigma_{j_1j_2}^2\bigg(
\sum\limits_{\substack{r\leq i_1 \wedge (i_1+h_1) \\ k \leq  i_2\wedge (i_2+h_2)}}b_{i_1-r}b_{i_2-k}b_{i_1+h_1-r}b_{i_2+h_2-k} +
\sum\limits_{\substack{r\leq i_1 \wedge (i_2+h_2) \\ k \leq i_2 \wedge (i_1+h_1)}}b_{i_1-r}b_{i_2-k}b_{i_1+h_1-k}b_{i_2+h_2-r}\Big)\nonumber\\
&&+\sum\limits_{j_1,j_2=1}^{p}\sigma_{j_1j_1}\sigma_{j_2j_2}
\sum\limits_{\substack{r\leq i_1\wedge i_2  \\ k \leq (i_1+h_1) \wedge (i_2 +h_2)}}b_{i_1-r}b_{i_2-r}b_{i_1+h_1-k}b_{i_2+h_2-k}\nonumber\\
&&+\Big(\sum\limits_{j_1\neq j_2=1}^{p}\sigma_{j_1j_1}\sigma_{j_2j_2}+2\sum\limits_{j_1\neq j_2=1}^{p}\sigma_{j_1j_2}^2\Big)
\sum\limits_{\substack{r\leq i_1\wedge i_2\wedge (i_1+h_1) \wedge (i_2 +h_2) }}b_{i_1-r}b_{i_2-r}b_{i_1+h_1-r}b_{i_2+h_2-r} \Big\}.
\end{eqnarray*}
By Assumption~\ref{ass:C1}-(ii) again,
\begin{eqnarray*}
&&4(n-h_2-3/2[n/2]-M/2)([n/2]-M-2h_2-1)\E\left\{\tilde{S}_{h_1,h_2,L} -a_{h_1}a_{h_2}\tr(\bms^2)\right\}\\
&=&\sum\limits_{i_1=1}^{[n/2]-M-2h_2-1}\sum\limits_{i_2=t+[n/2]}^{n-M-2h_2-1}\bigg\{
c_{i_2+h_2-i_1}c_{i_1+h_1-i_2}\tr(\bms^2) + c_{i_1-i_2}c_{i_2-i_1+h_2-h_1}\tr^2(\bms)\\
&&+\bigg(\sum\limits_{j_1\neq j_2=1}^{p}\sigma_{j_1j_1}\sigma_{j_2j_2}+2\sum\limits_{j_1\neq j_2=1}^{p}\sigma_{j_1j_2}^2\bigg)
\sum\limits_{\substack{r\leq i_1\wedge i_2\wedge (i_1+h_1) \wedge (i_2 +h_2) }}b_{i_1-r}b_{i_2-r}b_{i_1+h_1-r}b_{i_2+h_2-r}
\bigg\} \\
&=&o(p^2n^{-1}).
\end{eqnarray*}
Hence,
$\E\{(\tilde{S}_{h_1,h_2,L} -c_{h_1}c_{h_2}\tr(\bms^2))/c_{h_1}c_{h_2}\tr(\bms^2)\} = o(pn^{-3}) = o(1)$. Similarly, we can prove that
$\var\{(\tilde{S}_{h_1,h_2,L} -c_{h_1}c_{h_2}\tr(\bms^2))/c_{h_1}c_{h_2}\tr(\bms^2)\} = o(pn^{-2} + n^{-1}) = o(1)$. Thus,
\begin{eqnarray*}
\frac{\tilde{S}_{0,0,L}+2\sum\limits_{h_1=1}^M \tilde{S}_{h_1,0,L}+2\sum\limits_{h_2=1}^M \tilde{S}_{0,h_2,L}
+4\sum\limits_{h_1=1}^M\sum\limits_{h_2=1}^M \tilde{S}_{h_1,h_2,L}}{\bigg(c_0^2+4\sum\limits_{r=1}^Mc_0c_r+4\sum\limits_{r=1}^M\sum\limits_{s=1}^Mc_rc_s\bigg)\tr(\bms^2)}
=1+o_p(1).
\end{eqnarray*}
Above all, we get $\hat{\omega}/\omega\stackrel{p}{\rightarrow} 1$.
\end{proof}

\subsection{Proof of Theorem \ref{Tms-1}}
\begin{proof}
By the continuous mapping theorem and Theorem \ref{Th1}--\ref{Th2}, since $h(x) = \sup_{t}x(t)$ is continuous in the space $D[0,1]$ when $x(t)\in D[0,1]$, we have
 \begin{align*}
  S_{n,p}/\hat{\omega}\stackrel{d}{\rightarrow}  \max_{t\in [0,1]} V(t).
  \end{align*}
\end{proof}

\subsection{Proof of Proposition \ref{prop:Tms-alter}}

By Theorem~\ref{Th2}--\ref{Tms-1}, we have
\begin{equation*}
    \begin{aligned}
        S_{n,p}/\hat{\omega}=&\max_{1\leq k\leq n}\left\{W(k)-\hat{\mu}_{M,k}\right\}/\hat{\omega}\\
        \geq &\left\{W(\tau)-\hat{\mu}_{M,\tau}\right\}/\hat{\omega}\\
        \geq&\frac{\tau^2(n-\tau)^2}{n^3 \sqrt{p}}(\|\bdelta\|-\|\frac{1}{\tau}\sum_{i=1}^\tau \bepsilon_i-\frac{1}{n-\tau}\sum_{i=\tau+1}^n \bepsilon_i\|)^2/\hat{\omega}-\hat{\mu}_{M,\tau}/\hat{\omega}\\
        \asymp&p^{-1/2}n \|\bdelta\|^2+O_p(1).
    \end{aligned}
\end{equation*}
Thus, the test based on $S_{n,p}$ is consistent provided that $\|\bdelta\|^2\gtrsim p^{1/2}n^{-1}$.

\subsection{Proof of Theorem \ref{indnull}}

\begin{proof}
Let $t = k/n$, recall
\begin{align*}
\tilde{\bzeta}_{i,k/n}=\frac{1}{w_n-M}\sum\limits_{\ell=(i-1)w_n+1}^{iw_n-M}\bzeta_{\ell,k/n},
\end{align*}
with $\bzeta_{i,k/n}=a_{i,k}\bepsilon_{i}$. According to the proof of Theorem 1, we have
\begin{eqnarray*}
\omega^{-1}\{W(k)-\mu_{M,k}\}  = \mathcal{W}(\tilde{\bzeta}_{1,k/n},\dots,\tilde{\bzeta}_{k_n,k/n}) +o_p(1),
\end{eqnarray*}
where $\{\tilde{\bzeta}_{i,k/n}\}_{i=1}^n$ are independent and
\begin{align*}
\mathcal{W}(\tilde{\bzeta}_{1,k/n},\dots,\tilde{\bzeta}_{k_n,k/n})= \frac{2g(k)}{\sqrt{2p^{-1}\tr(\O_n^2)}}\sum\limits_{1\leq i < j\leq k_n}\bigg(\sum\limits_{\ell=(i-1)w_n+1}^{iw_n-M}\bzeta_{\ell,k}\bigg)^\top\bigg(\sum\limits_{\ell=(i-1)w_n+1}^{iw_n-M}\bzeta_{\ell,k}\bigg).
\end{align*}
Define $Z^{(NG)}_{0,j} = \max\limits_{k=1,\dots, n-1}|C^{(NG)}_{0,j}(k)|$ by substituting $\bepsilon_i$ for $\bX_i$ in $|C_{0,j}(k)|$, it suffices to show that $\max\limits_{k=1,\dots, n-1}\mathcal{W}(\tilde{\bzeta}_{1,k/n},\dots,\tilde{\bzeta}_{k_n,k/n})$ and $\max\limits_{j=1,\ldots,p}Z^{(NG)}_{0,j}$ are asymptotically independent.

 \subsubsection{Gaussian sequence.}\label{subsubsec:Gau}
For any fixed $x, y\in \mathbb{R}$, we define
\begin{eqnarray*}
A_{p} &:=& A_{p}(x) = \Big\{\max_{k=1,\dots, n-1}\mathcal{W}(\tilde{\eeta}_{1,k/n},\dots,\tilde{\eeta}_{k_n,k/n}) \leq   x\Big\},\\
B_{j}  &:=& B_{j}(y) = \Big\{Z^{(G)}_{0, j}>\mu_p\{\exp(-y)\}\Big\}.
\end{eqnarray*}
As $\pr(\cup_{j=1}^pB_{j})\rightarrow 1-\exp\{\exp(-y)\}$ and $\pr(A_p)\rightarrow F_{V}(x) $, we only need to show that
\begin{align*}
&\pr\bigg(\max_{k=1,\dots, n-1}\mathcal{W}(\tilde{\eeta}_{1,k/n},\dots,\tilde{\eeta}_{k_n,k/n}) \leq  x, \max_{j=1,\ldots,p}\max_{k=1,\dots, n-1}|C^{(G)}_{0,j}(k)| \leq \mu_p\{\exp(-y)\} \bigg)\\
&\rightarrow  F_{V}(x)\cdot \exp\{\exp(-y)\},
\end{align*}
or, equivalently,
\begin{align*}
\pr\bigg(\mathop{\bigcup}\limits_{j=1}^p A_pB_{j}\bigg) \rightarrow  F_{V}(x) \cdot [1-\exp\{\exp(-y)\}].
\end{align*}
For each $d\geq 1$, we define
\begin{align*}
\zeta_{AB}(p, d) &= \sum\limits_{1\leq j_1<\dots<j_d\leq p}\left|\pr(A_p B_{j_1}\dots B_{j_d})-\pr(A_p)\cdot\pr(B_{j_1} \dots B_{j_d})\right|,\\
H(p, d) &=  \sum\limits_{1\leq j_1<\dots<j_d\leq p}\pr(B_{j_1} \dots B_{j_d}).
\end{align*}
By the {\textrm {inclusion-exclusion principle}}, we observe that, for any integer $k\geq 1$,
\begin{align*}
\pr\bigg(\mathop{\bigcup}\limits_{j=1}^p A_pB_{j}\bigg)\leq \sum\limits_{1\leq j_1\leq p}\pr(A_pB_{j_1}) - \sum\limits_{1\leq j_1<j_2 \leq p}\pr(A_pB_{j_1}B_{j_2})+\dots\\
+\sum\limits_{1\leq j_1<\dots<j_{2k+1} \leq p}\pr(A_pB_{j_1} \dots B_{j_{2k+1}}),
\end{align*}
and
\begin{align*}
\pr\bigg(\mathop{\bigcup}\limits_{j=1}^p B_{j}\bigg)\geq \sum\limits_{1\leq j_1\leq p}\pr(B_{j_1}) - \sum\limits_{1\leq j_1<j_2 \leq p}\pr(B_{j_1}B_{j_2})+\dots\\
-\sum\limits_{1\leq j_1<\dots<j_{2k} \leq p}\pr(B_{j_1} \dots B_{j_{2k}}).
\end{align*}
Then,
\begin{eqnarray*}
\pr\bigg(\mathop{\bigcup}\limits_{j=1}^p A_pB_{j}\bigg) &\leq& \pr(A_p)\bigg\{\sum\limits_{1\leq j_1\leq p}\pr(B_{j_1}) - \sum\limits_{1\leq j_1<j_2 \leq p}\pr(B_{j_1}B_{j_2})+\dots\\
&&-\sum\limits_{1\leq j_1<\dots<j_{2k} \leq p}\pr(B_{j_1} \dots B_{j_{2k}})\bigg\} + \sum\limits_{d=1}^{2k}\zeta_{AB}(p, d) + H(p, 2k+1)\\
&\leq& \pr(A_p)\pr\bigg(\mathop{\bigcup}\limits_{j=1}^p B_{j}\bigg) + \sum\limits_{d=1}^{2k}\zeta_{AB}(p, d) + H(p, 2k+1).
\end{eqnarray*}
By using similar arguments, we have
\begin{eqnarray*}
\pr\bigg(\mathop{\bigcup}\limits_{j=1}^p A_pB_{j}\bigg) &\geq& \pr(A_p)\bigg\{\sum\limits_{1\leq j_1\leq p}\pr(B_{j_1}) - \sum\limits_{1\leq j_1<j_2 \leq p}\pr(B_{j_1}B_{j_2})+\dots\\
&&-\sum\limits_{1\leq j_1<\dots<j_{2k} \leq p}\pr(B_{j_1} \dots B_{j_{2k}})\bigg\} + \sum\limits_{d=1}^{2k}\zeta_{AB}(p, d)\\
&\geq& \pr(A_p)\pr\bigg(\mathop{\bigcup}\limits_{j=1}^p B_{j}\bigg) + \sum\limits_{d=1}^{2k}\zeta_{AB}(p, d).
\end{eqnarray*}
According to Equations (S.5) and (S.6) in \cite{wang2023JRSSB} and the following Lemma \ref{Th5-L1}, by fixing $k$ and letting $p\rightarrow \infty$, we obtain
\begin{eqnarray*}
\mathop{\lim\sup}\limits_{p\rightarrow \infty}\pr\bigg(\mathop{\bigcup}\limits_{j=1}^p A_pB_{j}\bigg) &\leq& F_{W}(x) \cdot [1-\exp\{\exp(-y)\}],\\
\mathop{\lim\inf}\limits_{p\rightarrow \infty}\pr\bigg(\mathop{\bigcup}\limits_{j=1}^p A_pB_{j}\bigg) &\geq& F_{W}(x) \cdot [1-\exp\{\exp(-y)\}].
\end{eqnarray*}
\end{proof}

\begin{lemma}
\label{Th5-L1}
Under the conditions in Theorem 5, then for each $d\geq 1$, as $p\rightarrow \infty$, we have $\zeta_{AB}(p, d)\rightarrow 0$.
\end{lemma}
\begin{proof}
For $i=1,\dots,n$, we define $\bxi_{i,(1)}=(\bxi_{i,j_1},\dots,\bxi_{i,j_d})^\top$ and $\bxi_{i,(2)}=(\bxi_{i,j_{d+1}},\dots,\bxi_{i,j_p})^\top$. For $k, l \in\{1,2\}$, define $\bGam_{kl}(0):=\cov(\bxi_{i,(k)}, \bxi_{i,(l)}) = c_h\bms_{kl}$. Since $\boldsymbol{U}_i := (U_{i,1},\dots,U_{i,p-d})^\top= \bxi_{i,(2)} - \bGam_{21}(0)\bGam_{11}(0)^{-1}\bxi_{i,(1)}$ and $\boldsymbol V_i  = (V_{i,1},\dots,V_{i,p-d})^\top:= \bGam_{21}(0)\bGam_{11}(0)^{-1}\bxi_{i,(1)}$ are independent, we partition $\bxi_{i,(2)}$ into $\bxi_{i,(2)}=\boldsymbol U_i+\boldsymbol V_i$. Recall that
\begin{align*}
\mathcal{W}(\tilde{\eeta}_{1,k/n},\dots,\tilde{\eeta}_{k_n,k/n})
=& 2\sum\limits_{1\leq i < j\leq k_n}g(k)(w_n-M)^2\tilde{\eeta}_{i,k/n}^\top\tilde{\eeta}_{j,k/n}/\sqrt{2p^{-1}\tr(\O_n^2)}.
\end{align*}

Then,
\begin{align*}
& \max_{k=1,\dots, n-1}\bigg\{g(k)(w_n-M)^2\sum\limits_{1\leq i < j\leq k_n}\tilde{\eeta}_{i,k/n}^\top\tilde{\eeta}_{j,k/n}\bigg\}\\
=& \max_{k=1,\dots, n-1}\bigg\{g(k)\sum\limits_{1\leq i < j\leq k_n}\bigg(\sum\limits_{\ell=(i-1)w_n+1}^{iw_n-M}a_{\ell,k}\bxi_{\ell}\bigg)^\top\bigg(\sum\limits_{\ell=(j-1)w_n+1}^{jw_n-M}a_{\ell,k}\bxi_{\ell}\bigg)\bigg\}\\
=& \max_{k=1,\dots, n-1}\bigg\{g(k)\sum\limits_{1\leq i < j\leq k_n}\sum\limits_{\ell_1=(i-1)w_n+1}^{iw_n-M}\sum\limits_{\ell_2=(j-1)w_n+1}^{jw_n-M}\!\!\!\!a_{\ell_1,k}a_{\ell_2,k}\bxi_{\ell_1}^\top\bxi_{\ell_2}\bigg\}\\
=& \max_{k=1,\dots, n-1}\bigg\{g(k)\sum\limits_{1\leq i < j\leq k_n}\sum\limits_{\ell_1=(i-1)w_n+1}^{iw_n-M}\sum\limits_{\ell_2=(j-1)w_n+1}^{jw_n-M}\!\!\!\!a_{\ell_1,k}a_{\ell_2,k}(\boldsymbol U_{\ell_1}^\top \boldsymbol U_{\ell_2} + \bxi_{\ell_1,(1)}^\top \bxi_{\ell_2,(1)}\\
&~~~~~~~~~~~~~~~~~~~~~~~~~~~~~~~~~~~~~~~~~~~~~~~~~~~~~~~~~~~~~~~~~~~~+2\boldsymbol U_{\ell_1}^\top \boldsymbol V_{\ell_2}+ \boldsymbol V_{\ell_1}^\top \boldsymbol V_{\ell_2})\bigg\}\\
:=&~S_{n,p,1}  + \Theta,
\end{align*}
where $S_{n,p,1}:=\max\limits_{k=1,\dots, n-1}\bigg\{g(k)\sum\limits_{1\leq i < j\leq k_n}\sum\limits_{\ell_1=(i-1)w_n+1}^{iw_n-M}\sum\limits_{\ell_2=(j-1)w_n+1}^{jw_n-M}a_{\ell_1,k}a_{\ell_2,k}\boldsymbol U_{\ell_1}^\top \boldsymbol U_{\ell_2}\bigg\}$ and
\begin{eqnarray*}
\Theta &\leq& \max_{k=1,\dots, n-1}\bigg\{g(k)\sum\limits_{1\leq i < j\leq k_n}\sum\limits_{\ell_1=(i-1)w_n+1}^{iw_n-M}\sum\limits_{\ell_2=(j-1)w_n+1}^{jw_n-M}a_{\ell_1,k}a_{\ell_2,k}\bxi_{\ell_1,(1)}^\top \bxi_{\ell_2,(1)}\bigg\}\\
&&+\max_{k=1,\dots, n-1}\bigg\{g(k)\sum\limits_{1\leq i < j\leq k_n}\sum\limits_{\ell_1=(i-1)w_n+1}^{iw_n-M}\sum\limits_{\ell_2=(j-1)w_n+1}^{jw_n-M}a_{\ell_1,k}a_{\ell_2,k}\boldsymbol V_{\ell_1}^\top \boldsymbol V_{\ell_2}\bigg\}\\
&&+2\max_{k=1,\dots, n-1}\bigg\{g(k)\sum\limits_{1\leq i < j\leq k_n}\sum\limits_{\ell_1=(i-1)w_n+1}^{iw_n-M}\sum\limits_{\ell_2=(j-1)w_n+1}^{jw_n-M}a_{\ell_1,k}a_{\ell_2,k}\boldsymbol U_{\ell_1}^\top \boldsymbol V_{\ell_2}\bigg\}\\
&:=&\Theta_1 +   \Theta_2 +  \Theta_3.
\end{eqnarray*}

We claim that, for any $\epsilon>0$ and sufficiently large $p$, there exists a sequence of positive constants $c:=c_p>0$ with  $c_p\rightarrow \infty$ such that
\begin{eqnarray}
\label{Th5-eq1}
\pr(|\Theta_i|\geq \epsilon\omega)\leq p^{-c}, \quad i=1,2,3,
\end{eqnarray}
which will be shown in later. Consequently, $\pr(|\Theta|\geq \epsilon\omega)\leq p^{-c}$ for some $c\rightarrow \infty$ and sufficiently large $p$.

By claim \eqref{Th5-eq1}, due to the independence of random vectors $\bU_i$ and  $\bxi_{i,(1)}$, we have
\begin{eqnarray*}
\pr(A_p(x)B_{j_1}\cdots B_{j_d}) &\leq& \pr(A_p(x)B_{j_1}\cdots B_{j_d}, |\Theta|<\epsilon\omega) + p^{-c}\nonumber\\
&\leq& \pr(S_{n,p,1}/\omega \leq  x + \epsilon, B_{j_1}\cdots B_{j_d}) + p^{-c}\nonumber\\
&=& \pr(S_{n,p,1}/\omega \leq  x + \epsilon)\cdot\pr(B_{j_1}\cdots B_{j_d}) + p^{-c}\nonumber\\
&\leq& \pr(A_p(x + 2\epsilon))\cdot\pr(B_{j_1}\cdots B_{j_d}) + 2p^{-c}.
\end{eqnarray*}
On the other hand, we have
\begin{eqnarray*}
\pr(S_{n,p,1}/\omega \leq x - \epsilon)\cdot \pr(B_{j_1}\cdots B_{j_d}) &=& \pr(S_{n,p,1}/\omega \leq x - \epsilon, B_{j_1}\cdots B_{j_d})\nonumber\\
&\leq& \pr(S_{n,p,1}/\omega \leq x - \epsilon, B_{j_1}\cdots B_{j_d}, |\Theta|<\epsilon \omega) + p^{-c}\nonumber\\
&\leq& \pr(A_p(x)B_{j_1}\cdots B_{j_d}) + p^{-c}.
\end{eqnarray*}
and
\begin{eqnarray*}
\pr(A_p(x-2\epsilon)) &\leq& \pr(A_p(x-2\epsilon), |\Theta|<\epsilon\omega) + p^{-c}\\
&\leq& \pr(S_{n,p,1}/\omega \leq x-\epsilon) + p^{-c}.
\end{eqnarray*}
Then,
\begin{eqnarray*}
\pr(A_p(x)B_{j_1}\cdots B_{j_d}) &\geq& \pr(A_p(x-2\epsilon)) \cdot \pr(B_{j_1}\cdots B_{j_d}) - 2p^{-c}.
\end{eqnarray*}
Since $\pr(A_p(x))$ is increasing in $x$, we have
\begin{eqnarray*}
&&\left|\pr(A_p(x)B_{j_1}\cdots B_{j_d}) - \pr(A_p(x))\pr(B_{j_1}\cdots B_{j_d})\right| \\
&\leq& \left\{\pr(A_p(x + 2\epsilon))-\pr(A_p(x - 2\epsilon))\right\}\pr(B _{j_1}\cdots B_{j_d})  + 2p^{-c}.
\end{eqnarray*}
Since for any $d\geq 1$, $H(p,d)\rightarrow (d!)^{-1}\exp(-d x/2)$ as $p\rightarrow \infty$ and $\pr(A_p(x + 2\epsilon))-\pr(A_p(x - 2\epsilon)) \rightarrow 0$ as $p\rightarrow \infty$ and $\epsilon\rightarrow 0$, we have
\begin{eqnarray*}
\zeta_{AB}(p,d)\leq  \{\pr(A_p(x + 2\epsilon))-\pr(A_p(x - 2\epsilon))\}\cdot H(p,d) + 2\binom{p}{d}p^{-c}\rightarrow 0,
\end{eqnarray*}
as $p\rightarrow \infty$ for any $d\geq 1$.

{\noindent\textbf{Proof of Equation \eqref{Th5-eq1}.}}

Let
 \begin{eqnarray*}
 W_{1,(1)}(k)& := & \frac{k^2(n-k)^2}{n^3\sqrt{p}}\sum\limits_{1\leq i < j\leq k_n}\bigg(\sum\limits_{\ell=(i-1)w_n+1}^{iw_n-M}a_{\ell,k}\bxi_{\ell,(1)}\bigg)^\top\bigg(\sum\limits_{\ell=(j-1)w_n+1}^{jw_n-M}a_{\ell,k}\bxi_{\ell,(1)}\bigg).
 \end{eqnarray*}
 Define $\O^{}_{w_n-M, M, 1} := \bGam_{11}(0)+2\sum\nolimits_{h=1}^{n}(1-\frac{|h|}{n})\bGam_{11}(h)$. Similar to the proof of Proposition \ref{L6.1-P2}, we define
 \begin{eqnarray*}
\varsigma(k)^2 :=  \frac{k^2(n-k)^2}{n^4}\cdot\frac{\tr(\O^{2}_{w_n-M, M,1})}{p},
 \end{eqnarray*}
 then we have $\var\{W_{1,(1)}(k)\}  =  \varsigma(k)^2\{1+o(1)\}$.
Then,
\begin{eqnarray*}
&&\pr(|\Theta_1|\geq \epsilon\omega) \\
&=& \pr\bigg(\bigg|\max_{k=1,\dots, n-1}\bigg\{g(k)\sum\limits_{1\leq i < j\leq k_n}\bigg(\sum\limits_{\ell=(i-1)w_n+1}^{iw_n-M}a_{\ell,k}\bxi_{\ell,(1)}\bigg)^\top\bigg(\sum\limits_{\ell=(j-1)w_n+1}^{jw_n-M}a_{\ell,k}\bxi_{\ell,(1)}\bigg)\bigg\}\bigg| \geq \epsilon \sqrt{2\tr(\O_n^2)/p}\bigg) \\
&\leq& n\pr\bigg(\bigg|\frac{n^{2}}{k(n-k)}\cdot\frac{W_{1,(1)}(k)}{\sqrt{\tr(\O^{2}_{w_n-M, M,1})/p}}\bigg| \geq \epsilon \frac{n^{2}}{k(n-k)}\cdot \frac{\sqrt{2\tr(\O_n^2)/p}}{\sqrt{\tr(\O^{2}_{w_n-M, M,1})/p}}\bigg) \\
& = & n\pr\bigg(\bigg|\frac{W_{1,(1)}(k)}{\sqrt{\tr(\bms^{2}_{11})/p}}\bigg| \geq C_\epsilon \sqrt{\frac{\tr(\bms^2)}{\tr(\bms_{11}^2)}}\bigg) \\
&\leq & n\exp\bigg\{-C_\epsilon \sqrt{\frac{\tr(\bms^2)}{\tr(\bms_{11}^2)}}\bigg\} \cdot
\E\bigg(\exp\bigg|\frac{W_{1,(1)}(k)}{\sqrt{\tr(\bms^{2}_{11})/p}}\bigg|\bigg)\\
&\leq & n\exp\bigg\{-C_\epsilon \sqrt{\frac{\tr(\bms^2)}{\tr(\bms_{11}^2)}}\bigg\} \cdot \log n \rightarrow 0,
\end{eqnarray*}
with Assumption~\ref{ass:C1} and condition $\log n = o(p^{1/4})$, where the last inequality follows with the fact
\begin{eqnarray*}
p^{1/2}\tr^{-1/2}(\bms^{2}_{11})W_{1,(1)}(k) / \log \log n \rightarrow 0, \quad a.s.
 \end{eqnarray*}
by the law of the iterated logarithm of zero-mean square integrable martingale (see Theorem 4.8 in \cite{hall2014martingale}).
Similarly, by $\log n = o(p^{1/4})$, we have
\begin{eqnarray*}
\pr(|\Theta_2|\geq \epsilon\omega) \leq  n\exp\bigg\{-C_\epsilon \sqrt{\frac{\tr(\bms^2)}{\tr(\bms_{22.1}\bms_{21}\bms_{11}^{-1}\bms_{12})}}\bigg\} \cdot \log n\rightarrow 0,
\end{eqnarray*}
and
\begin{eqnarray*}
\pr(|\Theta_3|\geq \epsilon\omega) \leq  n\exp\bigg\{-C_\epsilon \sqrt{\frac{\tr(\bms^2)}{\tr[(\bms_{21}\bms_{11}^{-1}\bms_{12})^2]}}\bigg\} \cdot \log n\rightarrow 0.
\end{eqnarray*}

\subsubsection{non-Gaussian sequence}
For $\boldsymbol z = (z_1,\dots, z_q)^\top \in\mathbb{R}^q$, we consider a smooth approximation of the maximum function $\boldsymbol z \rightarrow \max\limits_{1\leq i \leq q} z_i$, namely,
\begin{eqnarray*}
F_{\beta}(\boldsymbol z) = \beta^{-1}\log \bigg\{\sum\limits_{j=1}^q \exp(\beta z_j)\bigg\},
\end{eqnarray*}
where $\beta >0$ is the smoothing parameter that controls the level of approximation. An elementary calculation shows that $\forall \boldsymbol z \in \mathbb{R}^q$,
\begin{eqnarray*}
0 \leq F_{\beta}(\boldsymbol z) - \max\limits_{1\leq j\leq q}z_j\leq \beta^{-1}\log q,
\end{eqnarray*}
see \cite{MR3161448}.

For any $i = 1,\dots,n$ and $k=1,\dots, n-1$, we define
\begin{eqnarray*}
\breve{a}_{i,k} := \left\{
\begin{array}{lr}
1 - \dfrac{k}{n}, &  i\leq k,\\
-\dfrac{k}{n}, &   i> k,
\end{array}
\right.
\end{eqnarray*}
and reformulate
\begin{eqnarray*}
C_{\gamma,j}(k) &:=& \bigg\{\frac{k}{n}\bigg(1-\frac{k}{n}\bigg)\bigg\}^{-\gamma} \frac{1}{\sqrt{n}}\bigg(\sum\limits_{i=1}^{k}X_{ij}-\frac{k}{n}\sum\limits_{i=1}^{n}X_{ij}\bigg)/\hat{\sigma}_j\\
&= & \bigg\{\frac{k}{n}\bigg(1-\frac{k}{n}\bigg)\bigg\}^{-\gamma} \frac{1}{\sqrt{n}}\sum\limits_{i=1}^{n}\breve{a}_{i,k}X_{ij}/\hat{\sigma}_j,
\end{eqnarray*}
and
\begin{eqnarray*}
W(k)-\mu_{M,k}&=&\frac{2}{n\sqrt{p}}\sum\limits_{1\leq i <  j\leq n}\breve{a}_{i,k}\breve{a}_{j,k}\bX_{i}^\top\bX_{j} + o_p(\omega).
\end{eqnarray*}
with the analysis in Theorem 1.

Define
\begin{equation*}
    \begin{aligned}
    U(\bepsilon_1,\ldots,\bepsilon_n) &:= \beta^{-1}\log \bigg\{\sum\limits_{k=1}^{n-1} \exp\bigg(2\beta n^{-1}\sum\limits_{1\leq i <  j\leq n}\breve{a}_{i,k}\breve{a}_{j,k}\bepsilon_{i}^\top\bepsilon_{j}/\sqrt{2\tr(\O_n^2)}\bigg)\bigg\},\\
    V(\bepsilon_1,\ldots,\bepsilon_{n}) &:= \beta^{-1}\log \bigg\{\sum\limits_{j=1}^p\sum\limits_{k=1}^{n-1} \exp\bigg(\beta n^{-1/2} \sum\limits_{i=1}^{n}\breve{a}_{i,k}\epsilon_{ij}\bigg)\bigg\},
    \end{aligned}
\end{equation*}
by letting $\beta = n^{1/8}\log(np)$, it is suffices to prove
\begin{eqnarray*}
\pr(U(\bepsilon_1,\dots,\bepsilon_{n})\leq x, V(\bepsilon_1,\dots,\bepsilon_{n})\leq u_p\{\exp(-y)\}
\rightarrow F_V(x)\cdot \exp\{-\exp(-y)\}.
\end{eqnarray*}
According to the results of Section~\ref{subsubsec:Gau}, it remains to show that
$(U(\bepsilon_1,\dots,\bepsilon_{n}), V(\bepsilon_1,\dots,\bepsilon_{n}))$ has the same limiting distribution as $(U(\bxi_1,\dots,\bxi_{n}), V(\bxi_1,\dots,\bxi_{n}))$.
Similar to the proof of Lemma \ref{Th1-L2}, it suffices to show that
\begin{eqnarray*}
\E\{f(U(\bepsilon_1,\dots,\bepsilon_{n}), V(\bepsilon_1,\dots,\bepsilon_{n}))\} - \E\{f(U(\bxi_1,\dots,\bxi_{n}), V(\bxi_1,\dots,\bxi_{n}))\} \rightarrow 0,
\end{eqnarray*}
for any $f\in \mathcal{C}_b^3(\mathbb{R})$ as $(n,p)\rightarrow \infty$.

Furthermore, for $d=1,\dots,n+1$, we define
\begin{eqnarray*}
&&U_d:=U(\bepsilon_1,\dots,\bepsilon_{d-1},\bxi_d,\dots,\bxi_{n}),\\
&&U_{d,0}:=\beta^{-1}\log\left(\sum_{k=1}^{n}\exp\left\{ \beta   n^{-1}\left(\sum_{1\leq i<j\leq d-1} \breve{a}_{i,k}\breve{a}_{j,k}\bepsilon_i^\top \bepsilon_j +\sum_{d+1\leq i<j\leq n} \breve{a}_{i,k}\breve{a}_{j,k}\bxi_i^\top \bxi_j \right.\right.\right.\\
&&\qquad\qquad\qquad\qquad\qquad\qquad\qquad\left.\left.\left.+\sum_{i=1}^{d-1}\sum_{j=d+1}^n \breve{a}_{\ell_1,k}\breve{a}_{\ell_2,k}\bepsilon_i^\top \bxi_j \right)/\sqrt{2\text{tr}(\O_n^2)}\right\}\right)\\
&&V_d := V(\bepsilon_1,\dots,\bepsilon_{d-1},\bxi_d,\dots,\bxi_{n}),\\
&&V_{d,0} := \beta^{-1}\log \bigg\{\sum\limits_{j=1}^p\sum\limits_{k=1}^{n-1} \exp\bigg(\beta n^{-1/2} \sum\limits_{i=1}^{d-1}\breve{a}_{i,k}\epsilon_{ij}+\beta n^{-1/2}\sum\limits_{i=d+1}^{n}\breve{a}_{i,k}\xi_{ij}\bigg)\bigg\}.
\end{eqnarray*}

Then,
\begin{eqnarray*}
&&|\E\{f(U(\bepsilon_1,\dots,\bepsilon_{n}), V(\bepsilon_1,\dots,\bepsilon_{n}))\} - \E\{f(U(\bxi_1,\dots,\bxi_{n}), V(\bxi_1,\dots,\bxi_{n}))\}| \\
&\leq&\sum\limits_{d=1}^n|\E\{f(U_d,V_d)\} - \E\{f(U_{d+1},V_{d+1})\}|.
\end{eqnarray*}

By Taylor’s expansion, we have
\begin{eqnarray*}
&&f(U_d,V_d) - f(U_{d,0},V_{d,0}) \\
&=&f_1(U_{d,0},V_{d,0}) (U_d-U_{d,0}) + f_2(U_{d,0},V_{d,0}) (V_d-V_{d,0}) \\
&&+\frac{1}{2}f_{11}(U_{d,0},V_{d,0}) (U_d-U_{d,0})^2+\frac{1}{2}f_{22}(U_{d,0},V_{d,0}) (V_d-V_{d,0})^2 \\
&&+ \frac{1}{2}f_{12}(U_{d,0},V_{d,0}) (U_d-U_{d,0})(V_d-V_{d,0}) + O(|U_d-U_{d,0}|^3) + O(|V_d-V_{d,0}|^3),
\end{eqnarray*}
and
\begin{eqnarray*}
&&f(U_{d+1},V_{d+1}) - f(U_{d,0},V_{d,0}) \\
&=&f_1(U_{d,0},V_{d,0}) (U_{d+1}-U_{d,0}) + f_2(U_{d,0},V_{d,0}) (V_{d+1}-V_{d,0}) \\
&&+\frac{1}{2}f_{11}(U_{d,0},V_{d,0}) (U_{d+1}-U_{d,0})^2+\frac{1}{2}f_{22}(U_{d,0},V_{d,0}) (V_{d+1}-V_{d,0})^2 \\
&&+ \frac{1}{2}f_{12}(U_{d,0},V_{d,0}) (U_{d+1}-U_{d,0})(V_{d+1}-V_{d,0}) + O(|U_{d+1}-U_{d,0}|^3) + O(|V_{d+1}-V_{d,0}|^3),
\end{eqnarray*}
where $f=f(x,y)$, $f_1=\partial f/\partial x$, $f_2=\partial f/\partial y$, $f_{11}=\partial^2 f/\partial x^2$, $f_{22}=\partial^2 f/\partial y^2$, and $f_{12}=\partial^2 f/\partial x\partial y$.
By the proof of Proposition 8.2, we have
\begin{eqnarray*}
\E\{f_1(U_{d,0},V_{d,0})(U_{d}-U_{d,0})\} &=& \E\{f_1(U_{d,0},V_{d,0})(U_{d+1}-U_{d,0})\},\\
\E\{f_{11}(U_{d,0},V_{d,0})(U_{d}-U_{d,0})^2\} &=& \E\{f_{11}(U_{d,0},V_{d,0})(U_{d+1}-U_{d,0})^2\}.
\end{eqnarray*}
Furthermore, for $l = k+(j-1)n$, we define
\begin{eqnarray*}
z_{d,0,l} &:=& n^{-1/2}\sum\limits_{i=1}^{d-1}\breve{a}_{i,k}\epsilon_{ij} + n^{-1/2}\sum\limits_{i=d+1}^{n}\breve{a}_{i,k}\xi_{ij}, \\
z_{d,l} &:=& z_{d,0,l} + n^{-1/2}\breve{a}_{d,k}\xi_{dj}, \\
z_{d+1,l} &:=& z_{d,0,l} + n^{-1/2}\breve{a}_{d,k}\epsilon_{dj},
\end{eqnarray*}
and $\z_{d,0} = (z_{d,0,1},\dots,z_{d,0,np})^\top$, $\z_{d} = (z_{d,1},\dots,z_{d,np})^\top$. By Taylor's expansion, we have
\begin{eqnarray*}
&&V_{d} - V_{d,0} \\
&=& \sum\limits_{l=1}^{np}\partial_l F_{\beta}(\z_{d,0})(z_{d,l}-z_{d,0,l}) +\frac{1}{2}\sum\limits_{l_1=1}^{np}\sum\limits_{l_2=1}^{np}\partial_{l_1}\partial_{l_2} F_{\beta}(\z_{d,0})(z_{d,l_1}-z_{d,0,l_1})(z_{d,l_2}-z_{d,0,l_2}) \\
&&+\frac{1}{6}\sum\limits_{l_1=1}^{np}\sum\limits_{l_2=1}^{np}\sum\limits_{l_3=1}^{np}\partial_{l_1}\partial_{l_2}\partial_{l_3}F_{\beta}(\z_{d,0}+\iota(\z_{d}-\z_{d,0}))(z_{d,l_1}-z_{d,0,l_1})(z_{d,l_2}-z_{d,0,l_2}) (z_{d,l_3}-z_{d,0,l_3}) ,
\end{eqnarray*}
for some $\iota\in(0,1)$. Since $\E(\bepsilon_i) = \E(\bxi_i) = 0$ under $H_0$, and $\E(\bepsilon_i\bepsilon_i^{\top}) = \E(\bxi_i\bxi_i^{\top})$, it can be easily verified that
\begin{eqnarray*}
\E\{z_{d,l}  - z_{d,0,l}\mid\mathcal{F}_d\}  = \E\{z_{d+1,l}  - z_{d,0,l}\mid\mathcal{F}_d\}, \quad \E\{(z_{d,l}  - z_{d,0,l})^2\mid\mathcal{F}_d\}  = \E\{(z_{d+1,l}  - z_{d,0,l})^2\mid\mathcal{F}_d\}.
\end{eqnarray*}
By Lemma A.2 in \cite{MR3161448}, we have
\begin{eqnarray*}
\bigg|\sum\limits_{l_1=1}^{np}\sum\limits_{l_2=1}^{np}\sum\limits_{l_3=1}^{np}\partial_{l_1}\partial_{l_2}\partial_{l_3}F_{\beta}(\z_{d,0}+\iota(\z_{d}-\z_{d,0}))\bigg| \leq C\beta^2
\end{eqnarray*}

 By the definition of $\xi_{ij}$  and $\epsilon_{ij}$, $i=1,\ldots,n,j=1,\ldots,p$ and Assumption~\ref{ass:C4}, $\xi_{ij}$ are Gaussian random variables and $\epsilon_{ij}$ are sub-Gaussian random variables. Then we have $\pr(\max\limits_{1\leq i \leq n, 1\leq j \leq p}|\epsilon_{ij}|>C\log(np)) \rightarrow 0$, and   $\pr(\max\limits_{1\leq i \leq n, 1\leq j \leq p}|\xi_{ij}|>C\log(np)) \rightarrow 0$.

Hence,
\begin{eqnarray*}
&&\bigg|\frac{1}{6}\sum\limits_{l_1=1}^{np}\sum\limits_{l_2=1}^{np}\sum\limits_{l_3=1}^{np}\partial_{l_1}\partial_{l_2}\partial_{l_3}F_{\beta}(\z_{d,0}+\iota(\z_{d}-\z_{d,0}))(z_{d,l_1}-z_{d,0,l_1})(z_{d,l_2}-z_{d,0,l_2}) (z_{d,l_3}-z_{d,0,l_3})\bigg|\\
&\leq&C\beta^2n^{-3/2}\{\log(np)\}^3
\end{eqnarray*}
holds with probability approaching one. Consequently, with probability approaching one, we have
\begin{eqnarray*}
|\E\{f_2(U_{d,0},V_{d,0})(V_{d}-V_{d,0})\} - \E\{f_2(U_{d,0},V_{d,0})(V_{d+1}-V_{d,0})\}| \leq C\beta^2n^{-3/2}\{\log(np)\}^3.
\end{eqnarray*}

Similarly, it can be verified that
\begin{eqnarray*}
|\E\{f_{22}(U_{d,0},V_{d,0})(V_{d}-V_{d,0})^2\} - \E\{f_{22}(U_{d,0},V_{d,0})(V_{d+1}-V_{d,0})^2\}| \leq C\beta^2n^{-3/2}\{\log(np)\}^3,
\end{eqnarray*}
and
\begin{eqnarray*}
&&|\E\{f_{12}(U_{d,0},V_{d,0})(U_{d}-U_{d,0})(V_{d}-V_{d,0})\} - \E\{f_{12}(U_{d,0},V_{d,0})(U_{d+1}-U_{d,0})(V_{d+1}-V_{d,0})\}|\\
&\leq& C\beta^2n^{-3/2}\{\log(np)\}^3.
\end{eqnarray*}
By Lemma A.2 in \cite{MR3161448} and the expression of $V_{d} - V_{d,0}$ , we have $\E(|V_{d} - V_{d,0}|^3) = O\{n^{-3/2}\{\log(np)\}^3\}$.

Next, we consider $W_d,W_{d+1}$ and $W_{d,0}$. Define
\begin{equation*}
    \begin{aligned}
       z^w_{d,0,k}=& n^{-1}\sum_{1\leq \ell_1<\ell_2\leq d-1}\breve{a}_{\ell_1,k}\breve{a}_{\ell_2,k}\bepsilon_{\ell_1}^\top \bepsilon_{\ell_2}/\sqrt{2\text{tr}(\O_n^2)} \\
       &\qquad+ n^{-1}\sum_{d+1\leq \ell_1<\ell_2\leq n}\breve{a}_{\ell_1,k}\breve{a}_{\ell_2,k}\bxi_{\ell_1}^\top \bxi_{\ell_2}/\sqrt{2\text{tr}(\O_n^2)} \\
       &\qquad+ n^{-1}\sum_{\ell_1=1}^{d-1}\sum_{\ell_2=d+1}^{k_n}\breve{a}_{\ell_1,k}\breve{a}_{\ell_2,k}\bepsilon_{\ell_1}^\top \bxi_{\ell_2}/\sqrt{2\text{tr}(\O_n^2)}, \\
       z^w_{d,k}=&z^w_{d,0,k}+ n^{-1}\sum_{\ell_1=1}^{d-1} \breve{a}_{\ell_1,k}\breve{a}_{d,k}\bepsilon_{\ell_1}^\top \bxi_{d}/\sqrt{2\text{tr}(\O_n^2)}\\
       &\qquad+ n^{-1}\sum_{\ell_2=d+1}^n \breve{a}_{d,k}\breve{a}_{\ell_2,k}\bxi_{d}^\top \bxi_{\ell_2}/\sqrt{2\text{tr}(\O_n^2)},\\
       z^w_{d+1,k}=&z^w_{d,0,k}+ n^{-1}\sum_{\ell_1=1}^{d-1}\breve{a}_{\ell_1,k}\breve{a}_{d,k}\bepsilon_{\ell_1}^\top \bepsilon_{d}/\sqrt{2\text{tr}(\O_n^2)}\\
       &\qquad+ n^{-1}\sum_{\ell_2=d +1}^n \breve{a}_{d,k}\breve{a}_{\ell_2,k}\bepsilon_{d}^\top \bxi_{\ell_2}/\sqrt{2\text{tr}(\O_n^2)},
    \end{aligned}
\end{equation*}
   and let  $z^w_{d,0}=(z^w_{d,0,1},\ldots,z^w_{d,0,n })^\top$ and $z^w_d=(z^w_{d,1},\ldots,z^w_{d,n})^\top$.
By Taylor's expansion, we have
\begin{equation*}
\begin{aligned}
&~~W_d-W_{d, 0}= \\
& \sum_{l=1}^{n} \partial_l F_\beta\left(\boldsymbol{z}^w_{d, 0}\right)\left( z^w_{d, l}- z^w_{d, 0, l}\right)+\frac{1}{2} \sum_{l=1}^{n} \sum_{k=1}^{n} \partial_k \partial_l F_\beta\left(\boldsymbol{z}^w_{d, 0}\right)\left( z^w_{d, l}- z^w_{d, 0, l}\right)\left( z^w_{d, k}- z^w_{d, 0, k}\right) \\
& +\frac{1}{6} \sum_{l=1}^{n} \sum_{k=1}^{n} \sum_{v=1}^{n} \partial_v \partial_k \partial_l F_\beta\left(\boldsymbol{z}^w_{d, 0}+\iota\left(\boldsymbol{z}^w_d-\boldsymbol{z}^w_{d, 0}\right)\right)\left( z^w_{d, l}- z^w_{d, 0, l}\right)\left( z^w_{d, k}-\boldsymbol z^w_{d, 0, k}\right)\left(\boldsymbol z^w_{d, v}-\boldsymbol z^w_{d, 0, v}\right),
\end{aligned}
\end{equation*}
for some $\iota \in(0,1)$. Again, due to $\E(\bepsilon_i) = \E(\bxi_i) = 0$ under $H_0$, and $\E(\bepsilon_i\bepsilon_i^{\top}) = \E(\bxi_i\bxi_i^{\top})$, we verify that
\begin{eqnarray*}
\E(z_{d,l}^w  - z_{d,0,l}^w|\mathcal{F}_d)  = \E(z_{d+1,l}^w  - z_{d,0,l}^w|\mathcal{F}_d), \quad \E\{(z_{d,l}^w  - z_{d,0,l}^w)^2|\mathcal{F}_d\}  = \E\{(z_{d+1,l}^w  - z_{d,0,l}^w)^2|\mathcal{F}_d\}.
\end{eqnarray*}

We next consider $\E\left(\max_{1\leq k\leq n}\vert z^w_{d,k}-z^w_{d,0,k}\vert\right)$ with,
\begin{align}
        z^w_{d,k}-z^w_{d,0,k}=& n^{-1}\sum_{\ell_1=1}^{d-1}\breve{a}_{\ell_1,k}\breve{a}_{d,k}\bepsilon_{\ell_1}^\top \bxi_{d}/\sqrt{2\text{tr}(\O_n^2)}\label{eq:z_k^w1}\\
       &\qquad+ n^{-1}\sum_{\ell_2=d+1}^n \breve{a}_{d,k}\breve{a}_{\ell_2,k}\bxi_{d}^\top \bxi_{\ell_2}/\sqrt{2\text{tr}(\O_n^2)},\label{eq:z_k^w2}
\end{align}

Notice that, for Equation~\eqref{eq:z_k^w1},
\begin{equation*}
    \begin{aligned}
        n^{-1}\sum_{\ell_1=1}^{d-1}\breve{a}_{\ell_1,k}\breve{a}_{d,k}\bepsilon_{\ell_1}^\top \bxi_{d}/\sqrt{2\text{tr}(\O_n^2)}
        =n^{-1}\left\{\sum_{\ell_1=1}^{d-1} \breve{a}_{\ell_1,k}\bepsilon_{\ell_1}\right\}^\top \bxi_{d}\breve{a}_{d,k}/\sqrt{2\text{tr}(\O_n^2)},
    \end{aligned}
\end{equation*}
and by Assumption~\ref{ass:C1},
\begin{equation*}
    \begin{aligned}
        \sum_{\ell_1=1}^{d-1} \breve{a}_{\ell_1,k}\bepsilon_{\ell_1}=&\sum_{\ell_1=1}^{d-1}\breve{a}_{\ell_1,k}\sum_{\ell=0}^\infty \bms^{1/2}b_{\ell}\bZ_{\ell_1-\ell}\\
        =&\bms^{1/2}\sum_{\ell=-
        \infty}^{d-1}\left(\sum_{\ell_1=1}^{d-1}\breve{a}_{\ell_1,k} b_{\ell_1-\ell}\right)\bZ_\ell\\
        :=&\bms^{1/2}\sum_{\ell=-
        \infty}^{d-1}c_{\ell,k}\bZ_\ell,
    \end{aligned}
\end{equation*}
where $c_{\ell,k}=\sum_{\ell_1=1}^{d-1}\breve{a}_{\ell_1,k} b_{\ell_1-\ell}$ and for notational convenience, we set $b_{\ell}=0$ for $\ell<0$. And we have
\begin{equation}\label{eq:H0_phi_zd}
    \begin{aligned}
        \sum_{\ell=-
        \infty}^{d-1}c_{\ell,k}^2=&\sum_{\ell=-
        \infty}^{d-1}\left(\sum_{\ell_1=1}^{d-1}\breve{a}_{\ell_1,k} b_{\ell_1-\ell}\right)^2
        \leq \sum_{\ell=-
        \infty}^{d-1} \left(\sum_{\ell_1=1}^{d-1}\vert b_{\ell_1-\ell}\vert\right)^2\\
        =&\sum_{\ell=-
        \infty}^{d-1} \sum_{\ell_1=1}^{d-1}\sum_{\ell_2=1}^{d-1}\vert b_{\ell_1-\ell}\vert\vert b_{\ell_2-\ell}\vert=\sum_{s=1}^{d-1}\sum_{t=1}^{d-1}\sum_{\ell=\max(s,t)+1}^{\min(s,t)+d-1}\vert b_s\vert \vert b_t\vert\\
        \leq &\left(d-1\right)\left(\sum_{\ell=0}^\infty \vert b_\ell\vert\right)^2.
    \end{aligned}
\end{equation}

Taking the expectation on $\{
\bZ_{\ell}\}_{d-1}$,
\begin{equation*}
    \begin{aligned}
        \phi^2_{z,d}:=& \max_{1\leq k\leq n}\E \left\{\sum_{\ell=-
        \infty}^{d-1}\left(c_{\ell,k}\bZ_\ell^\top \bms^{1/2} \bxi_{d}\right)^2\right\}\\
        =&\max_{1\leq k\leq n} \left\{\sum_{\ell=-
        \infty}^{d-1}c_{\ell,k}^2 \right\}\bxi_{d}^\top \bms \bxi_{d}\\
        \lesssim &\left(d-1\right)\bxi_{d}^\top \bms \bxi_{d},
    \end{aligned}
\end{equation*}
where the last inequality follows from the Equation~\eqref{eq:H0_phi_zd} and Assumption~\ref{ass:C1}. Similarly, we have
\begin{equation*}
    \begin{aligned}
        M_{z,d}:=&\E\left\{\max_{1\leq k\leq n,\ell\leq d-1}\left(c_{\ell,k}\bZ_\ell^\top \bms^{1/2}\bxi_{d}\right)^2\right\}\\
        \leq &\sum_{\ell\leq d-1}\E\left\{\max_{1\leq k\leq n}(c_{\ell,k}^2)\left(\bZ_\ell^\top \bms^{1/2}\bxi_{d}\right)^2\right\}\\
         \leq &\sum_{\ell\leq d-1}\E\left\{\left(\sum_{\ell_1=1}^{d-1}\vert b_{\ell_1-\ell}\vert\right)^2\left(\bZ_\ell^\top \bms^{1/2}\bxi_{d}\right)^2\right\}\\
        \leq & \left(d-1\right)\left(\sum_{\ell=0}^\infty \vert b_\ell\vert\right)^2 \E\{(\bZ_{d-1}\bms^{1/2}\bxi_{d})^2\}\\
        \lesssim &\left(d-1\right)\bxi_{d}^\top \bms \bxi_{d}.
    \end{aligned}
\end{equation*}

By Lemma~\ref{LemmaE.1central}, we have
\begin{equation*}
    \begin{aligned}
        &\E \left\{\max_{1\leq k\leq n}\left\vert\sum_{\ell_1=1}^{d-1}\breve{a}_{\ell_1,k}\breve{a}_{d,k}\bepsilon_{\ell_1}^\top \bxi_{d}\right\vert\right\}\\
        \lesssim  &\E\left\{\sqrt{(d-1)\bxi_{d}^\top\bms\bxi_{d}\log p}\right\}+\E \left\{\sqrt{(d-1)\bxi_{d}^\top\bms\bxi_{d}}\log p\right\}\\
        \leq &2n^{1/2}\log n\{\tr(\bms^2)\}^{1/2}.
    \end{aligned}
\end{equation*}

Thus, we have
\begin{equation}\label{eq:z_k^w_rate}
    \begin{aligned}
        &\E\left[\max_{1\leq k\leq n}\left\vert n^{-1}\left\{\sum_{\ell_1=1}^{d-1} \breve{a}_{\ell_1,k}\bepsilon_{\ell_1}\right\}^\top \bxi_{d}\breve{a}_{d,k}/\sqrt{2\text{tr}(\O_n^2)}\right\vert\right]\\
        \leq &n^{-1}\E\left\{\max_{1\leq k\leq n}\left\vert \sum_{\ell_1=1}^{d-1}\breve{a}_{\ell_1,k}\bepsilon_{\ell_1}^\top \bxi_{d}\right\vert\right\}/\sqrt{2\text{tr}(\O_n^2)}\\
        \lesssim & n^{-1/2} \log n.
    \end{aligned}
\end{equation}
Taking the same procedure for Equation~\eqref{eq:z_k^w2} and combing the Equation~\eqref{eq:z_k^w_rate}, we have, $\E\left(\max_{1\leq k\leq n}\vert z^w_{d,k}-z^w_{d,0,k}\vert\right)\lesssim n^{-1/2}\log n$.

Therefore,
$$
\begin{aligned}
&\left\vert\frac{1}{6} \sum_{l,k,v=1}^{n} \partial_v \partial_k \partial_l F_\beta\left(\boldsymbol{z}^w_{d, 0}+\iota\left(\boldsymbol{z}^w_d-\boldsymbol{z}^w_{d, 0}\right)\right)\left(\boldsymbol z^w_{d, l}-\boldsymbol z^w_{d, 0, l}\right)\left(\boldsymbol z^w_{d, k}-\boldsymbol z^w_{d, 0, k}\right)\left(\boldsymbol z^w_{d, v}-\boldsymbol z^w_{d, 0, v}\right)\right\vert\\
& \lesssim\beta^2 n^{-3/2}(\log n)^3,\\
\end{aligned}
$$

Similarly,
$$
\begin{aligned}
&\left\vert\frac{1}{6} \sum_{l,k,v=1}^{n}  \partial_v \partial_k \partial_l F_\beta\left(\boldsymbol{z}^w_{d+1, 0}+\iota\left(\boldsymbol{z}^w_{d+1}-\boldsymbol{z}^w_{d, 0}\right)\right)\left(\boldsymbol z^w_{d+1, l}-\boldsymbol z^w_{d, 0, l}\right)\left(\boldsymbol z^w_{d+1, k}-\boldsymbol z^w_{d, 0, k}\right)\left(\boldsymbol z^w_{d+1, v}-\boldsymbol z^w_{d, 0, v}\right)\right\vert\\
&\lesssim\beta^2 n^{-3/2}(\log n)^3,
\end{aligned}
$$
hold with probability approaching one. Consequently, we obtain,
$$
\left\vert\E\left\{f_1\left(W_{d, 0}, V_{d, 0}\right)\left(W_d-W_{d, 0}\right)\right\}-\E\left\{f_1\left(W_{d, 0}, V_{d, 0}\right)\left(W_{d+1}-W_{d, 0}\right)\right\}\right\vert \lesssim \beta^2 n^{-3/2}(\log n)^3,
$$
$$
\left\vert\E\left\{f_2\left(W_{d, 0}, V_{d, 0}\right)\left(V_d-V_{d, 0}\right)\right\}-\E\left\{f_2\left(W_{d, 0}, V_{d, 0}\right)\left(V_{d+1}-V_{d, 0}\right)\right\}\right\vert \lesssim \beta^2 n^{-3/2}(\log n)^3.
$$

Similarly, it can be verified that
$$
\left\vert\E\left\{f_{11}\left(W_{d, 0}, V_{d, 0}\right)\left(W_d-W_{d, 0}\right)^2\right\}-\E\left\{f_{22}\left(W_{d, 0}, V_{d, 0}\right)\left(W_{d+1}-W_{d, 0}\right)^2\right\}\right\vert \lesssim \beta^2 n^{-3/2}(\log n)^3,
$$
$$
\left\vert\E\left\{f_{22}\left(W_{d, 0}, V_{d, 0}\right)\left(V_d-V_{d, 0}\right)^2\right\}-\E\left\{f_{22}\left(W_{d, 0}, V_{d, 0}\right)\left(V_{d+1}-V_{d, 0}\right)^2\right\}\right\vert \lesssim \beta^2 n^{-3/2}\{\log (np)\}^3,
$$
and
\begin{equation*}
    \begin{aligned}
& \left\vert\E\left\{f_{12}\left(W_{d, 0}, V_{d, 0}\right)\left(W_d-W_{d, 0}\right)\left(V_d-V_{d, 0}\right)\right\}-\E\left\{f_{12}\left(W_{d, 0}, V_{d, 0}\right)\left(W_{d+1}-W_{d, 0}\right)\left(V_{d+1}-V_{d, 0}\right)\right\}\right\vert \\
& \quad \lesssim \beta^2 n^{-3/2}(\log n)^3.
    \end{aligned}
\end{equation*}
 Combining the above facts with the condition $p\lesssim n^{\nu}$, we have
\begin{eqnarray*}
\sum\limits_{d=1}^{n}\left\vert\E\left\{f(U_d,V_d)\right\} - \E\left\{f(U_{d+1},V_{d+1})\right\}\right\vert \lesssim \beta^2 n^{-1/2}(\log n)^3  \rightarrow 0,
\end{eqnarray*}
as $(n,p)\rightarrow\infty$.

\end{proof}

\subsection{Proof of Theorem~\ref{indalter}}
 \begin{proof}
It suffices to show the conclusion holds for Gaussian sequence $\{\bxi_i\}_{i=1}^{n}$. Without loss of generality, we assume $\bmu_0 = \mathbf{0}$. Define $s_i = \ind{i>\tau}$, we rewrite
\begin{align*}
W(k)^{(G)}&=g(k)\sum_{i,j=1}^n a_{i,k}a_{j,k}(s_i\bdelta + \bxi_i)^\top(s_j\bdelta + \bxi_j)\\
&=g(k)\sum_{i,j=1}^n a_{i,k}a_{j,k}(\bxi_i^\top\bxi_j +s_i\bdelta^\top\bxi_j + s_j\bxi_i^\top\bdelta + s_is_j\bdelta^\top\bdelta).
\end{align*}

Then,
\begin{align*}
\E\left\{W(k)^{(G)}\right\}=\mu_{M,k} + g(k)\sum_{i,j=1}^n a_{i,k}a_{j,k}s_is_j\bdelta^\top\bdelta.
\end{align*}
We only need to prove $g(k)\sum_{i,j=1}^n a_{i,k}a_{j,k}s_i\bdelta^\top\bxi_j  = o_p\{\sqrt{2p^{-1}\tr(\O_n^2)}\}$.
We have
\begin{align*}
&\var\bigg[\bigg\{g(k)\sum_{i,j=1}^n a_{i,k}a_{j,k}s_i\bdelta^\top\bxi_j\bigg\}\bigg]\\
&=\{g(k)\}^2\bigg(\sum_{i_1=1}^n a_{i_1,k}s_{i_1}\bdelta\bigg)^\top\bigg\{\sum_{j_1,j_2=1}^n a_{j_1,k}a_{j_2,k}\bGam_M({|j_2-j_1|})\bigg\}\bigg(\sum_{i_2=1}^n a_{i_2,k}s_{i_2}\bdelta\bigg)\\
&=\{g(k)\}^2\bigg(\sum_{i_1=1}^n a_{i_1,k}s_{i_1}\bigg)^2\bdelta^\top\bigg\{\sum_{h\in \mathcal{M}}\sum_{j_1=1}^{n-h} a_{j_1,k}a_{j_1+h,k}\mathit{\bGam}_M({\vert h\vert})\bigg\}\bdelta\\
&=O\left(\frac{k^4(n-k)^4}{n^6 p}\frac{n^3}{k^2(n-k)^2}\bdelta^\top\O_n\bdelta\right)=O\left(np^{-1} \bdelta^\top\O_n\bdelta\right)=o\{p^{-1}\tr(\O_n^2)\},
\end{align*}
under $H_{1;n,p}$.

Thus, according to the proof of Theorem 1, under $H_{1;n,p}$, we have
\begin{align*}
&~~~~\frac{W(k)^{(G)}-\mu_{M,k}}{\sqrt{2p^{-1}\tr(\O_n^2)}}\\
&=\mathcal{W}(\tilde{\eeta}_{1,k/n},\dots,\tilde{\eeta}_{k_n,k/n}) + \frac{g(k)\sum_{i,j=1}^n a_{i,k}a_{j,k}s_is_j\bdelta^\top\bdelta}{\sqrt{2p^{-1}\tr(\O_n^2)}} + o_p(1)\\
&:=\mathcal{W}_{\mathcal{A}}(\tilde{\eeta}_{1,k/n},\dots,\tilde{\eeta}_{k_n,k/n}) + \mathcal{W}_{\mathcal{A}^c}(\tilde{\eeta}_{1,k/n},\dots,\tilde{\eeta}_{k_n,k/n}) + \frac{g(k)\sum_{i,j=1}^n a_{i,k}a_{j,k}s_is_j\bdelta^\top\bdelta}{\sqrt{2p^{-1}\tr(\O_n^2)}} + o_p(1),
\end{align*}
where
\begin{align*}
\mathcal{W}_{\mathcal{A}}(\tilde{\eeta}_{1,k/n},\dots,\tilde{\eeta}_{k_n,k/n}) &:= \frac{2g(k)}{\sqrt{2p^{-1}\tr(\O_n^2)}}\sum\limits_{1\leq i < j\leq k_n}\sum\limits_{\ell_1=(i-1)w_n+1}^{iw_n-M}\sum\limits_{\ell_2=(j-1)w_n+1}^{jw_n-M}\sum\limits_{t\in\mathcal{A}}\eta_{\ell_1,k/n,t} \eta_{\ell_2,k/n,t},\\
\mathcal{W}_{\mathcal{A}^c}(\tilde{\eeta}_{1,k/n},\dots,\tilde{\eeta}_{k_n.k/n}) &:= \frac{2g(k)}{\sqrt{2p^{-1}\tr(\O_n^2)}}\sum\limits_{1\leq i < j\leq k_n}\sum\limits_{\ell_1=(i-1)w_n+1}^{iw_n-M}\sum\limits_{\ell_2=(j-1)w_n+1}^{jw_n-M}\sum\limits_{t\in\mathcal{A}^c}\eta_{\ell_1,k/n,t} \eta_{\ell_2,k/n,t}.
\end{align*}
According to the same arguments as in Theorem 5 in \cite{wang2023JRSSB}, we can rewrite
\begin{align*}
M_{n,p} = \max\left\{\max_{k=1,\ldots,n-1}\max_{t\in\mathcal{A}}|C_{0, t}(k)|, \max_{k=1,\ldots,n-1}\max_{t\in\mathcal{A}^c}|C_{0, t}(k)|\right\}.
\end{align*}
Define $\bxi_{i\mathcal{A}} := (\xi_{it}, t \in \mathcal{A})^\top$, $\bxi_{i\mathcal{A}^c} := (\xi_{it}, t \in \mathcal{A}^c)^\top$,
$\bms_{\mathcal{A}, \mathcal{A}}:=\cov(\bxi_{i\mathcal{A}}, \bxi_{i\mathcal{A}})$, $\bms_{\mathcal{A}^c, \mathcal{A}^c}:=\cov(\bxi_{i\mathcal{A}^c}, \bxi_{i\mathcal{A}^c})$, and $\bms_{\mathcal{A}, \mathcal{A}^c}:=\cov(\bxi_{i\mathcal{A}}, \bxi_{i\mathcal{A}^c})$. We decompose $\bxi_{i\mathcal{A}^c} =\boldsymbol U_{i\mathcal{A}^c}+ \boldsymbol V_{i\mathcal{A}^c}$, where $\boldsymbol U_{i\mathcal{A}^c} := \bxi_{i\mathcal{A}^c} - \bms_{\mathcal{A}^c, \mathcal{A}}\bms_{\mathcal{A}, \mathcal{A}}^{-1}\bxi_{i\mathcal{A}}$, and
$\boldsymbol V_{i\mathcal{A}^c} := \bms_{\mathcal{A}^c, \mathcal{A}}\bms_{\mathcal{A}, \mathcal{A}}^{-1}\bxi_{i\mathcal{A}}$, which satisfy
$\boldsymbol U_{i\mathcal{A}^c} \sim N(\mathbf{0}, \bms_{\mathcal{A}^c, \mathcal{A}^c}-\bms_{\mathcal{A}^c, \mathcal{A}}\bms_{\mathcal{A}, \mathcal{A}}^{-1}\bms_{\mathcal{A}, \mathcal{A}^c})$,
$\boldsymbol V_{i\mathcal{A}^c} \sim N(\mathbf{0}, \bms_{\mathcal{A}^c, \mathcal{A}}\bms_{\mathcal{A}, \mathcal{A}}^{-1}\bms_{\mathcal{A}, \mathcal{A}^c})$, $\boldsymbol U_{i\mathcal{A}^c}$ and $\bxi_{i\mathcal{A}}$ are independent.
{Based on Theorem \ref{indnull}, we have known $\max_{k=1,\ldots,n-1}\mathcal{W}_{\mathcal{A}^c}(\tilde{\eeta}_{1,k/n},\dots,\tilde{\eeta}_{k_n,k/n})$ is asymptotically independent of $\max\limits_{k=1,\ldots,n-1}\max\limits_{t\in\mathcal{A}^c}|C_{0, t}(k)|$.} Hence, it suffices to prove $\max\limits_{k=1,\ldots,n-1}\mathcal{W}_{\mathcal{A}^c}(\tilde{\eeta}_{1,k/n},\dots,\tilde{\eeta}_{k_n,k/n})$ is asymptotically independent of $\bxi_{i\mathcal{A}}$.
We rewrite
\begin{align*}
&\mathcal{W}_{\mathcal{A}^c}(\tilde{\eeta}_{1,k/n},\dots,\tilde{\eeta}_{k_n,k/n})\\
=& \frac{2g(k)}{\sqrt{2p^{-1}\tr(\O_n^2)}}\sum\limits_{1\leq i < j\leq k_n}\sum\limits_{\ell_1=(i-1)w_n+1}^{iw_n-M}\sum\limits_{\ell_2=(j-1)w_n+1}^{jw_n-M}a_{\ell_1,k}a_{\ell_2,k}(\boldsymbol U_{\ell_1\mathcal{A}^c}^\top \boldsymbol U_{\ell_2\mathcal{A}^c} + 2\boldsymbol U_{\ell_1\mathcal{A}^c}^\top \boldsymbol V_{\ell_2\mathcal{A}^c} + \boldsymbol V_{\ell_1\mathcal{A}^c}^\top \boldsymbol V_{\ell_2\mathcal{A}^c}).
\end{align*}
With similar discussion on Equation~\eqref{Th5-eq1}, we have
\begin{align*}
&\pr\bigg(\bigg|\max_{k=1,\dots, n-1}\frac{2g(k)}{\sqrt{2p^{-1}\tr(\O_n^2)}}\sum\limits_{1\leq i < j\leq k_n}\sum\limits_{\ell_1=(i-1)w_n+1}^{iw_n-M}\sum\limits_{\ell_2=(j-1)w_n+1}^{jw_n-M}a_{\ell_1,k}a_{\ell_2,k}\boldsymbol U_{i\mathcal{A}^c}^\top \boldsymbol V_{i\mathcal{A}^c} \bigg|\geq \epsilon\bigg)\\
\leq& n\exp\bigg\{-C_\epsilon \sqrt{\frac{p}{|\mathcal{A}|}}\bigg\} \cdot \log n \rightarrow  0,\\
&\pr\bigg(\bigg|\max_{k=1,\dots, n-1}\frac{2g(k)}{\sqrt{2p^{-1}\tr(\O_n^2)}}\sum\limits_{1\leq i < j\leq k_n}\sum\limits_{\ell_1=(i-1)w_n+1}^{iw_n-M}\sum\limits_{\ell_2=(j-1)w_n+1}^{jw_n-M}a_{\ell_1,k}a_{\ell_2,k}\boldsymbol V_{i\mathcal{A}^c}^\top \boldsymbol V_{i\mathcal{A}^c} \bigg|\geq \epsilon\bigg)\\
\leq& n\exp\bigg\{-C_\epsilon \sqrt{\frac{p}{|\mathcal{A}|}}\bigg\} \cdot \log n \rightarrow  0,
\end{align*}
due to $|\mathcal{A}| = o\{p/(\log n) ^2\}$. Thus, we have
\begin{align*}
&\mathcal{W}_{\mathcal{A}^c}(\tilde{\eta}_{1,k/n},\dots,\tilde{\eta}_{k_n,k/n})\\
=& \frac{2g(k)}{\sqrt{2p^{-1}\tr(\O_n^2)}}\sum\limits_{1\leq i < j\leq k_n}\sum\limits_{\ell_1=(i-1)w_n+1}^{iw_n-M}\sum\limits_{\ell_2=(j-1)w_n+1}^{jw_n-M}a_{\ell_1,k}a_{\ell_2,k}\boldsymbol U_{\ell_1\mathcal{A}^c}^\top \boldsymbol U_{\ell_2\mathcal{A}^c} +o_p(1).
\end{align*}
It implies that $\max_{k=1,\ldots,n-1}\mathcal{W}_{\mathcal{A}^c}(\tilde{\eeta}_{1,k/n},\ldots,\tilde{\eeta}_{k_n,k/n})$ is asymptotically independent of $\bxi_{i\mathcal{A}}$ with Lemma \ref{Appendix-L4}.

\end{proof}

\subsection{Proof of Theorem \ref{consistency}}
\subsubsection{Consistency of $\hat{\tau}_{S}$}
To prove the consistency of $\hat{\tau}_S$, it suffices to show that, for any $\epsilon > 0$,
\begin{align*}
n\pr(|\hat{\tau}_S-\tau| > n\epsilon) \rightarrow 0.
\end{align*}
Without loss of generality, we assume $\hat{\tau}_S < \tau$ and define
\begin{align*}
   \widetilde{C}_{j}(k) = \frac{1}{\sqrt{n}}\left(S_{kj}-\frac{k}{n}S_{nj}\right),
\end{align*}
then
\begin{align*}
\sqrt{p}\left\{W(k)-\mu_{k,M}\right\}
=&\sum\limits_{j=1}^p \bigg\{\widetilde{C}^2_{ j}(k) - \frac{1}{n}\sum\limits_{h=0}^{M}\sum\limits_{i=1}^{n-h}\{2-\ind{h=0}\}\breve{a}_{i,k}\breve{a}_{i+h,k}\boldsymbol e_j^\top\bGam(h)\boldsymbol e_j\bigg\}\\
:=&\sum\limits_{j=1}^p \widetilde{C}^2_{ j}(k) - \sum\limits_{j=1}^p H_j(k),
\end{align*}
where
$H_j(k) = n^{-1}\sum\nolimits_{h=0}^{M}\sum\nolimits_{i=1}^{n-h}\{2-\ind{h=0}\}\breve{a}_{i,k}\breve{a}_{i+h,k}\boldsymbol e_j^\top\mathit{\bGam}(h)\boldsymbol e_j$ and $\boldsymbol e_j$ represents the $p$-dimensional unit vector with 1 in the $j$th entry and 0 elsewhere.
We will show that, for all $t$ such that $t < \tau-n\epsilon$,
\begin{align*}
   \pr\bigg(\sum\limits_{j=1}^p \widetilde{C}^2_{j}(t) - \sum\limits_{j=1}^pH_j(t) > \sum\limits_{j=1}^p \widetilde{C}^2_{j}(\tau) - \sum\limits_{j=1}^p H_j(\tau)\bigg)\rightarrow 0.
\end{align*}
Observe that
\begin{align*}
&\sum\limits_{j=1}^p \widetilde{C}^2_{j}(t) - \sum\limits_{j=1}^pH_j(t) - \sum\limits_{j=1}^p \widetilde{C}^2_{j}(\tau) + \sum\limits_{j=1}^p H_j(\tau)\\
=&\sum\limits_{j=1}^p \frac{1}{n}\left(S_{tj}-\frac{t}{n}S_{nj}\right)^2 - \sum\limits_{j=1}^p \frac{1}{n}\left(S_{\tau j}-\frac{\tau}{n}S_{nj}\right)^2 - \sum\limits_{j=1}^pH_j(t) + \sum\limits_{j=1}^p H_j(\tau)\\
:=&G_1+G_2.
\end{align*}
Let $\widetilde{S}_{kj} = \sum\nolimits_{i=1}^k \epsilon_{ij}$, then
\begin{align*}
G_1=&\sum\limits_{j=1}^p \frac{1}{n}\left(S_{tj}-\frac{t}{n}S_{nj}\right)^2 - \sum\limits_{j=1}^p \frac{1}{n}\left(S_{\tau j}-\frac{\tau}{n}S_{nj}\right)^2 \\
=&\sum\limits_{j=1}^p \bigg\{\frac{1}{n}\left(\widetilde{S}_{tj}-\frac{t}{n}\widetilde{S}_{nj}-\frac{t(n-\tau)}{n}\delta_j\right)^2 - \frac{1}{n}\left(\widetilde{S}_{\tau j}-\frac{\tau}{n}\widetilde{S}_{nj}-\frac{\tau(n-\tau)}{n}\delta_j\right)^2\bigg\}\\
=&\sum\limits_{j=1}^p \frac{(t^2-\tau^2)(n-\tau)^2}{n^3}\delta_j^2 - \sum\limits_{j=1}^p \bigg\{\frac{2t(n-\tau)}{n^2}\delta_j\left(\widetilde{S}_{tj}-\frac{t}{n}\widetilde{S}_{nj}\right) - \frac{2\tau(n-\tau)}{n^2}\delta_j\left(\widetilde{S}_{\tau j}-\frac{\tau}{n}\widetilde{S}_{nj}\right)\bigg\}\\
&+ \sum\limits_{j=1}^p \bigg\{\frac{1}{n}\left(\widetilde{S}_{tj}-\frac{t}{n}\widetilde{S}_{nj}\right)^2 - \frac{1}{n}\left(\widetilde{S}_{\tau j}-\frac{\tau}{n}\widetilde{S}_{nj}\right)^2\bigg\}\\
:=&G_{11} + G_{12} + G_{13}.
\end{align*}
Since $t < \tau-n\epsilon$, we have
\begin{align*}
G_{11} = \sum\limits_{j=1}^p \frac{(t^2-\tau^2)(n-\tau)^2}{n^3}\delta_j^2 <  -\frac{t+\tau}{n}n\|\bdelta\|^2\epsilon.
\end{align*}

It remains to prove, for large enough constant $C>0$,
\begin{equation}
    n\pr(G_{12} > Cn\|\bdelta\|^2\epsilon) \rightarrow 0,\label{Eq:G12}
\end{equation}
\begin{equation}
    n\pr(G_{13}+G_2 > Cn\|\bdelta\|^2\epsilon) \rightarrow 0.\label{Eq:G13+G2}
\end{equation}

Observe that
\begin{align*}
G_{12} = &\sum\limits_{j=1}^p \bigg\{\frac{2t(n-\tau)}{n^2}\delta_j\left(\widetilde{S}_{tj}-\frac{t}{n}\widetilde{S}_{nj}\right) - \frac{2\tau(n-\tau)}{n^2}\delta_j\left(\widetilde{S}_{\tau j}-\frac{\tau}{n}\widetilde{S}_{nj}\right)\bigg\}\\
=&\sum\limits_{j=1}^p \frac{2(n-\tau)}{n^2}\delta_j\bigg\{(\tau-t)\sum\limits_{i=1}^t\epsilon_{ij}+\tau\sum\limits_{i=t+1}^{\tau}\epsilon_{ij}-\frac{\tau^2-t^2}{n}\sum\limits_{i=1}^n\epsilon_{ij}\bigg\}\\
:=&G_{121} + G_{122} + G_{123}.
\end{align*}

We first consider the term $G_{123}$,
\begin{align*}
n \pr(G_{123}>Cn\|\bdelta\|^2\epsilon) =& n \pr\bigg( \frac{2(n-\tau)(\tau^2-t^2)}{n^3}\sum\limits_{i=1}^n\sum\limits_{j=1}^p\delta_j\epsilon_{ij}>Cn\|\bdelta\|^2\epsilon\bigg)\\
=& n \pr\bigg( \frac{2(n-\tau)(\tau^2-t^2)}{n^3} \bdelta^\top \bms^{1/2}\sum_{i=1-M}^n\omega_i\bZ_i>Cn\|\bdelta\|^2\epsilon\bigg)\\
\leq& 2n\frac{n\{6+\E (Z_{11}^4)\}\{\lambda_{\max}(\bms)\}^2\|\bdelta\|^4(\sum_{\ell=0}^\infty\vert b_\ell\vert)^4}{C^4 n^4\|\bdelta\|^8\epsilon^4}\left\{\frac{2(n-\tau)(\tau^2-t^2)}{n^3} \right\}^4\\
\lesssim&\frac{1}{n^2\|\bdelta\|^4\epsilon^4}\rightarrow 0,
\end{align*}
where the first inequality holds by Markov inequality,
\begin{equation*}
    \begin{aligned}
        \sum\limits_{i=1}^n\sum\limits_{j=1}^p\delta_j\epsilon_{ij}=&\bdelta^\top \bms^{1/2}\left\{\sum_{s=1-M}^0 (\sum_{\ell=1-s}^M b_{\ell})\bZ_s+\sum_{s=1}^{n-M} (\sum_{\ell=0}^M b_{\ell})\bZ_s+\sum_{s=n-M+1}^n (\sum_{\ell=0}^{n-s} b_{\ell})\bZ_s\right\}\\
        :=&\bdelta^\top \bms^{1/2}\sum_{s=1-M}^n\omega_s\bZ_s,
    \end{aligned}
\end{equation*}
and
\begin{equation*}
    \begin{aligned}
        \E \left\{(\bdelta^\top \bms^{1/2}\bZ_i)^4\right\}=&\sum_{j=1}^p (\bdelta^\top \bms^{1/2})_{j}^4 \E (Z_{ij}^4)+6\sum_{1\leq j\neq \ell\leq p}(\bdelta^\top \bms^{1/2})_{j}^2(\bdelta^\top \bms^{1/2})_{\ell}^2 \E (Z_{ij}^2)\E (Z_{i\ell}^2)\\
        \leq &\E (Z_{i1}^4)\left\{\sum_{j=1}^p(\bdelta^\top \bms^{1/2})_{j}^2\right\}^2+6\left\{\sum_{j=1}^p(\bdelta^\top \bms^{1/2})_{j}^2\right\}^2,
    \end{aligned}
\end{equation*}
where $\vert\omega_s\vert\leq \sum_{\ell=0}^\infty\vert b_\ell\vert<\infty$ by Assumption~\ref{ass:C1} for all $s=1-M,\ldots,n$ and $(\cdot)_j$ represents the $j$th entry of a vector. Similarly, we can show that $n \pr(G_{121}>Cn\|\bdelta\|^2\epsilon)\rightarrow 0 $ and $n \pr(G_{122}>Cn\|\bdelta\|^2\epsilon)\rightarrow 0 $.

We next consider the term $G_{13}+G_2$. Notice that, by Theorem~\ref{Th2}, we have, $\E(G_{13})-G_2\rightarrow 0$. Taking the same procedure as \cite{10.1142/s201032631950014x} and replace the Lemma A.1 by Markov's inequality, we have
\begin{equation*}
    \begin{aligned}
        n\pr(G_{13}+G_2 > Cn\|\bdelta\|^2\epsilon)\lesssim  n\frac{n\{\lambda_{\max}(\bms)\}^2(\sum_{\ell=0}^\infty\vert b_\ell\vert)^4}{n^4\|\bdelta\|^4\epsilon^2}+o(1)\rightarrow 0.
    \end{aligned}
\end{equation*}
Consequently, the Equation \eqref{Eq:G12} and \eqref{Eq:G13+G2} hold, which complete the proof.

\subsubsection{Consistency of $\hat{\tau}$ or $\hat{\tau}^\dagger$}
According to Theorem 6 in \cite{wang2023JRSSB}, both $\hat{\tau}_M$ and $\hat{\tau}_M^\dagger$ are consistent, that is, $(\hat{\tau}_M - \tau)/n = o_p(1)$ and $(\hat{\tau}_M^\dagger - \tau)/n = o_p(1)$. We next examine the consistency of $\hat{\tau}$ and $\hat{\tau}^\dagger$ under the following three scenarios:
(i) $\|\bdelta\|_\infty < c \sqrt{\log(np)/n}$ for some constant $c > 0$, while $n \|\bdelta\|_2^2  \to \infty$;
(ii) $\|\bdelta\|_\infty / \sqrt{\log(np)/n} \to \infty$, while $n\|\bdelta\|_2^2 < c$ for some constant $c > 0$; and
(iii) both $\|\bdelta\|_\infty / \sqrt{\log(np)/n} \to \infty$ and $n \|\bdelta\|_2^2  \to \infty$.

For Case (i), Theorem~\ref{prop:Tms-alter} implies that $p_{S_{n,p}} \to 0$, while $p_{M^{\dagger}{n,p}} \not\to 0$. Hence, $\Pr(p_{T_{ms}} < p_{M^{\dagger}_{n,p}}) \to 1$. Combined with the consistency result above, it follows that $(\hat{\tau}^{\dagger} - \tau)/n = o_p(1)$. A similar argument shows that $\hat{\tau}^{\dagger}$ is also consistent under Case (ii). For Case (iii), the result follows immediately since at least one of the estimators $\hat{\tau}_{S}$, $\hat{\tau}_{M}$, or $\hat{\tau}_{M^{\dagger}}$ is consistent.

\subsection{Proof of the uniform consistency of the long-run variance estimator}

We first restate some notations and statistics in \cite{chan2022AoS}. Let $\sigma_j=\lim_{n\rightarrow\infty}n\var(\bar{X}_{ij})=\sum_{k\in\mathbb{Z}}\gamma_{k,j}$, where $\gamma_{k,j}=\cov(X_{0,j},X_{k,j})$. The $m$th order difference-based estimator of $\sigma_j$ is
\begin{equation*}
    \hat{\sigma}_{j,(m)}=\sum_{\vert k\vert<\ell_j}K(k/\ell_j)\hat{\gamma}_{k,j}^D,
\end{equation*}
where $K$ is a kernel and
\begin{equation*}
    \hat{\gamma}_{k,j}^D=\frac{1}{n}\sum_{i=mh+\vert k\vert+1}^n D_{i,j} D_{i-\vert k\vert,j},\\
\end{equation*}
and the $m$th order lag-$h$ difference statistics are
\begin{equation*}
    D_{i,j}=\sum_{s=0}^m d_s X_{i-s h,j},~~~~i=mh+1,\ldots,n,
    d_0+\ldots+d_m=0, d_0^2+\ldots,d_m^2=1.
\end{equation*}

\begin{assumption}\label{ass:chan_ass5}
    (Assumption 5 in \cite{chan2022AoS},Near-origin property) The kernel $K$ satisfies that there exists $\tilde{q}\in\mathbb{N}$ and $B\in \mathbb{R}\setminus\{0\}$ such that $\{K(t)-K(0)\}/\vert t\vert^{\tilde{q}}\rightarrow B$ as $t\rightarrow 0$.
\end{assumption}
\begin{lemma}\label{lemma:LRV}
    Suppose  the Assumption 1 and 5(Assumption~\ref{ass:chan_ass5}) in \cite{chan2022AoS} and Assumption~\ref{ass:C1}--\ref{ass:C3} hold and there exists $0<\kappa_q<3$, such that $\tilde{q}=3-\kappa_q$, $\nu\lesssim \min\{q/2-2,\tilde{q}q/(1+2\tilde{q})-1\}$ and $\|\bdelta\|_\infty=O(1)$. For $\delta_n=n^{-\tilde{\delta}}$, $\tilde{\delta}>0$ sufficiently small, it holds that
    \begin{equation*}
        \pr(\max_{1\leq h\leq p}\vert \hat{\sigma}_h-\sigma_h \vert\geq \delta_n^2)\lesssim n^{-C},~~~~ C>0,
    \end{equation*}
    as $n\rightarrow\infty$.
\end{lemma}
\begin{proof}
    We follow the proof of Lemma E.5 in \cite{10.1214/15-aos1347}. It follows that for large enough $p$ and $n$, we have,
\begin{equation}\label{eq:lemmaE.5}
    \begin{aligned}
            \pr(\max_{1\leq h\leq p}\vert \hat{\sigma}_h^2-\sigma_h^2\vert\geq \delta_n^2\inf_{h}^*\sigma_h)\leq& \sum_{h=1}^p\sum_{j=1}^{\ell_h} \pr(\vert \hat{\gamma}_{j,h}^D-\gamma_{j,h}\vert\geq \delta_n^2\inf_{h}^*\sigma_h/(2\ell_h))\\
            \leq& (\delta_n^2\inf_{h}^*\sigma_h)^{-p/2}\sum_{h=1}^p\sum_{j=1}^{\ell_h} (2\ell_h)^{q/2}\|\hat{\gamma}_{j,h}^D-\gamma_{j,h}\|_{q/2}^{q/2}.
    \end{aligned}
    \end{equation}
    Let $u_{\tilde{q}}=\sum_{k\in \mathbb{Z}}\vert k\vert^{\tilde{q}}\sum_{\ell=0}^\infty b_\ell b_{\ell+k}$. By Assumption~\ref{ass:C1} and Lemma~\ref{Appendix-L3}, we have
\begin{equation*}
    u_{\tilde{q}}\leq\sum_{k\in \mathbb{Z}}\vert k\vert^{\tilde{q}}\sum_{\ell=k}^\infty\vert b_\ell\vert\leq \sum_{k\in \mathbb{Z}} o(\vert k\vert^{\tilde{q}-4})=O(1).
\end{equation*}
We repeat the Lemma B.1 in \cite{chan2022AoS} and acquire the results $\sqrt{n}\max_{1\leq h\leq p}\|\hat{\gamma}_{j,h}^D-\gamma_{j,h}\|_{q/2}=O(1)$ for $j=o(n)$. From the Theorem 4.1 in \cite{chan2022AoS}, we see that the optimal $\ell_h$ has the same optimal rate $\ell_h\asymp n^{1/(1+2\tilde{q})}$ where $\tilde{q}=3-\kappa_q$, $\kappa_q>0$. Thus, the Equation~\eqref{eq:lemmaE.5} is bounded by $p n^{2\tilde{\delta}}n^{(q/2+1)/(1+2\tilde{q})}n^{-q/2}\lesssim n^{-C}$. Using that $\vert\hat{\sigma}_h-\sigma_h\vert\leq (\inf_h^* \sigma_h)^{-1}\vert \hat{\sigma}_h^2-\sigma_h^2\vert$, the proof is completed.
\end{proof}

\subsection{Some Useful Lemmas}

\begin{lemma}
\label{Appendix-L1}
Suppose $\{z_{tj}; t=1,\dots,n, j=1,\dots,p\}$ are independent and identically distributed random variables satisfying $\E(z_{tj})=0$, $\E(z_{tj}^2)=1$ and $\E(z_{tj}^4)=\mu_4<\infty$. Let $\w:=\sum\nolimits_{t=0}^{\infty}a_t\z_t$, where $\z_t=(z_{t1},\dots,z_{tp})^\top$ and $\sum\nolimits_{t=0}^{\infty}|a_{t}|<\infty$.
Then, there exists a positive constant $\tau_1\leq 3$ such that for all $p\times p$ positive semi-definite matrix $\B$,
\begin{equation*}
\E\{(\w^\top\B\w)^2\}\leq \tau_1 \{\E(\w^\top\B\w)\}^2.
\end{equation*}
\end{lemma}

\begin{proof}
\begin{eqnarray*}
\E(\w^\top\B\w)=\sum\limits_{j=1}^p b_{jj}\E\{\sum\limits_{t=0}^{\infty}a_t^2z_{ti}^2\}=\sum\limits_{t=0}^{\infty}a_t^2\tr(\B).
\end{eqnarray*}
\begin{eqnarray*}
\E\{(\w^\top\B\w)^2\}&=&\sum\limits_{j_1,j_2,j_3,j_4=1}^p b_{j_1j_2}b_{j_3j_4}\E\bigg\{\prod\limits_{s=1}^4\bigg(\sum\limits_{t=0}^{\infty}a_tz_{t,j_s}\bigg)\bigg\}\\
&=&\sum\limits_{j=1}^p b_{jj}^2\bigg(\sum\limits_{t=0}^{\infty}a_t^4\mu_4+3\sum\limits_{t_1\neq t_2\geq 0}a_{t_1}^2a_{t_2}^2\bigg)\\
&&+\sum\limits_{j_1\neq j_2=1}^p (2b_{j_1j_2}^2+b_{j_1j_1}b_{j_2j_2})\bigg(\sum\limits_{t=0}^{\infty}a_t^2\bigg)^2\\
&\leq&\tr^2(\B)\bigg(\sum\limits_{t=0}^{\infty}a_t^4\mu_4+3\sum\limits_{t_1\neq t_2\geq 0}a_{t_1}^2a_{t_2}^2\bigg)\\
&&+\{2\tr(\B^2)+\tr^2(\B)\}\bigg(\sum\limits_{t=0}^{\infty}a_t^2\bigg)^2\\
&\leq &(6+\mu_4)\bigg(\sum\limits_{t=0}^{\infty}a_t^2\bigg)^2\tr^2(\B).
\end{eqnarray*}
Taking $\tau_1=6+\mu_4$, the lemma has been proved.
\end{proof}

\begin{lemma}
\label{Appendix-L2}
Under Assumptions \ref{ass:C1}--\ref{ass:C3}, for any $k_1,k_2,k_3,k_4\in\{1,\dots,n\}$, there exists a constant $\tau_2$ such that
\begin{eqnarray}
\left|\E\big[\bepsilon_{k_1}^\top \bepsilon_{k_2}-\tr\{\bGam(|k_1-k_2|\}\big]\big[\bepsilon_{k_3}^\top \bepsilon_{k_4}-\tr\{\bGam(|k_3-k_4|\}\big]\right|\leq \tau_2\tr(\O_{n}^2), \label{L6.1}\\
\big|\E\big[\boldsymbol\varepsilon_{k_1}^\top\boldsymbol\varepsilon_{k_2}-\tr\{\bGam_M(|k_1-k_2|\}\big]\big[\boldsymbol\varepsilon_{k_3}^\top\boldsymbol\varepsilon_{k_4}-\tr\{\bGam_M(|k_3-k_4|\}\big]\big|\leq \tau_2\tr(\O_{n,M}^2). \label{L6.2}
\end{eqnarray}
\end{lemma}

\begin{proof}
Following H\"{o}lder inequality,
\begin{eqnarray*}
&&\left|\E\big[\bepsilon_{k_1}^\top \bepsilon_{k_2}-\tr\{\bGam(|k_1-k_2|\}\big]\big[\bepsilon_{k_3}^\top \bepsilon_{k_4}-\tr\{\bGam(|k_3-k_4|\}\big]\right|\\
&\leq &\left(\E\big[\bepsilon_{k_1}^\top \bepsilon_{k_2}-\tr\{\bGam(|k_1-k_2|\})\big]^2\E\big[\bepsilon_{k_3}^\top \bepsilon_{k_4}-\tr\{\bGam(|k_3-k_4|\}\big]^2\right)^{1/2}.
\end{eqnarray*}
It suffice to prove that, for all $h>0$,
\begin{eqnarray*}
\E\big[\bepsilon_{0}^\top \bepsilon_{h}-\tr\{\bGam(h)\}\big]^2\leq \tau_2\tr(\O_{n}^2).
\end{eqnarray*}
Recall $\bepsilon_h=\bms^{1/2}\sum\nolimits_{k=0}^{\infty}b_k\bZ_{h-k}$, we split $\bepsilon_h$ into two independent parts as
\begin{eqnarray*}
\bepsilon_h=\bms^{1/2}\sum\limits_{k=0}^{\infty}b_{k+h}\bZ_{-k}+\bms^{1/2}\sum\limits_{k=0}^{h-1}b_k\bZ_{h-k}:=\bepsilon_{h,1}+\bepsilon_{h,2}.
\end{eqnarray*}
Thus,
\begin{eqnarray*}
&&\E\Big[\big\{\bepsilon_{0}^\top\bepsilon_{h,1}-\tr(\bGam_{h})\big\}^2\Big]\\
&=&\E\Big[\big\{\sum\limits_{k=0}^{\infty}b_k\bZ_{-k}^{\top}\bms\sum\limits_{k=0}^{\infty}b_{k+h}\bZ_{-k}-\tr(\bGam_{h})\big\}^2\Big]\\
&=&\E\Big[\big\{\sum\limits_{k=0}^{\infty}b_kb_{k+h}\bZ_{-k}^{\top}\bms \bZ_{-k}-\tr(\bGam_{h})+\sum\limits_{k_1\neq k_2}^{\infty}b_{k_1}b_{k_2+h}\bZ_{-k_1}^{\top}\bms \bZ_{-k_2}\big\}^2\Big]\\
&=&\E\Big[\big\{\sum\limits_{k=0}^{\infty}b_kb_{k+h}\bZ_{-k}^{\top}\bms \bZ_{-k}-\tr(\bGam_{h})\big\}^2\Big]+\E\Big\{\big(\sum\limits_{k_1\neq k_2}^{\infty}b_{k_1}b_{k_2+h}\bZ_{-k_1}^{\top}\bms \bZ_{-k_2}\big)^2\Big\}\\
&=&\sum\limits_{k=0}^{\infty}b_k^2b_{k+h}^2\E\Big[\big\{\bZ_{-k}^{\top}\bms \bZ_{-k}-\tr(\bms)\big\}^2\Big]+\sum\limits_{k_1\neq k_2}^{\infty}b_{k_1}^2b_{k_2+h}^2\E\Big\{\big(\bZ_{-k_1}^{\top}\bms \bZ_{-k_2}\big)^2\Big\}\\
&&+\sum\limits_{k_1\neq k_2}^{\infty}b_{k_1}b_{k_2+h}b_{k_2}b_{k_1+h}\E\Big\{\big(\bZ_{-k_1}^{\top}\bms \bZ_{-k_2}\big)\big(\bZ_{-k_2}^{\top}\bms \bZ_{-k_1}\big)\Big\}\\
&\leq& \Big\{\sum\limits_{k=0}^{\infty}b_k^2b_{k+h}^2(\mu_4+1)+\Big(\sum\limits_{k=0}^{\infty}b_k^2\Big)\Big(\sum\limits_{k=0}^{\infty}b_{k+h}^2\Big)+\Big(\sum\limits_{k=0}^{\infty}b_kb_{k+h}\Big)^2\Big\}\tr(\bms^2)\\
&\leq& (\mu_4+3)\Big(\sum\limits_{k=0}^{\infty}b_k^2\Big)^2\tr(\bms^2),
\end{eqnarray*}
and
\begin{eqnarray*}
\E\Big\{\big(\bepsilon_{0}^\top\bepsilon_{h,2}\big)^2\Big\}&=&\E\Big\{\Big(\sum\limits_{k=0}^{\infty}b_k\bZ_{-k}^\top\bms\sum\limits_{k=0}^{h-1}b_k\bZ_{h-k}\Big)^2\Big\}\\
&=&\Big(\sum\limits_{k=0}^{\infty}b_k^2\Big)\Big(\sum\limits_{k=0}^{h-1}b_k^2\Big)\tr(\bms^2).
\end{eqnarray*}
Equation (\ref{L6.1}) is hold with $\tr(\O_n^2)=\sum_{h_1,h_2\in \mathcal{N}}(1-\frac{\vert h_1\vert}{n})(1-\frac{\vert h_2\vert}{n})a_{h_1}a_{h_2}\tr(\bms^2)$.
Equation (\ref{L6.2}) is guaranteed by repeating the above prove with $b_{k,M}=b_{k}\mathbb{I}_{\{k\leq M\}}$.
\end{proof}

\begin{lemma}
\label{Appendix-L3}
$\{b_n, n\geq 1\}$ is a sequence of numbers. Suppose that there exist a constant $k>1$ such that $b_n=o(n^{-k})$, then
\begin{eqnarray*}
\sum\limits_{m=n}^{\infty}b_m=o(n^{-(k-1)}).
\end{eqnarray*}
\end{lemma}

\begin{proof}
There exist $N\in\mathbb{N}$, such that for all $n\geq N, n^k b_n\leq 1$. Then, for all $n\geq N$,
\begin{eqnarray*}
n^{k-1}\sum\limits_{m=n}^{\infty}b_m=\sum\limits_{m=1}^{\infty}\frac{1}{n}\sum\limits_{t=mn}^{(m+1)n-1}n^k b_t=o(n^{-(k-1)}).
\end{eqnarray*}

\end{proof}

\begin{lemma}\label{LemmaE.1central}
    Let $\X_1^{(0)},\ldots,\X_n^{(0)}$ be independent centered random vectors in $\mathbb R^p$ with $p\geq 2$, and denote $\X_i^{(0)}=(X_{i1}^{(0)},\ldots,X_{ip}^{(0)})^{\top}$. Define $Z=\max_{1\leq j\leq p}\left\vert \sum_{i=1}^n X_{ij}^{(0)} \right\vert$, $M=\max_{1\leq i\leq n}\max_{1\leq j\leq p}\left\vert X_{ij}^{(0)}\right\vert$ and $\sigma^2=\max_{1\leq j\leq p}\sum_{i=1}^n \E \{(X_{ij}^{(0)})^2\}$. Then,
    \begin{equation*}
        \E (Z)\lesssim K\left\{ \sigma\sqrt{\log p}+\sqrt{\E (M^2)}\log p \right\}.
    \end{equation*}
\end{lemma}
\begin{proof}
    See Lemma E.1 in \cite{chernozhukov2017central}.
\end{proof}

\begin{lemma}
\label{Appendix-L4}
\label{S}
Let $\{(U, U_p, \tilde{U}_p) \in \mathbb{R}^3; p \geq 1 \}$ and $\{(V, V_p, \tilde{V}_p) \in \mathbb{R}^3; p \geq 1 \}$ be two sequences of random variables with
$U_p \rightarrow U$ and $V_p \rightarrow V$ in distribution as $p\rightarrow \infty$. Assume $U$ and $V$ are continuous random variables and that
\begin{eqnarray*}
\tilde{U}_p = U_p + o_p (1) ~  {\textrm{and}}  ~ \tilde{V}_p = V_p + o_p (1).
\end{eqnarray*}
If $U_p$ and $V_p$ are asymptotically independent, then $\tilde{U}_p$ and $\tilde{V}_p$ are also asymptotically independent.
\end{lemma}

\begin{proof}
See Lemma 7.10 in \cite{feng}.

\end{proof}

 \bibliographystyle{asa}

\bibliography{ref}

\end{document}